\documentclass[11pt]{amsart}

\usepackage{graphicx}
\usepackage{framed}
\usepackage{ulem}
\usepackage{pgf,tikz,pgfplots}
\usepackage{float}
\usepackage{mathrsfs}
\usepackage{mathtools}
\usepackage{mathpple}
\usepackage[all,cmtip]{xy}
\usepackage{amsmath}
\usepackage{tikz-cd}
\usepackage{mathdots}
\usepackage{yhmath}
\usepackage{cancel}
\usepackage{color}
\definecolor{deepgreen}{rgb}{0.0,0.35,0.0}
\usepackage{siunitx}
\usepackage{array}
\usepackage{multirow}
\usepackage{amssymb}
\usepackage{gensymb}
\usepackage{tabularx}
\usepackage{booktabs}
\usepackage{makecell}
\usetikzlibrary{fadings}
\usetikzlibrary{patterns}
\usetikzlibrary{shapes}
\usetikzlibrary{arrows.meta,calc,decorations.markings}

\usetikzlibrary{matrix,arrows,decorations.pathmorphing}
\usepackage{enumerate}
\usepackage{enumitem}
\usepackage{extarrows}
\usepackage{appendix}

\usepackage{circledsteps}
\usepackage[bookmarks=true, bookmarksopen=true]{hyperref}

\DeclareMathOperator{\Id}{Id}

\DeclareMathOperator{\cont}{cont}

\theoremstyle{plain}
\newtheorem{thm}{Theorem}[section]
\newtheorem{thm-defn}[thm]{Theorem/Definition}
\newtheorem{lem}[thm]{Lemma}
\newtheorem{lem-defn}[thm]{Lemma/Definition}
\newtheorem{prop}[thm]{Proposition}

\newtheorem{defn}[thm]{Definition}

\newtheorem{rmk}[thm]{Remark}

\hypersetup{colorlinks=true,linkcolor=blue,citecolor=blue,urlcolor=blue}

\newcommand{\R}{\mathbb R}
\newcommand{\C}{\mathbb C}
\newcommand{\Z}{\mathbb Z}
\newcommand{\dd}{\mathrm d}
\newcommand{\re}{\operatorname{Re}}
\newcommand{\im}{\operatorname{Im}}

\newcommand{\Crit}{\operatorname{Crit}}

\newcommand{\J}{\mathcal J}
\newcommand{\Kk}{\mathcal K}

\title[Lefschetz Thimbles and Resurgence]{Picard-Lefschetz theory and alien calculus: a case study}

\author[Si Li]{Si Li}
\author[Yong Li]{Yong Li}
\author[Xinxing Tang]{Xinxing Tang}

\thanks{S.~Li: Yau Mathematical Sciences Center, Tsinghua University, Beijing, China.
Email: \texttt{sili@mail.tsinghua.edu.cn}.}

\thanks{Y.~Li: Beijing Institute of Mathematical Sciences and Applications,
Beijing, China.
Email: \texttt{liyong@bimsa.cn}.}

\thanks{X.~Tang: Beijing Institute of Mathematical Sciences and Applications,
Beijing, China.
Email: \texttt{tangxinxing@bimsa.cn}.}

\date{}

\keywords{exponential integrals, Lefschetz thimbles, Picard--Lefschetz theory,
Stokes phenomenon, resurgence, alien calculus, Airy integral, Bessel model,
Gamma function}

\begin{document}

\begin{abstract}
We compare Picard--Lefschetz theory and resurgence in three basic
one-dimensional exponential integrals: the Airy model, the Bessel
model, and the Gamma model.  On the Picard--Lefschetz side,
we describe the Lefschetz thimbles and compute the  connecting
trajectories between critical points appearing at Stokes phases.  On the resurgent side, we analyze
the Borel singularities of the saddle expansions and use alien operators to
recover the same Stokes coefficients.  These examples serve as explicit
finite-dimensional test cases for the dictionary between thimble
wall-crossing and alien calculus.
\end{abstract}
\maketitle
\tableofcontents

\section{Introduction}
Complex exponential integrals form a central object of study in asymptotic analysis, complex geometry, and mathematical physics. A prototypical example is an integral of the form
\begin{equation*}
I_\Gamma(\hbar)=\int_\Gamma e^{-S(z)/\hbar}\,\mu
\end{equation*}
where $S(z)$ is a holomorphic function on a complex manifold, $\mu$ is a holomorphic form, and $\Gamma$ is a contour (or cycle) in a suitable relative homology class $[\Gamma]$. We take $\hbar$ to be valued in $\C^\times$.  Such integrals naturally arise in semiclassical analysis, in the study of differential equations via Laplace-type representations, and in complexified path integrals in quantum field theory.

The asymptotic behavior of $I_\Gamma(\hbar)$ for small $\hbar$ is controlled by the critical points of the phase function
$$
\dd S=0.
$$
Near non-degenerate critical points $\{p_\sigma\}$, one obtains a formal asymptotic expansion of the form
$$
I_\Gamma(\hbar)\sim \sum_{\sigma}e^{-S(p_\sigma)/\hbar}\hbar^{n/2}\left(c_{\sigma,0}+c_{\sigma,1}\hbar+c_{\sigma,2}\hbar^2+\cdots\right)
$$
where the coefficients are determined by the local Taylor expansion of $S$ (see for example \cite{AGV2,Pham1983}). However, this expansion is typically divergent and only asymptotic in nature, reflecting the global complexity of the integral.

A powerful framework to capture the full analytic structure is provided by Picard–Lefschetz theory, which describes how integration cycles can be decomposed into a basis of Lefschetz thimbles. Each thimble $\J_\sigma$ is obtained by downward gradient flow of $\re(S/\hbar)$ associated to a critical point $p_\sigma$. If we decompose $\Gamma$ into Lefschetz thimbles in a suitable relative homology class
$$
[\Gamma]=\sum_\sigma n_\sigma [\J_\sigma],
$$
then one obtains a representation of the integral by
$$
I_\Gamma(\hbar)=\sum_{\sigma}n_\sigma\int_{\J_\sigma} e^{-S(z)/\hbar}\,\mu.
$$
We collect basics of this construction relevant to our discussion in Section \ref{sec:general}.

A crucial feature of this decomposition is that it depends on the phase of the parameter $\hbar$. As $\arg \hbar$ varies, the gradient flow structure changes, and the thimble decomposition undergoes discontinuous jumps across certain rays, known as Stokes lines. These jumps of $\{n_\sigma\}$ are governed by the Picard–Lefschetz formula: the discontinuous change in asymptotic expansions corresponds to a change in the basis of integration cycles. Similar phenomenon appears if we vary parameters in defining the phase function $S$.

On the other hand, the theory of resurgence developed by \'Ecalle (\cite{Ecalle81a},\cite{Ecalle81b},\cite{Ecalle85}) provides a general framework to study global analytic natures of divergent asymptotic expansions.  The key idea is that the Borel transform of a divergent formal series often defines a function with controlled singularities in the complex Borel plane, and these singularities encode the global analytic structure of the original problem. The process of Borel–Laplace resummation provides a systematic way to reconstruct actual functions from their asymptotic expansions, at least within suitable sectors of the complex plane. More concretely, given a formal power series with $\re \alpha>0$,
$$
\tilde\phi(\hbar)=\sum_{n\geq 0} a_n \hbar^{n+\alpha}
$$
its Borel transform is defined by 
$$
\hat \phi(\xi)=\sum_{n=0}^\infty \frac{a_n}{\Gamma(n+\alpha)}\xi^{n+\alpha-1}
$$
which may be holomorphic on the local logarithmic Riemann surface at $\xi=0$ and admit analytic continuation beyond it. If the resulting function has sufficiently mild growth, one can define the Laplace transform
$$
\mathcal L_\theta\hat\phi(\hbar)=\int_0^{e^{i\theta}\infty}e^{-\xi/\hbar} \hat\phi(\xi)\dd\xi
$$
yielding a resummed function depending on a direction $\theta$. 

However, when singularities of the Borel transform lie on certain directions,
one has to pass around them from one side or the other, leading to possible
ambiguities in the resummation.
These ambiguities are precisely governed by the Stokes phenomenon. A central insight of resurgence theory is that the discontinuities across Stokes lines are not arbitrary but are organized by a rich algebraic structure encoded in alien operators. These alien operators measure the singular behavior of the Borel transform near its singular points. See Section \ref{sec:resurgence} for a review. 
$$
\text{Borel-Laplace transform}\Longrightarrow \text{Analytic continuation}\Longrightarrow \text{Alien operators} 
$$

The key geometric objects linking the above two viewpoints are vanishing cycles. Consider a Lefschetz thimble $\J$ from a non-degenerate critical point $p$ obtained by the gradient flow. Assume 
$$
\hbar=|\hbar|e^{i\theta}.
$$
Then the values of $e^{-i\theta}S$ on $\J$ form a segment 
$$
e^{-i\theta}(S(p)+\xi), \qquad \xi\in e^{i\theta}[0,\infty)
$$
\begin{figure}[!htbp]
\centering
\begin{tikzpicture}[scale=0.85,>=Latex]

\tikzset{
  thimblebd/.style={red!70!black,line width=1.1pt},
  fib/.style={green!70!black,line width=1.0pt},
  proj/.style={yellow!70!black,dashed,line width=0.9pt},
  ray/.style={red!70!black,line width=1.1pt},
  maparrow/.style={->,black,line width=1.0pt},
  pt/.style={circle,fill=green!70!black,inner sep=1.7pt}
}
\node at (4,3.5) {\(\mathcal J\)};
\draw[thimblebd]
  (-3.0,2.2)
  .. controls (-2.0,4.1) and (0.8,4.2) ..
  (3.2,4.3);
\draw[thimblebd]
  (-3.0,2.2)
  .. controls (-1.8,1.0) and (0.2,1.7) ..
  (3.2,2.3);
\node[pt,label={[below left]{$p$}}] (p) at (-3.0,2.2) {};
\node[pt,inner sep=1.2pt] (q) at (0.05,1.63) {};
\draw[fib]
  (0.05,1.63)
  .. controls (-0.5,1.8) and (-0.5,3.45) ..
  (0.05,4.08);
\draw[fib,dashed]
  (0.05,4.08)
  .. controls (0.68,3.42) and (0.75,2.72) ..
  (0.05,1.63);
\node[green!70!black] at (0.92,3.2) {\(\sum_\xi\)};
\draw[proj] (p) -- (-3.0,-0.95);
\draw[proj] (q) -- (0.05,-0.25);
\draw[maparrow] (-0.55,1.45) -- (-0.55,-0.25);
\node[right,black] at (-0.35,0.8) {\(S(z)-S(p)\)};
\draw[ray,->] (-2.9,-0.95) -- (3.3,0.45);
\node[below right,black] at (2.45,0.35) {\(e^{i\theta}\mathbb R_{\ge0}\ \subset\ \xi\) -plane};
\node[pt,label={[below left]{$0$}}] at (-2.9,-0.95) {};
\node[pt] at (0.05,-0.29) {};
\node[green!70!black, below right] at (0.05,-0.22) {\(\xi\)};
\end{tikzpicture}
\caption{}
\end{figure}

The fibers
$$
\Sigma_\xi= \J\cap \{S-S(p)=\xi\}
$$
form a family of cycles that shrink to zero size at the critical points, representing the vanishing cycles. The integration over $\J$ can be performed by first integrating over $\gamma_\xi$, and then over $\xi$, resulting in 
$$
\int_{\J}e^{-S/\hbar}\ \mu=e^{-S(p)/\hbar}\int_0^{e^{i\theta}\infty}e^{-\xi/\hbar} \int_{\Sigma_\xi} \frac{\mu}{\dd S}.
$$
Here $\beta=\frac{\mu}{dS}$ is the Gelfand–Leray form of $\mu$ characterized by
$$
\dd S\wedge \beta= \mu.
$$
This is possible since $dS\neq 0$ away from critical points. Thus the Borel transform is precisely related to the integration over vanishing cycles.

Comparing the above two viewpoints, we can build up a precise correspondence between Picard-Lefschetz theory and alien operators, both of which describe the equivalent Stokes phenomenon. This is summarized in the following table (see Section \ref{sec:general} for detailed explanations). 

\medskip
\begin{table}[htbp]
  \centering
 \caption{Picard-Lefschetz v.s. Resurgence} \label{table:correspondence}
\renewcommand{\arraystretch}{1.35}
\begin{tabular}{p{0.45\textwidth}p{0.45\textwidth}}
\toprule
Picard--Lefschetz side &  Resurgent side \\
\midrule
Critical point \(p\) of \(S\) & Formal object \(\widetilde I_p(\hbar)\) \\
Critical value difference \(S(q)-S(p)\) & Borel singularity at \(\omega=S(q)-S(p)\) \\
Stokes phase \(\im(e^{-i\theta}(S(q)-S(p)))=0\) & Singular direction \(\arg \omega\) in the Borel plane \\
Gradient trajectory from \(q\) to \(p\) at the Stokes phase & Alien derivative \(\Delta^+_{S(q)-S(p)}\widetilde I_p\) producing the sector at \(q\) \\
Signed count of connecting trajectories & Stokes/alien coefficient in the pointed alien operator \\
Jump of thimble basis \(\mathbf J^{<}=\mathbf J^{>}R_+\) & Stokes automorphism \(\mathfrak S^+\widetilde{\mathbf I}=\widetilde{\mathbf I}R_+\) \\
Dual intersection pairing \(\langle [\J_p],[\Kk_q]\rangle\) & Choice of normalization of formal objects and Stokes constants \\
Broken PL chain \(p_0\to p_1\to\cdots\to p_r\) &  Iterated alien \(\dot\Delta^+_{S(p_r)-S(p_{r-1})} \cdots \dot\Delta^+_{S(p_1)-S(p_0)}\) \\
\bottomrule
\end{tabular}
\end{table}

\vskip 1em

In this paper, we explain this correspondence by a case study of three concrete examples: the Airy, Bessel, and Gamma models. We perform explicit computations from both the Picard–Lefschetz theory and the resurgence theory, and illustrate how the alien/Stokes coefficients reproduce the
Picard–Lefschetz trajectory counts. These three chosen examples represent three typical situations on connecting trajectories between critical points when crosing the Stokes line. The airy model has a unique connecting trajectory; the Bessel model has multiple connecting trajectories; and the Gamma model has  infinite many critical points and connecting trajectories. The phenomenon of broken lines appears in the example of Gamma model, and we illustrate how this is captured by the iterated alien operators on the resurgence side.

All the above models are finite dimensional. They are used to make explicit
the dictionary between Picard-Lefschetz theory and alien calculus. In the infinite dimensional case, this dictionary gives a useful strategy for studying quantum field theory, where the Picard--Lefschetz geometry is difficult to compute directly. In the forth-coming work \cite{LLT2026-2}, we will generalize such analysis 
\[
\text{Picard--Lefschetz trajectory count} 
\quad \Longleftrightarrow \quad
\text{alien coefficient}
\]
to an infinite-dimensional Picard–Lefschetz problem of Morse–Floer type on the path space. There the resurgence analysis of the heat kernel leads to nontrivial predictions on solutions of the Morse-flow. 

The relation among Stokes phenomena, exponential integrals, and wall-crossing structures has been explored from several complementary perspectives. In particular, the work of Kontsevich and Soibelman \cite{KS-2022,KS-2024} investigate analyticity and resurgence phenomena in the context of Lefschetz thimbles and their wall-crossing behavior, and develop a Floer-theoretic framework for exponential integrals. One aim of the present work is to introduce the viewpoint of alien calculus into the geometric setting of Picard–Lefschetz theory and wall-crossing, and to show that the interplay with vanishing cycles admits a natural interpretation in terms of the algebraic structures underlying resurgence.

\smallskip
\subsubsection*{Acknowledgments}

The authors would like to thank Gerg\H{o} Nemes, Ryszard Nest, David Sauzin, Hongfei Shu, Shanzhong Sun,
and Yifan Wu for helpful discussions.  We are especially grateful to
Maxim Kontsevich for his inspiring talk in TSIMF, Sanya, China in 2022, which provided
both inspiration and impetus for the preparation of this note. {S.~L. was supported by NSFC No.12571068.  Y.~L. was partially supported by BNSFC No.JR25001.  X.~T. was supported by BNSFC No.JR25001, BNSFC Youth Project Youth Project No.1254041 and NSFC Youth Project No.12501079.}

\medskip
\section{Exponential integrals: Picard-Lefschetz v.s. alien calculus}\label{sec:general}

This section recalls the basic framework that will be used throughout the Airy, Bessel, and Gamma examples below: the Morse flow of exponential integrals, the associated Lefschetz thimbles and their wall-crossing, and the corresponding resurgent interpretation via Borel singularities and alien operators. We refer to classical references for further details, such as Milnor \cite{Milnor-1968} on the local geometry of vanishing cycles, Pham's vanishing-homology approach to the saddle-point method \cite{Pham1983}, Arnol'd--Gusein-Zade--Varchenko's exposition \cite{AGV2}, and Vassiliev's monograph on Picard--Lefschetz theory \cite{Vassiliev2002}. For a modern Morse-theoretic formulation in terms of Lefschetz thimbles, see for example Witten \cite{Witten2011CS}. When a more intrinsic homological framework for admissible integration cycles is required, one may use the rapid-decay homology in the sense of Hien \cite{Hien2009}. In the present note, we first recall the finite-dimensional framework in a general form, but the concrete examples treated below are all one-dimensional; in each of them the relevant hypotheses can be verified directly.

\smallskip
\subsection{The basic data}

Given
\begin{itemize}
    \item a complex manifold $X$ with $\dim_{\C}X=n$;
    \item a holomorphic Morse function $S:X\rightarrow\C$ with suitable growth at $\infty$;
    \item a holomorphic volume form $\mu$ on $X$, locally, $\mu=\mu(z)\dd^nz$;
    \item an oriented cycle $\Gamma$ with $\dim_{\R}\Gamma=n$ such that
    \[\re(S)\big|_{\Gamma}\rightarrow+\infty\text{ at } "\infty" \text{ of }\Gamma.\]
\end{itemize}
We shall study the exponential integrals of the form
\begin{equation}\label{eq:general-integral}
I_\Gamma(\hbar)=\int_\Gamma e^{-S(z)/\hbar}\,\mu.
\end{equation}
The main goals are:
\begin{itemize}
    \item[$\diamond$] to recall the finite-dimensional
    Picard--Lefschetz description of exponential integrals in terms of
    thimbles, dual thimbles, Stokes phases, and connecting trajectories;
    \item[$\diamond$] to compute explicitly, in the Airy, Bessel, and Gamma
    models, the same Stokes matrices from both the Picard--Lefschetz
    wall-crossing of thimbles and the resurgent Stokes automorphisms of
    saddle expansions;
    \item[$\diamond$] to compare the two computations and verify that the
    alien/Stokes coefficients reproduce the Picard--Lefschetz trajectory
    counts.
\end{itemize}

For these purposes, it is convenient to fix a positive parameter \(|\hbar|>0\) and to rotate the action instead of the small parameter. Thus we write
\begin{equation*}
\hbar=|\hbar|e^{i\theta},
\qquad
S_\theta:=e^{-i\theta}S,
\end{equation*}
so that
\begin{equation*}
e^{-S/\hbar}=e^{-S_\theta/|\hbar|}.
\end{equation*}
The steepest-descent geometry is therefore governed by the pair of real-valued functions
\begin{equation*}\label{eq:Ftheta-Gtheta-general}
F_\theta:=\re(S_\theta)=\re(e^{-i\theta}S),
\qquad
G_\theta:=\im(S_\theta)=\im(e^{-i\theta}S).
\end{equation*}

\smallskip
\subsection{Negative gradient flow}

Choose a Hermitian metric $h$ on \(X\) and consider the negative gradient flow of \(F_\theta\),
\begin{equation}\label{eq:negative-gradient-flow-general}
\frac{dz}{ds}=-\nabla_h F_\theta(z).
\end{equation}

Let $g_h:=\operatorname{Re}h$ be the underlying Riemannian metric, and let $\langle\cdot,\cdot\rangle_h$ and $\|\cdot\|_h$ denote the corresponding inner product and norm.  The gradient $\nabla_h f$ of a real-valued function $f$ is characterized by
\[
g_h(\nabla_h f,\xi)=\dd f(\xi),\qquad \forall\,\xi~\in TX.
\]
In a local holomorphic chart $z=(z^1,\dots,z^n)$, write
\[
h=\sum_{i,j=1}^n h_{i\bar j}(z)\,\dd z^i\otimes \dd\bar z^j,
\qquad
(h^{i\bar j})=(h_{i\bar j})^{-1}.
\]
Then the negative gradient flow of $F_\theta=\operatorname{Re}(S_\theta)$ takes the coordinate form
\begin{equation}\label{eq:complex-gradient-standard}
\frac{\dd z^i}{\dd s}
=-2\sum_{j=1}^n h^{i\bar j}(z)\,\partial_{\bar j}F_\theta(z)
=-\sum_{j=1}^n h^{i\bar j}(z)\,\overline{\partial_j S_\theta(z)}
=-\sum_{j=1}^n h^{i\bar j}(z)\,\overline{e^{-i\theta}\partial_j S(z)}.
\end{equation}
In one complex dimension, if $h=h_{z\bar z}(z)\,dz\,d\bar z$, this becomes
\[
\frac{\dd z}{\dd s}=-h^{z\bar z}(z)\,\overline{e^{-i\theta}S'(z)}.
\]
For the flat metric $h_{z\bar z}\equiv 1$, we arrive at  \[\dot z=-\overline{e^{-i\theta}S'(z)}.\]

The fundamental identities along \eqref{eq:negative-gradient-flow-general} are the following.

\begin{lem}\label{lem:F-decreases-G-conserved}
Let \(z(s)\) be a solution of the gradient flow \eqref{eq:negative-gradient-flow-general}. Then
\begin{equation}\label{eq:F-decreases}
\frac{\dd}{\dd s}F_\theta(z(s))=-\lVert \nabla_h F_\theta(z(s))\rVert_h^2\le 0,
\end{equation}
and
\begin{equation}\label{eq:G-conserved}
\frac{\dd}{\dd s}G_\theta(z(s))=0.
\end{equation}
Thus every flow line lies in a level set of \(G_\theta\), while \(F_\theta\) decreases strictly along every non-stationary flow line.
\end{lem}

\begin{proof} $S_\theta=F_\theta+iG_\theta$ is holomorphic, hence $\partial_j F_\theta=i \partial_j G_\theta$. By equation \eqref{eq:complex-gradient-standard} 
$$
\frac{\dd S_\theta(z(s))}{\dd s}=\sum_{i=1}^n\frac{\dd z^i}{\dd s}\frac{\partial S_\theta(z)}{\partial z^i}=-\sum_{i,j=1}^n h^{i\bar j}(z(s))\,\partial_i S_\theta(z(s))\,\overline{\partial_j S_\theta(z(s))}=-\lVert \nabla_h F_\theta(z(s))\rVert_h^2.
$$
The lemma follows by comparing the real and imaginary part of this equation. 
\end{proof}

\smallskip
\subsection{Lefschetz thimbles: local and global descriptions}

\begin{defn}\label{def:stable-unstable}
For \(p\in\Crit(S)\), define the (global) stable and unstable manifolds for the negative gradient flow \eqref{eq:negative-gradient-flow-general} by
\begin{equation*}
W^s_\theta(p;-\nabla_h F_{\theta}):=\{z_0\in X\mid z(s;z_0)\to p\text{ as }s\to +\infty\},
\end{equation*}
\begin{equation*}
W^u_\theta(p;-\nabla_h F_{\theta}):=\{z_0\in X\mid z(s;z_0)\to p\text{ as }s\to -\infty\}.
\end{equation*}
We also denote them by $\J_p^{\theta}$ and $\Kk_p^{\theta}$ respectively, and call $\J_p^{\theta}$ Lefschetz thimble and $\Kk_p^{\theta}$ dual thimble attached to a critical point $p$.
\end{defn}

Let us analyze the local structure first. Let $p\in\Crit(S)$ be a nondegenerate critical point. By the holomorphic Morse lemma, there exist holomorphic coordinates $w=(w_1,\dots,w_n)$ centered at $p$ such that
\begin{equation}\label{eq:morse-normal-form}
S_{\theta}(w)=S_{\theta}(p)+\frac12\sum_{j=1}^n w_j^2.
\end{equation}
Writing $w_j=x_j+iy_j$, one gets
\begin{equation*}\label{eq:local-FG}
F_{\theta}(w)-F_{\theta}(p)
=\frac12\sum_{j=1}^n(x_j^2-y_j^2),
\qquad
G_{\theta}(w)-G_{\theta}(p)
=\sum_{j=1}^n x_jy_j.
\end{equation*}

To first order at $p$, after choosing normal coordinates for the metric, the negative gradient flow is
\begin{equation}\label{eq:linear-local-flow}
\dot w_j=-\overline{w_j},
\qquad\text{equivalently}\qquad
\dot x_j=-x_j,
\quad
\dot y_j=y_j.
\end{equation}
Thus $F_{\theta}$ has real Morse index $n$ at every nondegenerate critical point of a holomorphic function on an $n$-dimensional complex manifold.

\begin{prop}[Local stable and unstable models]\label{prop:local-model}
Near a nondegenerate critical point $p$, the local stable and unstable manifolds for $-\nabla_h F_{\theta}$ are middle-dimensional real submanifolds.  In the leading holomorphic Morse coordinates \eqref{eq:morse-normal-form}, they have tangent spaces
\begin{equation*}\label{eq:local-tangent-model}
T_p\J_p^\theta=\R^n=\{y=0\},
\qquad
T_p\mathcal{K}_p^\theta=i\R^n=\{x=0\}.
\end{equation*}
More precisely, for a sufficiently small ball $B_\varepsilon(p)$, one has local stable and unstable disks
\begin{align}
\J_{p,\operatorname{loc}}^\theta
&=\{x+i\psi_s(x): |x|<\varepsilon\},
\qquad \psi_s(0)=0,
\quad D\psi_s(0)=0,\label{eq:local-stable-disk}\\
\mathcal{K}_{p,\operatorname{loc}}^\theta
&=\{\psi_u(y)+iy: |y|<\varepsilon\},
\qquad \psi_u(0)=0,
\quad D\psi_u(0)=0.\label{eq:local-unstable-disk}
\end{align}
Both are contained in the level set $G_{\theta}=G_{\theta}(p)$.
\end{prop}

\begin{proof}
The linearized flow is hyperbolic, with stable space $\{y=0\}$ and unstable space $\{x=0\}$.  The stable and unstable manifold theorem gives the smooth disks \eqref{eq:local-stable-disk} and \eqref{eq:local-unstable-disk}.  By Lemma \ref{lem:F-decreases-G-conserved}, $G_{\theta}$ is constant on every flow line.  Since every point of the local stable or unstable disk flows to $p$ in one time direction, the value of $G_{\theta}$ on these disks must equal $G_{\theta}(p)$.
\end{proof}

When $n=1$, the local model is especially transparent.  Write
\[
S_{\theta}=S_{\theta}(p)+\frac12w^2,
\qquad w=x+iy.
\]
Then
\[
\J_{p,\operatorname{loc}}^\theta=\{y=0\},
\qquad
\Kk_{p,\operatorname{loc}}^\theta=\{x=0\}.
\]
The stable thimble has two local half-branches
\[
A_p^\theta=\{x>0,\ y=0\},
\qquad
B_p^\theta=\{x<0,\ y=0\},
\]
and the unstable thimble has two local half-branches
\[
C_p^\theta=\{x=0,\ y>0\},
\qquad
D_p^\theta=\{x=0,\ y<0\}.
\]
With the usual orientation convention, one may write locally
\begin{equation}\label{orientationthimble}
    \J_p^\theta\sim A_p^\theta-B_p^\theta,
\qquad
\Kk_p^\theta\sim C_p^\theta-D_p^\theta.
\end{equation}
Here each half-branch \(A_p^\theta,B_p^\theta,C_p^\theta,D_p^\theta\) is
oriented outward from the critical point \(p\); hence the signs in
\(A_p^\theta-B_p^\theta\) and \(C_p^\theta-D_p^\theta\) determine the induced
orientations of the full local thimble and dual thimble.

Let $\phi_s$ denote the flow of $-\nabla_h F_{\theta}$. With the above analysis, the global stable and unstable manifolds are also given by
\begin{equation*}\label{eq:flowout-local-disks}
\J_p^\theta
=\bigcup_{s\ge 0}\phi_{-s}\bigl(\J_{p,\operatorname{loc}}^\theta\bigr),
\qquad
\Kk_p^\theta
=\bigcup_{s\ge 0}\phi_{s}\bigl(\Kk_{p,\operatorname{loc}}^\theta\bigr).
\end{equation*}

Again, by Lemma \ref{lem:F-decreases-G-conserved}, one has
\begin{equation*}\label{eq:G-level-global}
\J_p^\theta\subset \{G_{\theta}=G_{\theta}(p)\},
\qquad
\Kk_p^\theta\subset \{G_{\theta}=G_{\theta}(p)\}.
\end{equation*}
Moreover, along $\J_p^\theta$, the forward flow decreases $F_{\theta}$ to $F_{\theta}(p)$, while the backward flow increases $F_{\theta}$.  Therefore, under the usual tameness/no-critical-points-at-infinity hypotheses, the outgoing end of $\J_p^\theta$ lies in the region
\[
F_{\theta}\to +\infty.
\]
Similarly, the outgoing end of $\Kk_p^\theta$ lies in the region
\[
F_{\theta}\to -\infty.
\]

For $R\gg 0$, define the positive and negative ends
\begin{equation*}\label{eq:relative-ends}
X_{+R}^\theta=\{x\in X:F_{\theta}(x)\ge R\},
\qquad
X_{-R}^\theta=\{x\in X:F_{\theta}(x)\le -R\}.
\end{equation*}
The truncated thimble $\J_p^\theta\cap\{F_{\theta}\le R\}$ has boundary in $X_{+R}^\theta$, and the truncated dual thimble $\Kk_p^\theta\cap\{F_{\theta}\ge -R\}$ has boundary in $X_{-R}^\theta$.  Hence they define relative cycles
\begin{equation*}\label{eq:relative-classes}
[\J_p^\theta]\in H_n(X,X_{+R}^\theta;\Z),
\qquad
[\Kk_p^\theta]\in H_n(X,X_{-R}^\theta;\Z).
\end{equation*}
For different sufficiently large $R$, these relative groups are canonically identified, so we also write
\begin{equation*}\label{eq:infinity-relative-groups}
H_n(X,X_{+\infty}^\theta;\Z),
\qquad
H_n(X,X_{-\infty}^\theta;\Z).
\end{equation*}

The phase $\theta$ determines which trajectories are visible, because every trajectory is constrained to a level set of $G_{\theta}$.  For simplicity, assume that the critical values of \(S\) are pairwise distinct, then for two critical points $p,q$, a necessary condition for a gradient trajectory from $q$ to $p$ is
\begin{equation*}\label{eq:necessary-stokes-condition}
G_{\theta}(p)=G_{\theta}(q),
\end{equation*}
that is,
\begin{equation*}\label{eq:stokes-wall-condition}
\im\bigl(e^{-i\theta}(S(p)-S(q))\bigr)=0.
\end{equation*}
Thus the exceptional phases are precisely the Stokes phases determined by differences of critical values.

\begin{defn}\label{def:regular-phase}  
A phase $\theta$ is called regular, or non-Stokes, if
\begin{equation*}\label{eq:regular-phase}
G_{\theta}(p)\ne G_{\theta}(q)
\qquad\text{for all distinct }p,q\in\Crit(S).
\end{equation*}
Equivalently, no two distinct critical values of $S$ lie on a line of direction $e^{i\theta}\R$ in the $S$-plane.
\end{defn}

For a regular phase, $\J_p^\theta$ and $\Kk_q^\theta$ can intersect only when $p=q$.  Indeed, if $x\in \J_p^\theta\cap\Kk_q^\theta$, then $x$ lies on a trajectory with forward limit $p$ and backward limit $q$.  Conservation of $G_{\theta}$ forces
\[
G_{\theta}(p)=G_{\theta}(x)=G_{\theta}(q),
\]
so regularity implies $p=q$.

\begin{defn}\label{def:tame-data}
We say that $(X,S,\theta,h)$ is tame Lefschetz data if the following conditions hold.
\begin{enumerate}[label=\textup{(\arabic*)}, leftmargin=2.5em]
\item $S$ has finitely many critical points, and all of them are nondegenerate.
\item The phase $\theta$ is regular in the sense of Definition \ref{def:regular-phase}.
\item The flow of $-\nabla_hF_{\theta}$ is complete on the relevant stable and unstable manifolds.
\item There exists $R\gg 0$ such that outside the compact region $|F_{\theta}|\le R$, the function $F_{\theta}$ has no critical points and the gradient flow escapes to one of the two ends $F_{\theta}\to\pm\infty$.
\item The stable and unstable manifolds are transverse.\footnote{For regular $\theta$, this is automatic after a harmless generic choice of K\"ahler metric that does not change the local holomorphic Morse data.}
\end{enumerate}
\end{defn}

\begin{rmk}[Finite and locally finite versions]
Condition \textup{(1)} is the finite-critical-point version, and it is the one
used literally in the Airy and Bessel models below.  In the Gamma model the
critical set is countable rather than finite, so the definition should be
understood in a locally finite sense: every compact subset of \(X\) contains
only finitely many critical points, and the set of critical values has no
finite accumulation point.
Equivalently, after fixing a Stokes phase \(\theta_*\), for every
\(c\in\R\) and every bounded interval \([a,b]\subset\R\), the set
\[
\left\{
p\in\Crit(S):
\operatorname{Im}\!\left(e^{-i\theta_*}S(p)\right)=c,\quad
a\le
\operatorname{Re}\!\left(e^{-i\theta_*}S(p)\right)
\le b
\right\}
\]
is finite.  Thus the Picard--Lefschetz computation may first be performed
on finite subsets of critical values and then passed to the locally finite
limit.
In this setting, the thimble-basis statement below is interpreted in locally
finite relative homology, or equivalently in the corresponding completed
thimble module.  The resulting Stokes transformations are infinite triangular
transformations understood in this completed sense.  For systematic treatments
of locally finite Stokes data, see Hien \cite{Hien2009} and Sabbah
\cite{Sabbah2013}.
\end{rmk}

\begin{thm}[Thimble basis theorem, \cite{Pham1983,Vassiliev2002,Witten2011CS}]\label{thm:thimble-basis}
Assume $(X,S,\theta,h)$ is tame Lefschetz data and $\dim_\C X=n$.  Then every critical point $p\in\Crit(S)$ determines:
\begin{enumerate}[label={\textup{(\roman*)}}, leftmargin=2.5em]
\item a Lefschetz thimble
\[
\J_p^\theta=W^s(p;-\nabla_h F_{\theta}),
\qquad
[\J_p^\theta]\in H_n(X,X_{+\infty}^\theta;\Z);
\]
\item a dual thimble
\[
\Kk_p^\theta=W^u(p;-\nabla_h  F_{\theta}),
\qquad
[\Kk_p^\theta]\in H_n(X,X_{-\infty}^\theta;\Z).
\]
\end{enumerate}
Moreover,
\begin{equation*}\label{eq:basis-positive}
\bigl\{[\J_p^\theta]:p\in\Crit(S)\bigr\}
\end{equation*}
is a $\Z$-basis of the relative homology group
\[
H_n(X,X_{+\infty}^\theta;\Z),
\]
and
\begin{equation*}\label{eq:basis-negative}
\bigl\{[\Kk_p^\theta]:p\in\Crit(S)\bigr\}
\end{equation*}
is a $\Z$-basis of
\[
H_n(X,X_{-\infty}^\theta;\Z).
\]
With compatible orientations\footnote{Our convention for orienting the
thimble and dual thimble is the one described in
\eqref{orientationthimble}. }, the intersection pairing satisfies
\begin{equation*}\label{eq:dual-pairing}
\langle[\J_p^\theta],[\Kk_q^\theta]\rangle=\delta_{pq}.
\end{equation*}
In one complex dimension, this is the left-handed convention for the local
pair of thimble and dual thimble, chosen so that their oriented intersection
at the critical point is \(+1\). In particular,
\begin{equation*}\label{eq:rank-count}
\operatorname{rank}H_n(X,X_{+\infty}^\theta;\Z)
=\#\Crit(S).
\end{equation*}
\end{thm}

\smallskip
\subsection{Connecting trajectories and Picard-Lefschetz formula}

At a Stokes phase $\theta_*$, where $G_{\theta_*}(p)=G_{\theta_*}(q)$ for some $p\ne q$, connecting trajectories may appear.  When $\theta$ crosses such a wall, the thimble basis changes by a Picard--Lefschetz transformation.  

More explicitly, fix a Stokes phase \(\theta_*\), and let \(p_0,\dots,p_r\) be the critical points whose critical values lie on the same Stokes line (i.e. along the line in the direction $e^{i\theta_*}$), ordered so that
\begin{equation*}
F_{\theta_*}(p_0)>\cdots>F_{\theta_*}(p_r).
\end{equation*}
Write the two one-sided thimble bases as
\begin{equation*}
\mathbf J^<:=([J_{p_0}^<],\dots,[\J_{p_r}^<]),
\qquad
\mathbf J^>:=([\J_{p_0}^>],\dots,[\J_{p_r}^>]).
\end{equation*}
Here $<$ (resp. $>$) means the direction $\theta_*-\varepsilon$ (resp. $\theta_*+\varepsilon$) with small positive $\varepsilon$.
Then there are unique unitriangular matrices \(R_+\) and \(R_-\) such that 
\begin{equation*}
\mathbf J^>=\mathbf J^< R_-,
\qquad
\mathbf J^<=\mathbf J^> R_+,
\end{equation*}
hence \(R_-=R_+^{-1}\). If either \(R_+\) or \(R_-\) has the elementary form
\begin{equation*}
\begin{pmatrix}
1 &        &        &        &        &        \\
  & \ddots &        &        &        &        \\
  &        & 1      & N      &        &        \\
  &        & 0      & 1      &        &        \\
  &        &        &        & \ddots &        \\
  &        &        &        &        & 1
\end{pmatrix},
\end{equation*}
with all omitted off-diagonal entries equal to \(0\), then the elementary Picard--Lefschetz jump formula implies that \(N\) is the signed count of direct connecting trajectories from \(p_i\) to \(p_{i+1}\), where \(N\) sits in the \((i,i+1)\)-position; in complex dimension one, \(|N|\) is the number of such trajectories. See \cite[Sec.~3.1.2]{Witten2011CS}; cf.\ also \cite{AGV2,Pham1983,Vassiliev2002}.

This is exactly the situation relevant for the examples below. 
\begin{itemize}
    \item For Airy and Bessel, one has \(r=1\), so the relevant matrix is
\begin{equation*}
\begin{pmatrix}
1 & N\\
0 & 1
\end{pmatrix},
\end{equation*}
with \(N=1\) in the Airy case and \(N=2\) in the Bessel case.
   \item  For the Gamma model, the same discussion applies with an infinite unitriangular matrix in place of the finite one; in that example, one of \(R_+\) and \(R_-\) is not elementary, while the other is of nearest-neighbor form, and it is this latter matrix from which the direct connecting trajectories are read off.
\end{itemize}

\begin{rmk}
If we do not assume that the critical values of \(S\) are pairwise distinct, then Stokes walls should be defined for distinct critical values rather 
than for individual critical points.  For two critical values \(A\neq B\), the 
phase \(\theta_*\) is Stokes if
\[
\im\bigl(e^{-i\theta_*}(A-B)\bigr)=0.
\]
At such a phase the Picard--Lefschetz jump is a block transformation between the 
critical clusters \(\operatorname{Crit}_A(S)\) and \(\operatorname{Crit}_B(S)\). 
If \(p_\alpha\in \operatorname{Crit}_A(S)\), then
\[
[\mathcal J_\alpha^{\theta_{*+}}]
=
[\mathcal J_\alpha^{\theta_{*-}}]
+
\sum_{\substack{B\neq A\\
\operatorname{Im}(e^{-i\theta_*}(A-B))=0}}
\sum_{p_\beta\in\operatorname{Crit}_B(S)}
n_{\alpha\beta}
[\mathcal J_\beta^{\theta_{*-}}],
\]
where \(n_{\alpha\beta}\) is the signed count of \(F_{\theta_*}\)-gradient
trajectories from \(p_\alpha\) to \(p_\beta\).  Critical points with the same
critical value \(A\) form a resonant block.  They do not produce an ordinary
Stokes jump among themselves, since any finite-energy trajectory between two
such critical points would have zero energy and hence would be constant.  
\end{rmk}

\smallskip
\subsection{Resurgent interpretation}\label{sec:resurgence}
For background on resurgence, alien calculus, and transseries, we refer the reader to \'Ecalle's foundational works \cite{Ecalle81a,Ecalle81b,Ecalle85}, to the pedagogical account in \cite{MitschiSauzin2016}, and to Dorigoni's introduction from the transseries point of view \cite{Dorigoni2014}.
From the resurgent point of view, the primary objects are the formal saddle
expansions of the exponential integrals on the Lefschetz thimbles.  Fix a
regular phase \(\theta\), and let \(\mathfrak s_\theta\subset\C^\times\) be
a sufficiently small angular sector centered at direction \(\theta\),
containing no Stokes ray in its interior.  For a critical point
\(p\in\Crit(S)\), let \(\J_p^\theta\) be the corresponding Lefschetz
thimble.  Then, as \(\hbar\to0\) with \(\hbar\in\mathfrak s_\theta\), the
thimble integral
\begin{equation*}
I_{\J_p^\theta}(\hbar)
:=
\int_{\J_p^\theta} e^{-S/\hbar}\,\mu
\end{equation*}
admits a formal asymptotic expansion produced by the saddle-point method:
\begin{equation*}
I_{\J_p^\theta}(\hbar)
\sim
\widetilde I_p(\hbar),
\qquad
\widetilde I_p(\hbar)
=
e^{-S(p)/\hbar}\,
\widetilde\Psi_p(\hbar),
\end{equation*}
where
\begin{equation*}
\widetilde\Psi_p(\hbar)
=
C_p\,\hbar^{n/2}
\left(1+\sum_{m\ge1}a_m\hbar^m\right)
\in
\hbar^{n/2}\C[[\hbar]].
\end{equation*}
Here \(C_p\in\C^\times\) is the Gaussian coefficient determined by the
Hessian of \(S\), the holomorphic volume form \(\mu\), and the chosen local
orientation of the thimble through \(p\).  In the one-dimensional examples
below, this coefficient takes the form \(C_p=A_p\sqrt{2\pi}\), with
\(|A_p|=1\) recording the local orientation phase.

Once this local orientation is fixed, the Gaussian coefficient and the
reduced formal power series depend only on the local saddle \(p\), not on
the regular phase \(\theta\).  Thus the role of \(\theta\) is to determine
which thimble basis is realized globally; when \(\theta\) crosses a Stokes
direction, the thimble basis may jump, while the local formal objects
\(\widetilde I_p(\hbar)\) remain the same and are recombined by the Stokes
automorphism.  This is precisely the role played by alien operators.

For \(\re\alpha>0\), we use the shifted Borel transform
\begin{equation}\label{Boreltransform}
\mathcal B
\left(
\sum_{m\ge0} a_m \hbar^{m+\alpha}
\right)(\xi)
:=
\sum_{m\ge0} a_m
\frac{\xi^{m+\alpha-1}}{\Gamma(m+\alpha)}
\end{equation}
and denote
$$
\widehat\Psi_p(\xi)=\mathcal B(\widetilde\Psi_p(\hbar)).
$$
We now introduce the alien operator for an integrable singularity. 
Within the general framework of resurgent functions, we shall mainly use a special class of resurgent formal series whose Borel transforms have only integrable singularities,
\footnote{
Let \(\omega\) be a singular point of an analytically continued Borel germ
\(\widehat\phi\).  We say that the singularity at \(\omega\) is integrable if,
after writing \(\zeta=\xi-\omega\), every local branch of
\(\widehat\phi(\omega+\zeta)\) on any fixed sector of the logarithmic Riemann
surface of \(\zeta\) is integrable at \(\zeta=0\).  Equivalently, in the sense
used here, on every such fixed sector one has 
\[
\widehat\phi(\omega+\zeta)=o(\zeta^{-1}),
\qquad \zeta\to0,\] uniformly on closed subsectors.  The sector is allowed to have opening larger
than \(2\pi\).
}
meaning that, after analytic continuation, the corresponding local germs are
locally integrable near every singular point; we denote by
\(\widetilde{\mathcal R}^{\mathrm{int}}\) the algebra of such integrable
resurgent formal power series.
Although alien operators can be defined in a much broader resurgent setting,
all three examples treated below lie in
\(\widetilde{\mathcal R}^{\mathrm{int}}\), so we restrict ourselves to this
integrable case and do not recall the general definition.

Fix a singular point \(\omega\) of the Borel germ \(\widehat\Psi_p\). The alien operator \(\Delta_\omega^{+}\) is defined by taking the local variation of \(\widehat\Psi_p\) at \(\omega\), translating the resulting germ back to the origin, and then applying the inverse Borel transform; schematically, one analytically continues the Borel germ to the singular point, takes its local monodromy there, and then re-centers the resulting germ at the origin, as indicated in Figure~\ref{fig:alien-operator-schematic}.
\begin{equation*}
\Delta_\omega^{+}:=\mathcal B^{-1}\circ \tau_{-\omega}\circ \mathrm{var}_\omega^{+}\circ \mathcal B,
\end{equation*}
where
\begin{equation*}
(\tau_{-\omega}f)(\xi):=f(\xi+\omega).
\end{equation*}
Here
\begin{equation*}
\mathrm{var}_\omega^{+}(f):=\cont_{\gamma_\omega^{+}}f-\cont_{\gamma_\omega^{-}}f,
\end{equation*}
and \(\cont_{\gamma}f\) denotes the analytic germ obtained by analytically
continuing the germ \(f\) from a neighborhood of the origin along the path
\(\gamma\) to the endpoint \(\gamma(1)\). See Figure \ref{fig:alien-operator-schematic} for the shape of $\gamma_\omega^{\pm}$. \(\mathrm{var}_\omega^+(f)\) measures the local
monodromy of the analytically continued Borel germ at \(\omega\).

\begin{figure}[H]
\centering
\begin{tikzpicture}[>=Latex, scale=0.62]

\coordinate (O) at (0,0);
\coordinate (W) at (5.2,3.5);
\coordinate (WW) at (7.8,5.27);
\coordinate (M1) at (2.6,1.75);
\coordinate (M2) at (1.71,1.13);

\fill[green!20] (5.6,3.7) circle (0.38);
\draw[green!50!black] (5.6,3.7) circle (0.38);
\node[green!40!black] at (6,4.5) {$D$};

\filldraw[black] (O) circle (1.3pt);
\node[below left] at (O) {$0$};

\draw[dashed, thick] (O) -- (WW);
\draw[->,thick]
    (0.38,0.00)
    .. controls (1.2,0.05) and (3.1,0.85) ..
    (4,1.95)
    .. controls (4.62,2.55) and (4.70,5.00) .. (5.4,3.9)
    ;
\draw[->,thick]
    (0.38,0.00)
    .. controls (1.2,0.05) and (3.2,0.75) ..
    (5.35,2.7)
    .. controls (5.62,3) and (5.8,3.2) .. (5.6,3.7)
    ;
\draw[red, line width=1.1pt]
    ($(W)+(-0.12,-0.12)$) -- ($(W)+(0.12,0.12)$);
\draw[red, line width=1.1pt]
    ($(W)+(-0.12,0.12)$) -- ($(W)+(0.12,-0.12)$);
\draw[red, line width=1.1pt]
    ($(M1)+(-0.12,-0.12)$) -- ($(M1)+(0.12,0.12)$);
\draw[red, line width=1.1pt]
    ($(M1)+(-0.12,0.12)$) -- ($(M1)+(0.12,-0.12)$);
\draw[red, line width=1.1pt]
    ($(M2)+(-0.12,-0.12)$) -- ($(M2)+(0.12,0.12)$);
\draw[red, line width=1.1pt]
    ($(M2)+(-0.12,0.12)$) -- ($(M2)+(0.12,-0.12)$);
\node[below,red] at (W) {$w$};
\node[above] at (6.05,1.58) {$\gamma_w^+$};
\node[above] at (4.65,4.58) {$\gamma_w^-$};
\end{tikzpicture}
\caption{The path \(\gamma_\omega^+\) analytically continues the Borel germ
from the origin to a small disc \(D\) behind the singular point \(\omega\),
passing to the right of the preceding singularities on the same ray.  The
path \(\gamma_\omega^-\) is obtained from \(\gamma_\omega^+\) by adjoining a
small clockwise loop around \(\omega\) inside \(D\).}
\label{fig:alien-operator-schematic}
\end{figure}
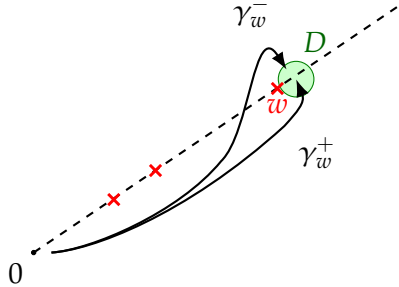

For example, if \(\widehat\Psi_p(\xi)\) has nearest singularity at \(\omega\) and locally is of the form
\begin{equation*}
\widehat\Psi_p(\xi)=H(\xi)+\frac{1}{2\pi i}\log(\xi-\omega)\,\widehat\Psi_q(\xi-\omega),
\end{equation*}
with \(H(\xi)\) holomorphic near \(\omega\), then
\begin{equation*}
\mathrm{var}_\omega^{+}\widehat\Psi_p(\xi)=\widehat\Psi_q(\xi-\omega),
\end{equation*}
and therefore
\begin{equation*}
\Delta_\omega^{+}\widehat\Psi_p=\widehat\Psi_q. 
\end{equation*}
Equivalently, we write $\Delta_\omega^{+}\widetilde\Psi_p=\widetilde\Psi_q$. 

The pointed alien operator is defined to be
\begin{equation*}
\dot\Delta_\omega^{+}:=e^{-\omega/\hbar}\Delta_\omega^{+},
\end{equation*}
and the Stokes automorphism along a fixed ray \(d\) is
\begin{equation}\label{eq:Stokesauto}
\mathfrak S^{+} :=\Id+\sum_{\omega\in d}\dot\Delta_\omega^{+}.
\end{equation}
The opposite alien operator \(\Delta_\omega^{-}\) is defined in the similar way, with \(\mathrm{var}_\omega^{-}\) taken from the negative side of the Stokes ray; correspondingly,
\begin{equation*}
\dot\Delta_\omega^{-}:=e^{-\omega/\hbar}\Delta_\omega^{-},
\qquad
\mathfrak S^{-}:=\Id+\sum_{\omega\in d}\dot\Delta_\omega^{-}.
\end{equation*}
The two Stokes automorphisms are inverse to each other,
\begin{equation*}
\mathfrak S^{-}=(\mathfrak S^{+})^{-1}.
\end{equation*}

We next define the lateral Borel--Laplace sums for the series $\widetilde I_p(\hbar)=e^{-S(p)/\hbar}\widetilde\Psi_p(\hbar)$. Fix a ray
\begin{equation*}
d=e^{i\theta}\R_{>0}
\end{equation*}
in the Borel plane, and let \(d^{>}\) and \(d^{<}\) denote the two standard lateral deformations of \(d\), passing respectively to the right and to the left of the singular points on \(d\).  
If a Borel germ \(\widehat\Psi(\xi)\) admits analytic continuation along \(d^{>}\) and \(d^{<}\) and has at most exponential growth there, we define
\begin{equation*}
\mathcal L^{>} \widehat\Psi(\hbar):=\int_{d^{>}} e^{-\xi/\hbar}\,\widehat\Psi(\xi)\,\dd\xi,
\qquad
\mathcal L^{<} \widehat\Psi(\hbar):=\int_{d^{<}} e^{-\xi/\hbar}\,\widehat\Psi(\xi)\,\dd\xi,
\end{equation*}
and set
\begin{equation*}
\mathcal S^{>} \widetilde I_p(\hbar):=e^{-S(p)/\hbar}\mathcal L^{>} \widehat\Psi_p(\hbar),
\qquad
\mathcal S^{<} \widetilde I_p(\hbar):=e^{-S(p)/\hbar}\mathcal L^{<} \widehat\Psi_p(\hbar).
\end{equation*}
The two lateral Borel-Laplace sums satisfy
\begin{equation}\label{Stokesauto}
\mathcal S^{<}=\mathcal S^{>}\circ \mathfrak S^{+},
\qquad
\mathcal S^{>}=\mathcal S^{<}\circ \mathfrak S^{-}.
\end{equation}
See Figure \ref{fig:stokes-automorphism-zeta-plane}.

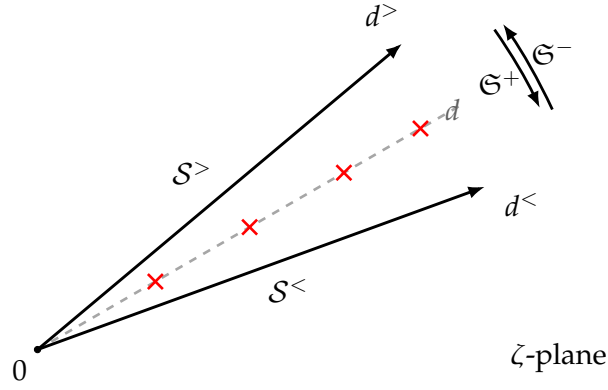
\begin{figure}[!htbp]
\centering
\begin{tikzpicture}[scale=0.9,>=Latex]

\tikzset{
  stokesray/.style={
    gray!70,
    dashed,
    line width=1pt
  },
  lateralsum/.style={
    black,
    line width=1.1pt,
    -{Latex[length=2.2mm]}
  },
  stokesarcminus/.style={
    black,
    line width=1.05pt,
    -{Latex[length=2.2mm]}
  },
  stokesarcplus/.style={
    black,
    line width=1.05pt,
    -{Latex[length=2.2mm]}
  }
}
\coordinate (O) at (0,0);
\draw[stokesray] (O) -- (30:7.2);
\draw[lateralsum] (O) -- (20:7.0);
\draw[lateralsum] (O) -- (40:7.0);
\filldraw[black] (O) circle (1.3pt);
\node[below left] at (O) {$0$};
\node[below right] at (6.8,0.25) {$\zeta$-plane};
\node[gray!75!black] at (30:7.05) {$d$};
\foreach \r in {2.0,3.6,5.2,6.5}{
  \draw[red, line width=1pt]
    ($(30:\r)+(-0.11,-0.11)$) -- ($(30:\r)+(0.11,0.11)$);
  \draw[red, line width=1pt]
    ($(30:\r)+(-0.11,0.11)$) -- ($(30:\r)+(0.11,-0.11)$);
}
\node at (20:7.15) [below right] {$d^{<}$};
\node at (40:7.15) [above left] {$d^{>}$};
\node at (20:3.45) [below right] {$\mathcal S^{<}$};
\node at (40:3.55) [above left] {$\mathcal S^{>}$};
\draw[stokesarcminus]
  (25:8.35) arc[start angle=25,end angle=35,radius=8.35];
\node at (30:8.75) {$\mathfrak S^{-}$};
\draw[stokesarcplus]
  (35:8.2) arc[start angle=35,end angle=25,radius=7.65];
\node at (30:7.85) {$\mathfrak S^{+}$};
\end{tikzpicture}
\caption{Stokes ray \(d\) in the Borel \(\zeta\)-plane, its neighboring lateral directions \(d^{<}\) and \(d^{>}\), and the corresponding Stokes automorphisms \(\mathfrak S^\pm\).}
\label{fig:stokes-automorphism-zeta-plane}
\end{figure}

\begin{prop}\label{prop:thimble-lateral-sum}
Let \(p\in\Crit(S)\) be non-degenerate, and let \(\J_p^{>}\) and \(\J_p^{<}\) be the Lefschetz thimbles in the two generic phases adjacent to a Stokes phase, with the same local orientation chosen at \(p\). Then
\begin{equation}\label{eq:thimble-lateral-sum}
I_{\J_p^{>}}(\hbar)=\mathcal S^{>} \widetilde I_p(\hbar),
\qquad
I_{\J_p^{<}}(\hbar)=\mathcal S^{<} \widetilde I_p(\hbar).
\end{equation}
\end{prop}

\begin{proof}
Choose a holomorphic volume form \(\mu\) near \(p\), and write the thimble integral in the form
\[
I_{\J_p^{\gtrless}}(\hbar)=\int_{\J_p^{\gtrless}} e^{-S/\hbar}\,\mu.
\]
Since \(p\) is a nondegenerate critical point, the holomorphic Morse lemma gives local coordinates \(u=(u_1,\dots,u_n)\) centered at \(p\) such that
\[
S(z)=S(p)+\frac12\sum_{j=1}^n u_j^2,
\qquad
\mu=a(u)\,\dd u_1\wedge\cdots\wedge \dd u_n,
\qquad a(0)\neq 0.
\]
On the complement of the critical locus, let \(\mu/\dd S\) denote the Gelfand--Leray \((n-1)\)-form characterized by
\[
\dd S\wedge \frac{\mu}{\dd S}=\mu.
\]
Set
\[
\xi:=S(z)-S(p).
\]

Along \(\J_p^{\gtrless}\), the map \(\xi\) takes values on the lateral ray \(d^{\gtrless}\).  For each regular value \(\xi\in d^{\gtrless}\), define the vanishing cycle slice
\[
\Sigma_{\xi}^{\gtrless}:=\J_p^{\gtrless}\cap\{S-S(p)=\xi\}.
\]
Then fiber integration gives
\begin{equation}\label{eq:thimble-change-variable}
I_{\J_p^{\gtrless}}(\hbar)
=
e^{-S(p)/\hbar}
\int_{d^{\gtrless}} e^{-\xi/\hbar}\,\rho_p^{\gtrless}(\xi)\,\dd\xi,
\qquad
\rho_p^{\gtrless}(\xi):=\int_{\Sigma_{\xi}^{\gtrless}} \frac{\mu}{\dd S}.
\end{equation}

\begin{figure}[!htbp]
\centering
\begin{tikzpicture}[scale=0.8,>=Latex]

\tikzset{
  thimblebd/.style={red!70!black,line width=1.1pt},
  fib/.style={green!70!black,line width=1.0pt},
  proj/.style={yellow!70!black,dashed,line width=0.9pt},
  ray/.style={red!70!black,line width=1.1pt},
  maparrow/.style={->,black,line width=1.0pt},
  pt/.style={circle,fill=green!70!black,inner sep=1.7pt}
}
\node at (4,3.5) {\(\mathcal J_p^{>}\subset z\text{-space}\)};
\draw[thimblebd]
  (-3.0,2.2)
  .. controls (-2.0,4.1) and (0.8,4.2) ..
  (3.2,4.3);
\draw[thimblebd]
  (-3.0,2.2)
  .. controls (-1.8,1.0) and (0.2,1.7) ..
  (3.2,2.3);
\node[pt,label={[below left]{$p$}}] (p) at (-3.0,2.2) {};
\node[pt,inner sep=1.2pt] (q) at (0.05,1.63) {};
\draw[fib]
  (0.05,1.63)
  .. controls (-0.5,1.8) and (-0.5,3.45) ..
  (0.05,4.08);
\draw[fib,dashed]
  (0.05,4.08)
  .. controls (0.68,3.42) and (0.75,2.72) ..
  (0.05,1.63);
\node[green!70!black] at (0.92,3.2) {\(\Sigma_\xi^>\)};
\draw[proj] (p) -- (-3.0,-0.95);
\draw[proj] (q) -- (0.05,-0.25);
\draw[maparrow] (-0.55,1.45) -- (-0.55,-0.25);
\node[right,black] at (-0.35,0.8) {\(S(z)-S(p)\)};
\node at (0,-1.65) {\(\xi\text{-plane}\)};
\draw[ray,->] (-2.9,-0.95) -- (3.3,0.45);
\node[below right,black] at (2.45,0.35) {\(e^{id^>}\mathbb R_{\ge0}\ \subset\ \xi\) -plane};
\node[pt,label={[below left]{$0$}}] at (-2.9,-0.95) {};
\node[pt] at (0.05,-0.29) {};
\node[green!70!black, below right] at (0.05,-0.22) {\(\xi\)};
\end{tikzpicture}
\caption{
A schematic picture of the map \(\xi=S(z)-S(p)\): the thimble \(\mathcal J_p^{>}\) is sent to a ray in the \(\xi\)-plane, and for a regular value \(\xi\) the fibre
\(\Sigma_\xi^>=\mathcal J_p^{>}\cap\{S(z)-S(p)=\xi\}\)
is drawn schematically as a vanishing cycle.
}
\label{fig:thimble-ray-fibre-schematic}
\end{figure}
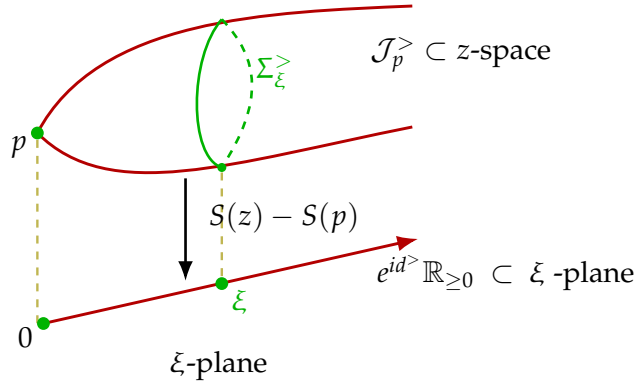

In the Morse coordinates above, \(\Sigma_{\xi}^{\gtrless}\) is the vanishing \((n-1)\)-sphere cut out by \(\sum u_j^2=2\xi\) inside the chosen thimble, so one obtains the standard expansion
\[
\rho_p^{\gtrless}(\xi)
=
A_p\,\frac{(2\xi)^{\frac n2-1}}{\Gamma(\frac n2)}
\Bigl(1+b_1\xi+b_2\xi^2+\cdots\Bigr),
\qquad |A_p|=1.
\]
Because the local orientation at \(p\) is fixed once and for all, the same phase \(A_p\) appears on the two lateral sides.  Therefore \(\rho_p^{\gtrless}(\xi)\) is precisely the two-sided analytic continuation of the Borel transform \(\widehat\Psi_p=\mathcal B\widetilde\Psi_p\) of the local saddle series
\[
\widetilde I_p(\hbar)=e^{-S(p)/\hbar}\widetilde\Psi_p(\hbar),
\qquad
\widetilde\Psi_p(\hbar)\in \hbar^{n/2}\C[[\hbar]].
\]
Substituting this identification into \eqref{eq:thimble-change-variable} gives
\[
I_{\J_p^{\gtrless}}(\hbar)
=
e^{-S(p)/\hbar}
\int_{d^{\gtrless}} e^{-\xi/\hbar}\widehat\Psi_p(\xi)\,\dd\xi
=
\mathcal S^{\gtrless}\widetilde I_p(\hbar),
\]
which is exactly \eqref{eq:thimble-lateral-sum}.  

When \(n=1\), the fiber \(\Sigma_{\xi}^{\gtrless}\) consists of the two points \(z_A^{\gtrless}(\xi)\) and \(z_B^{\gtrless}(\xi)\), and \eqref{eq:thimble-change-variable} reduces to
\[
\rho_p^{\gtrless}(\xi)
=
\frac{\dd z_A^{\gtrless}}{\dd\xi}(\xi)-\frac{\dd z_B^{\gtrless}}{\dd\xi}(\xi),
\]
which is the original one-dimensional branch formula.
\end{proof}

Assume for the
moment that \(S\) has finitely many non-degenerate critical points.  For each
\(p\in \Crit(S)\), write the corresponding formal expansion in the form
\[
\widetilde I_p(\hbar)
=
e^{-S(p)/\hbar}\widetilde\Psi_p(\hbar),
\qquad
\widetilde\Psi_p\in \widetilde{\mathcal R}^{\mathrm{int}}.
\]
We define the critical-value graded vector space
\[
\mathcal M_S
:=
\bigoplus_{A\in S(\Crit(S))}\mathcal M_A,
\qquad
\mathcal M_A
:=
\operatorname{Span}_{\C}
\left\{
\widetilde I_p \mid p\in\Crit(S),\ S(p)=A
\right\}.
\]
Thus \(\widetilde I_p\) is homogeneous of degree \(S(p)\).  In the case where
the critical values are pairwise distinct, each graded piece
\(\mathcal M_{S(p)}\) is one-dimensional and is generated by
\(\widetilde I_p\).

Throughout this paper, we use the convention that the ordinary alien operator
\(\Delta_\omega^\pm\) acts only on the reduced resurgent series 
\(\widetilde\Psi_p\) and leaves the exponential factor fixed:
\[
\Delta_\omega^\pm
\left(
e^{-S(p)/\hbar}\widetilde\Psi_p
\right)
:=
e^{-S(p)/\hbar}\Delta_\omega^\pm\widetilde\Psi_p.
\]
Hence the pointed alien operator
\[
\dot\Delta_\omega^\pm:=e^{-\omega/\hbar}\Delta_\omega^\pm
\]
acts on a homogeneous element by
\[
\dot\Delta_\omega^\pm
\left(
e^{-S(p)/\hbar}\widetilde\Psi_p
\right)
=
e^{-(S(p)+\omega)/\hbar}
\Delta_\omega^\pm\widetilde\Psi_p.
\]

Write
\begin{equation*}
\widetilde{\mathbf I}:=\bigl(\widetilde I_{p_0},\dots,\widetilde I_{p_r}\bigr)
\end{equation*}
for the asymptotic series of 
\begin{equation*}
\mathbf I^{>}:=\mathbf I_{\mathbf J^{>}}
=\bigl(I_{\J_{p_0}^>},\dots,I_{\J_{p_r}^>}\bigr) 
\qquad
\text{and}
\qquad
\mathbf I^{<}:=\mathbf I_{\mathbf J^{<}}=
\bigl(I_{\J_{p_0}^<},\dots,I_{\J_{p_r}^<}\bigr) 
\end{equation*}
where
\begin{equation*}
\mathbf J^{>}:=\bigl(\J_{p_0}^>,\dots,\J_{p_r}^>\bigr)
\qquad \text{and}
\qquad
\mathbf J^{<}:=\bigl(\J_{p_0}^<,\dots,\J_{p_r}^<\bigr).
\end{equation*}

\begin{prop}
\label{prop:PL-origin-pointed-alien}
Assume that \(S\) has finitely many non-degenerate critical points and that
the critical values \(S(p)\), \(p\in\Crit(S)\), are pairwise distinct.  Fix a
Stokes ray \(d\) and the two adjacent generic phases defining
\(\mathbf J^{>}\) and \(\mathbf J^{<}\).  Let
\(R_+\in GL_{r+1}(\Z)\) be the Picard--Lefschetz wall-crossing matrix defined by
\begin{equation}\label{eq:PL-matrix-R-plus}
\mathbf J^{<}=\mathbf J^{>}R_+ .
\end{equation}
Then the Stokes automorphism preserves the critical-value graded vector space
\(\mathcal M_S\), and one has
\begin{equation}\label{eq:Stokes-action-PL-matrix}
\mathfrak S^+\widetilde{\mathbf I}
=
\widetilde{\mathbf I}R_+ .
\end{equation}
Consequently, for every \(p_j\in\Crit(S)\) and every singular point
\(\omega\in d\),
\begin{equation}\label{eq:pointed-alien-action-gap-formula}
\dot\Delta_\omega^+\widetilde I_{p_j}
=
\sum_{\substack{0\le i\le r\\ S(p_i)-S(p_j)=\omega}}
(R_+)_{ij}\,\widetilde I_{p_i}.
\end{equation}
In particular, \(\dot\Delta_\omega^+\widetilde I_{p_j}=0\) unless
\(\omega\) is an action gap from \(p_j\) to another critical value.  Thus each
\(\dot\Delta_\omega^+\) is homogeneous of degree \(\omega\) with respect to
the critical-value grading,
\[
\dot\Delta_\omega^+(\mathcal M_A)\subset \mathcal M_{A+\omega},
\]
where \(\mathcal M_{A+\omega}=0\) if \(A+\omega\notin S(\Crit(S))\).
\end{prop}

\begin{proof}
By the thimble basis theorem and the Picard--Lefschetz formula, the two
one-sided thimble bases are related by the integral matrix
\eqref{eq:PL-matrix-R-plus}; see Theorem~\ref{thm:thimble-basis} and the
discussion of Picard--Lefschetz jumps above.  Since the exponential integral
is linear in the integration cycle, \eqref{eq:PL-matrix-R-plus} gives
\begin{equation}\label{eq:integral-PL-matrix-R-plus}
\mathbf I^{<}=\mathbf I^{>}R_+ .
\end{equation}

On the other hand, Proposition~\ref{prop:thimble-lateral-sum} identifies the
two one-sided thimble integrals with the corresponding lateral
Borel--Laplace sums:
\[
I_{\J_{p_i}^{>}}=\mathcal S^{>}\widetilde I_{p_i},
\qquad
I_{\J_{p_i}^{<}}=\mathcal S^{<}\widetilde I_{p_i}.
\]
Together with the Stokes relation \eqref{Stokesauto}, this implies, for each
\(j\),
\[
\mathcal S^{>}
\bigl(\mathfrak S^+\widetilde I_{p_j}\bigr)
=
\mathcal S^{<}\widetilde I_{p_j}
=
I_{\J_{p_j}^{<}}.
\]
Using \eqref{eq:integral-PL-matrix-R-plus}, the last term is
\[
I_{\J_{p_j}^{<}}
=
\sum_{i=0}^r (R_+)_{ij} I_{\J_{p_i}^{>}}
=
\mathcal S^{>}
\left(
\sum_{i=0}^r (R_+)_{ij}\widetilde I_{p_i}
\right).
\]
Hence
\[
\mathcal S^{>}
\left(
\mathfrak S^+\widetilde I_{p_j}
-
\sum_{i=0}^r (R_+)_{ij}\widetilde I_{p_i}
\right)=0.
\]
By the standard uniqueness of finite exponential asymptotic expansions in a
regular lateral sector, the lateral summation map is injective on the finite
critical-value graded space \(\mathcal M_S\).  Therefore
\[
\mathfrak S^+\widetilde I_{p_j}
=
\sum_{i=0}^r (R_+)_{ij}\widetilde I_{p_i},
\]
which proves \eqref{eq:Stokes-action-PL-matrix} and, in particular,
\(\mathfrak S^+(\mathcal M_S)\subset\mathcal M_S\). It remains only to take the homogeneous components with respect to the
critical-value grading.  By the convention fixed above,
\[
\dot\Delta_\omega^+
\left(
e^{-S(p_j)/\hbar}\widetilde\Psi_{p_j}
\right)
=
e^{-(S(p_j)+\omega)/\hbar}
\Delta_\omega^+\widetilde\Psi_{p_j},
\]
so \(\dot\Delta_\omega^+\) has degree \(\omega\).  Since
\[
\mathfrak S^+=\Id+\sum_{\omega\in d}\dot\Delta_\omega^+,
\]
the degree \(S(p_j)+\omega\) part of
\(\mathfrak S^+\widetilde I_{p_j}\) is precisely
\(\dot\Delta_\omega^+\widetilde I_{p_j}\) when \(\omega\neq0\).  Taking the
same homogeneous component in
\[
\mathfrak S^+\widetilde I_{p_j}
=
\sum_{i=0}^r (R_+)_{ij}\widetilde I_{p_i}
\]
gives exactly \eqref{eq:pointed-alien-action-gap-formula}.  If no critical
point \(p_i\) satisfies \(S(p_i)-S(p_j)=\omega\), the corresponding homogeneous
piece is zero.
\end{proof}

Thus the matrices obtained from alien calculus govern the same wall-crossing relations as those of the Lefschetz thimbles. 

\begin{rmk} The viewpoint used here is closely related to the analysis of remainders in the method of steepest descents developed by Berry and Howls \cite{BerryHowls1991Hyperasymptotics} and further studied by Boyd \cite{Boyd1993SteepestDescents}. After truncating the asymptotic expansion associated with one saddle point, the remainder can be represented in terms of integrals associated with adjacent saddle points. In this way, contributions from other critical points enter through the remainder term. In the present paper, the same interaction between critical points is encoded by the pointed alien action on the critical-value graded vector space. For the Airy function, explicit remainder representations of this type are given in \cite[Sec.~5, especially (53)--(57)]{Boyd1993SteepestDescents}; for the modified Bessel function, see \cite[Sec.~6, especially (67)--(70)]{Boyd1993SteepestDescents}. Related representations for the Gamma function were obtained by Boyd \cite{Boyd1994Gamma}; see also \cite{Nemes2015Gamma} for refinements leading to improved error bounds and exponentially improved asymptotic expansions. \end{rmk}

\smallskip
\subsection{The Hopf algebra of pointed alien operators}
\label{subsec:alien-Hopf-plus}

We now record the Hopf-algebraic structure carried by the pointed alien
operators \(\dot\Delta_w^+\), where \(w\in\Lambda_d\subset d\) and
\(\Lambda_d\) is the additive semigroup generated by the nonzero Borel
singularities on the Stokes ray \(d\). For more details on the algebraic structures behind alien operators and
mould calculus, we refer the reader to \cite{Ecalle81a} and
\cite{Sauzin2009SaddleNodeMould}. The Stokes automorphism in the
direction \(d\) is written as
\begin{equation}\label{eq:Stokes-as-dotDelta-plus-series}
\mathfrak S_d^+
=
\Id+\sum_{w\in\Lambda_d}\dot\Delta_w^+ .
\end{equation}

For a word
\[
\underline{w}=w_1\cdots w_r,
\qquad
w_i\in\Lambda_d,
\]
we write
\begin{equation}\label{eq:Sauzin-word-convention-dotDelta-plus}
\dot\Delta_{\underline{w}}^{+}
:=
\dot\Delta_{w_r}^{+}\cdots \dot\Delta_{w_1}^{+}.
\end{equation}
For the empty word \(\emptyset\), we set
\[
\dot\Delta_{\emptyset}^{+}:=\Id,
\qquad
|\emptyset|=0.
\]

Let \(\mathcal H_d\) be the completed noncommutative algebra formally generated
by the operators \(\dot\Delta_w^+\), with unit
\(\dot\Delta_\emptyset^+=\Id\), completed with respect to the total weight
\[
|\underline{w}|:=w_1+\cdots+w_r.
\]
The product is operator composition.  With the convention
\eqref{eq:Sauzin-word-convention-dotDelta-plus}, one has
\begin{equation*}\label{eq:dotDelta-plus-product-opposite-concat}
\dot\Delta_{\underline{w}}^{+}\,
\dot\Delta_{\underline{v}}^{+}
=
\dot\Delta_{\underline{v}\underline{w}}^{+}.
\end{equation*}

The coproduct is determined by the modified Leibniz rule for
\(\dot\Delta_w^+\).  Since the full Stokes automorphism is an algebra
automorphism, we have
\[
\mathfrak S_d^+(FG)=\mathfrak S_d^+(F)\mathfrak S_d^+(G).
\]
Using \eqref{eq:Stokes-as-dotDelta-plus-series} and comparing the homogeneous
part of total weight \(w\), we obtain
\begin{equation*}\label{eq:modified-Leibniz-dotDelta-plus}
\dot\Delta_w^+(FG)
=
\bigl(\dot\Delta_w^+F\bigr)G
+
F\bigl(\dot\Delta_w^+G\bigr)
+
\sum_{\substack{w_1+w_2=w\\ w_1,w_2\in\Lambda_d}}
\bigl(\dot\Delta_{w_1}^+F\bigr)
\bigl(\dot\Delta_{w_2}^+G\bigr).
\end{equation*}

This motivates the coproduct
\begin{equation}\label{eq:dotDelta-plus-coproduct}
\Delta_{\mathcal H}\bigl(\dot\Delta_w^+\bigr)
=
\dot\Delta_w^+\otimes1
+
1\otimes\dot\Delta_w^+
+
\sum_{\substack{w_1+w_2=w\\ w_1,w_2\in\Lambda_d}}
\dot\Delta_{w_1}^+\otimes \dot\Delta_{w_2}^+ .
\end{equation}
It is extended to \(\mathcal H_d\) as an algebra homomorphism:
\[
\Delta_{\mathcal H}(XY)
=
\Delta_{\mathcal H}(X)\Delta_{\mathcal H}(Y).
\]
Consequently the Stokes automorphism is group-like:
\begin{equation*}\label{eq:Stokes-group-like-dotDelta-plus}
\Delta_{\mathcal H}(\mathfrak S_d^+)
=
\mathfrak S_d^+\otimes \mathfrak S_d^+ .
\end{equation*}

The counit is defined by
\begin{equation*}\label{eq:dotDelta-plus-counit}
\varepsilon(\dot\Delta_{\emptyset}^{+})=1,
\qquad
\varepsilon(\dot\Delta_{\underline{w}}^{+})=0
\quad
\text{for every nonempty word } \underline{w}.
\end{equation*}
The antipode is the map
\[
\mathsf S_{\mathcal H}:\mathcal H_d\to\mathcal H_d
\]
defined on homogeneous generators by
\begin{equation*}\label{eq:dotDelta-plus-antipode-generator}
\mathsf S_{\mathcal H}(\dot\Delta_w^+)
=
\sum_{r\ge1}(-1)^r
\sum_{\substack{w_1+\cdots+w_r=w\\
w_i\in\Lambda_d}}
\dot\Delta_{w_1}^+\dot\Delta_{w_2}^+\cdots\dot\Delta_{w_r}^+ .
\end{equation*}
Note that $0\notin\Lambda_d$. The antipode is extended anti-multiplicatively:
\begin{equation*}\label{eq:dotDelta-plus-antipode-word}
\mathsf S_{\mathcal H}
\left(
\dot\Delta_{\underline{w}}^{+}
\right)
=
\mathsf S_{\mathcal H}
\left(
\dot\Delta_{w_r}^{+}\cdots\dot\Delta_{w_1}^{+}
\right)
=
\mathsf S_{\mathcal H}(\dot\Delta_{w_1}^{+})
\cdots
\mathsf S_{\mathcal H}(\dot\Delta_{w_r}^{+}).
\end{equation*}
Equivalently, \(\mathsf S_{\mathcal H}\) is characterized by the convolution
inverse identity
\[
m(\mathsf S_{\mathcal H}\otimes\Id)\Delta_{\mathcal H}(X)
=
\varepsilon(X)1
=
m(\Id\otimes\mathsf S_{\mathcal H})\Delta_{\mathcal H}(X).
\]
Since \(\mathfrak S_d^+\) is group-like, the antipode sends it to its inverse:
\begin{equation*}\label{eq:antipode-Stokes-inverse}
\mathsf S_{\mathcal H}(\mathfrak S_d^+)
=
(\mathfrak S_d^+)^{-1}
=
\mathfrak S_d^- .
\end{equation*}
Thus, on the level of Stokes matrices, the antipode corresponds to inverse
wall-crossing.

Finally, the logarithmic pointed alien operator is obtained by taking the
logarithm of the Stokes automorphism:
\begin{equation}\label{eq:dotDelta-log-Stokes}
\dot\Delta_d
:=
\log \mathfrak S_d^+
=
\log\left(\Id+\sum_{w\in\Lambda_d}\dot\Delta_w^+\right).
\end{equation}
We write
\[
\dot\Delta_d=\sum_{w\in\Lambda_d}\dot\Delta_w
\]
for its decomposition into homogeneous weights.  Comparing the homogeneous
part of weight \(w\), one obtains the so-called alien derivations
\begin{equation}\label{eq:dotDelta-from-dotDelta-plus}
\dot\Delta_w
=
\sum_{r\ge1}\frac{(-1)^{r-1}}{r}
\sum_{\substack{w_1+\cdots+w_r=w\\ w_i\in\Lambda_d}}
\dot\Delta_{w_1}^+\dot\Delta_{w_2}^+\cdots\dot\Delta_{w_r}^+ .
\end{equation}

\begin{rmk}[Gamma model as a guide to the Hopf dictionary]
In the Gamma model (see Section \ref{sec:gamma-model} for details), after the rotation used below, the Stokes ray is
\(d=\R_{>0}\), and the action semigroup is
\(\Lambda_d=2\pi\mathbb N^\ast\).  The pointed alien operators act by shifts,
\[
\dot\Delta_{2\pi k}^+\widetilde I_n=\widetilde I_{n-k}.
\]
Thus the Hopf operations admit the following Picard--Lefschetz interpretation
in this example.

\begin{table}[H]
\centering
\renewcommand{\arraystretch}{1.45}
\begin{tabular}{p{0.45\textwidth}|p{0.45\textwidth}}
\hline
\textbf{Picard--Lefschetz side} & \textbf{Alien Hopf side} \\
\hline
Broken chain
\(p_{n-2}\to p_{n-1}\to p_n\).
&
Product of pointed alien operators:
\(\displaystyle
(\dot\Delta_{2\pi}^+)^2\widetilde I_n
=
\dot\Delta_{4\pi}^+\widetilde I_n.
\)
\\
\hline
Product thimble \(\J_a\times\J_b\), where the total action difference splits
between the two factors.
&
Coproduct / modified Leibniz rule:
\(\displaystyle
\Delta_{\mathcal H}(\dot\Delta_w^+)
=
\dot\Delta_w^+\otimes1
+
1\otimes\dot\Delta_w^+
+
\sum_{\substack{w_1+w_2=w\\ w_1,w_2\in\Lambda_d}}
\dot\Delta_{w_1}^+\otimes\dot\Delta_{w_2}^+.
\)
\\
\hline
Inverse wall-crossing matrix.
&
Antipode:
\(\displaystyle
\mathsf S_{\mathcal H}(\mathfrak S^+)
=
(\mathfrak S^+)^{-1}
=
\mathfrak S^-.
\)
\\
\hline
\end{tabular}
\caption{The Gamma model interpretation of the Hopf operations.}
\label{table:gamma-hopf-dictionary}
\end{table}
\end{rmk}

\medskip
\section{The Airy model}

We specialize to the Airy action
\begin{equation*}
S(z)=\frac{z^3}{3}-z,
\qquad z\in\C,
\end{equation*}
and the corresponding exponential integrals
\begin{equation*}
I_\Gamma(\hbar)=\int_\Gamma e^{-S(z)/h}\,\dd z,
\end{equation*}
for some appropriate integral cycle $\Gamma$. 

\begin{rmk}
We call this the Airy model because it is the specialization at \(x=1\) of
the standard Airy phase
\[
S_x(z)=\frac{z^3}{3}-xz.
\]
If \(x\) is kept as a parameter, then for suitable rapidly decaying cycles, the integral
\[
I_\Gamma(x,\hbar)
=
\int_\Gamma
\exp\left(-\frac{1}{\hbar}\left(\frac{z^3}{3}-xz\right)\right)\,\dd z
\]
satisfies the semiclassical Airy equation
\[
\hbar^2\,\frac{\dd^2 I_\Gamma}{\dd x^2}=x\,I_\Gamma.
\]
\end{rmk}

The critical points
$p_+=1$ and $p_-=-1$ are non-degenerate and the critical values are
\begin{equation*}
S(p_+)=-\frac23,
\qquad
S(p_-)=\frac23.
\end{equation*}
Choosing the standard hermritian metric, the negative gradient flow is given by
\begin{equation}\label{eq:airy-flow-complex}
\dot{z}=-\overline{e^{-i\theta}(z^2-1)}.
\end{equation}

Since
\begin{equation*}
S(p_-)-S(p_+)=\frac43,
\end{equation*}
the Stokes  phases are
\begin{equation*}
\theta_*=0\pmod\pi.
\end{equation*}

\smallskip
\subsection{Lefschetz thimbles near the Stokes phase}
For a fixed regular phase \(\theta\), that is, $\theta \notin \pi\Z$, the shape of the stable and unstable manifolds is determined by two pieces of data: the fact that every flow line is contained in a level set of 
$$
G_\theta=\im\bigl(e^{-i\theta}S)
$$
and the local flow directions at each critical point. 
\begin{itemize}
    \item The stable and unstable manifolds through \(p_-\) and \(p_+\) lie in the level sets
\begin{equation*}\label{eq:airy-Gtheta-levels}
G_\theta(p_-)=-\frac23\sin\theta
\qquad\text{and}\qquad
G_\theta(p_+)=\frac23\sin\theta.
\end{equation*}
    \item Linearizing \eqref{eq:airy-flow-complex} at the two critical points gives
\begin{equation*}
\dot\xi=
2e^{i\theta}\overline{\xi}
\qquad\text{at }p_-=-1;
\qquad
\dot\xi=
-2e^{i\theta}\overline{\xi}
\qquad\text{at }p_+=1.
\end{equation*}
Hence locally the flow directions are 
\begin{equation*}\label{eq:airy-linearization-summary}
\begin{aligned}
&\text{at }p_-=-1: &&\arg\xi=\frac\theta2 \pmod\pi \text{ unstable},
&&\quad\arg\xi=\frac\theta2+\frac\pi2 \pmod\pi \text{ stable},
\\
&\text{at }p_+=1:&&\arg\xi=\frac\theta2 \pmod\pi \text{ stable},
&&\quad\arg\xi=\frac\theta2+\frac\pi2 \pmod\pi \text{ unstable}.
\end{aligned}
\end{equation*}
\end{itemize}

For example, consider sufficient small $\theta>0$, the configuration of 
\[
W^s_\theta(p_-)\cup W^u_\theta(p_-)
\qquad\text{and}\qquad
W^s_\theta(p_+)\cup W^u_\theta(p_+).
\]
is shown in the right panel of Figure~\ref{fig:airy-three-phases}.

\smallskip
\subsection{Connecting trajectories and the Picard--Lefschetz jump formula}
In real coordinates \(z=u+iv\), the flow equation reads as follows:
\begin{equation}\label{eq:airy-flow-real}
\begin{aligned}
u_s&=-(u^2-v^2-1)\cos\theta-2uv\sin\theta,
\\
v_s&=-(u^2-v^2-1)\sin\theta+2uv\cos\theta,
\end{aligned}
\end{equation}
and the imaginary part $G_{\theta}(z)$ in the real coordinates $(u,v)$ reads as 
\begin{equation*}
G_\theta(u,v)
=\Bigl(u^2v-\frac{v^3}{3}-v\Bigr)\cos\theta
-\Bigl(\frac{u^3}{3}-uv^2-u\Bigr)\sin\theta.
\end{equation*}
At the Stokes phase $\theta_*=0$, 
\begin{equation*}
G_0^{-1}(0)
=
\{v=0\}\cup\left\{u^2=1+\frac{v^2}{3}\right\}.
\end{equation*}

\noindent(i) First we consider the real axis \(v=0\). The flow equation \eqref{eq:airy-flow-real} reduces to
\[
u_s=1-u^2.
\]
Then we have 
\begin{itemize}
\item for \(-1<u<1\), one has \(u_s>0\), and the forward flow converges to \(u=1\);
\item for \(u>1\), one has \(u_s<0\), and the forward flow also converges to \(u=1\);
\item for \(u<-1\), one has \(u_s<0\), so the forward flow
escapes to \(-\infty\), while the backward flow converges to \(u=-1\). 
\end{itemize}
Therefore
\[
W_0^s(p_+)\cap \{v=0\}
=
\{(u,v):v=0,\ u>-1\},
\]
and
\[
W_0^u(p_-)\cap \{v=0\}
=
\{(u,v):v=0,\ u<1\}.
\]
In particular, the open interval
\[
\{(u,0):-1<u<1\}=W_0^u(p_-)\cap W_0^s(p_+)\cap \{v=0\}
\]
corresponds to the unique real connecting trajectory from \(z_-\) to \(z_+\), see Proposition \ref{prop:airy-exact-trajectory} below.

\smallskip
\noindent (ii) Second we consider the hyperbola
\[
u^2-\frac{v^2}{3}=1.
\]
On the right branch, write
\[
u=\cosh r,\qquad v=\pm \sqrt3\sinh r,\qquad r\ge 0.
\]
The flow equation is
\[
u_s=1-u^2+v^2
=
1-\cosh^2 r+3\sinh^2 r
=
2\sinh^2 r.
\]
We obtain the flow of $r$ by
\[
r_s=2\sinh r.
\]
Thus \(r\to 0\) as \(s\to-\infty\), while \(r\to+\infty\) as \(s\to+\infty\).
Therefore the right branch of the hyperbola is the unstable manifold of \(p_+\):
\[
W_0^u(p_+)
=
\left\{
(u,v):u^2-\frac{v^2}{3}=1,\ u\ge 1
\right\}.
\]
Similarly, the left branch of the hyperbola gives the stable
manifold of \(p_-\):
\[
W_0^s(p_-)
=
\left\{
(u,v):u^2-\frac{v^2}{3}=1,\ u\le -1
\right\}.
\]

Combining the preceding analysis, the configuration of $W_0^{u,s}(p_{\pm})$ is shown in the middle panel of Figure~\ref{fig:airy-three-phases}; the blue segment in that panel is precisely the connecting branch along the real axis.

\begin{prop}\label{prop:airy-exact-trajectory}
At \(\theta_*=0\), up to the $s$-translation, the connecting trajectory from \(p_-=-1\) to \(p_+=1\) is
\begin{equation*}
z(s)=\tanh(s-s_0),
\qquad s_0\in\R.
\end{equation*}
In particular,
\begin{equation*}
z(s)\to -1\quad (s\to-\infty),
\qquad
z(s)\to 1\quad (s\to+\infty).
\end{equation*}
\end{prop}

\begin{proof}
Since
\begin{equation*}
F_0(p_-)=\frac23>-\frac23=F_0(p_+),\qquad G_0(p_-)=G_0(p_+)=0,
\end{equation*}
the possible trajectory is from $p_-$ to $p_+$, and every such trajectory is contained in the real axis $[-1,1]$. Restricting \eqref{eq:airy-flow-real} to \(v=0\), we obtain
\begin{equation*}\label{eq:airy-real-axis-flow}
\dot u=1-u^2,\quad \text{equivalently, }~~ \frac{\dd u}{1-u^2}=\dd s.
\end{equation*}
Integrating gives
\begin{equation*}
\operatorname{arctanh}(u)=s-s_0,
\end{equation*}
and therefore
\begin{equation*}
u(s)=\tanh(s-s_0).
\end{equation*}
This proves both existence and uniqueness up to translation in \(s\).
\end{proof}

For \(|\theta|\ll1\), the branch of \(W^u_\theta(p_-)\) (resp. $W_{\theta}^s(p_+)$) that deforms from the real interval \((-1,1)\) lies on
\begin{equation*}\label{eq:airy-level-minus}
G_\theta(u,v)=-\frac23\sin\theta, \quad (\text{resp. }\, G_\theta(u,v)=\frac23\sin\theta).
\end{equation*}
Writing these two branches as graphs
\begin{equation*}
v=v_-(u;\theta),
\qquad
v=v_+(u;\theta),
\qquad
u\in(-1,1),
\end{equation*}
we compute, to first order in \(\theta\),
\begin{equation*}\label{eq:airy-small-branch-minus}
v_-(u;\theta)
=
\theta\,\frac{S(u)-\frac23}{u^2-1}
+
O(\theta^2),
\end{equation*}
and
\begin{equation*}\label{eq:airy-small-branch-plus}
v_+(u;\theta)
=
\theta\,\frac{S(u)+\frac23}{u^2-1}
+
O(\theta^2).
\end{equation*}
Using
\begin{equation*}
S(u)-\frac23=\frac{(u-2)(u+1)^2}{3},
\qquad
S(u)+\frac23=\frac{(u-1)^2(u+2)}{3},
\end{equation*}
we can read off the signs on \((-1,1)\):
\begin{equation*}\label{eq:airy-sign-conclusion}
\begin{cases}
\theta<0:\quad v_-(u;\theta)<0,\quad v_+(u;\theta)>0,\\[4pt]
\theta>0:\quad v_-(u;\theta)>0,\quad v_+(u;\theta)<0.
\end{cases}
\end{equation*}
Thus for \(\theta\neq0\) sufficiently small, the connecting trajectory disappears, the unstable branch from \(p_-\) and the stable branch into \(p_+\) move to opposite sides of the real axis.  The complementary branches are smooth deformations of the outer branches $W_0^{u,s}(p_{\pm})$, and together they determine the full stable and unstable manifolds for fixed regular \(\theta\). See Figure \ref{fig:airy-three-phases}. 

\begin{figure}[H]
\centering

\tikzset{
  flowstable/.style={
    black,
    line width=1.1pt,
    -{Latex[length=2.2mm]}
  },
  flowunstable/.style={
    green!60!black,
    line width=1.3pt,
    -{Latex[length=2.2mm]}
  },
  flowtraj/.style={
    blue!75!black,
    line width=2.5pt,
    -{Latex[length=2.5mm]}
  }
}

\colorlet{sectorred}{red!18}
\begin{minipage}[t]{0.32\textwidth}
\centering
\[
\theta<0
\]
\resizebox{\linewidth}{!}{
\begin{tikzpicture}[scale=1.0, >=Latex]
\def\R{4.9}
\def\th{-18}
\fill[sectorred] (0,0) --
    ({\th/3-30}:\R) arc ({\th/3-30}:{\th/3+30}:\R) -- cycle;
  \fill[sectorred] (0,0) --
    ({\th/3+90}:\R) arc ({\th/3+90}:{\th/3+150}:\R) -- cycle;
  \fill[sectorred] (0,0) --
    ({\th/3+210}:\R) arc ({\th/3+210}:{\th/3+270}:\R) -- cycle;
\draw[->, thin] (-5.4,0) -- (5.4,0) node[right] {$u$};
\draw[->, thin] (0,-4.8) -- (0,4.8) node[above] {$v$};
\filldraw[black] (-1,0) circle (1.5pt);
\filldraw[black] ( 1,0) circle (1.5pt);
\node[below left=1pt] at (-1,0) {$p_{-}$};
\node[below right=1pt] at (1,0) {$p_{+}$};
\draw[flowstable]
    ({117}:\R) .. controls (-1.15,2.85) and (-1,0.95) .. (-1,0);
\draw[flowstable]
    ({245}:\R) .. controls (-1.25,-2.95) and (-1.05,-1.30) .. (-1,0);
\draw[flowstable]
    ({-6}:\R) .. controls (3.95,-0.38) and (2.05,-0.15) .. (1,0);
\draw[flowstable]
    ({114}:\R) .. controls (-0.7,3) and (-0.6,1) .. (-0.5,0.5) .. controls (-0.5,0.15) and (0.2,0.1) .. (1,0);

\draw[flowunstable]
    (-1,0) .. controls (-2.05,0.15) and (-3.95,0.38) .. ({174}:\R);
\draw[flowunstable]
    (-1,0) .. controls (-0.2,-0.1) and (0.5,-0.15)..(0.5,-0.5) .. controls (0.6,-1) and (0.7,-3) .. ({294}:\R);
\draw[flowunstable]
    (1,0) .. controls (0.95,0.95) and (1.15,2.85) .. ({65}:\R);
\draw[flowunstable]
    (1,0) .. controls (1,-0.95) and (1.15,-2.85) .. ({297}:\R);
\end{tikzpicture}
}
\end{minipage}
\hfill
\begin{minipage}[t]{0.32\textwidth}
\centering
\[
\theta=0
\]
\resizebox{\linewidth}{!}{
\begin{tikzpicture}[scale=1.15, >=Latex]
\def\R{4.8}
\fill[sectorred] (0,0) -- (-30:\R) arc (-30:30:\R) -- cycle;
\fill[sectorred] (0,0) -- (90:\R) arc (90:150:\R) -- cycle;
\fill[sectorred] (0,0) -- (210:\R) arc (210:270:\R) -- cycle;
\draw[->, thin] (-5.4,0) -- (5.4,0) node[right] {$u$};
\draw[->, thin] (0,-4.8) -- (0,4.8) node[above] {$v$};

\draw[flowstable, domain=4.05:0, samples=120, smooth]
    plot ({-sqrt(1 + (\x *\x)/3)}, {\x});
\draw[flowstable, domain=-4.05:0, samples=120, smooth]
    plot ({-sqrt(1 + (\x*\x)/3)}, {\x});
\draw[flowstable] (5.1,0) -- (1,0);

\draw[flowunstable] (-1,0) -- (-5.1,0);
\draw[flowunstable, domain=0:4.05, samples=120, smooth]
    plot ({sqrt(1 + (\x*\x)/3)}, {\x});
\draw[flowunstable, domain=0:-4.05, samples=120, smooth]
    plot ({sqrt(1 + (\x*\x)/3)}, {\x});

\draw[flowtraj] (-1,0) -- (1,0);
\filldraw[black] (-1,0) circle (1.5pt);
\filldraw[black] ( 1,0) circle (1.5pt);
\node[below left=1pt] at (-1,0) {$p_{-}$};
\node[below right=1pt] at (1,0) {$p_{+}$};
\end{tikzpicture}
}
\end{minipage}
\hfill
\begin{minipage}[t]{0.32\textwidth}
\centering
\[
\theta>0
\]
\resizebox{\linewidth}{!}{
\begin{tikzpicture}[scale=1.15, >=Latex]
\def\R{4.9}
\def\th{18}
\fill[sectorred] (0,0) --
    ({\th/3-30}:\R) arc ({\th/3-30}:{\th/3+30}:\R) -- cycle;
\fill[sectorred] (0,0) --
    ({\th/3+90}:\R) arc ({\th/3+90}:{\th/3+150}:\R) -- cycle;
\fill[sectorred] (0,0) --
    ({\th/3+210}:\R) arc ({\th/3+210}:{\th/3+270}:\R) -- cycle;

\draw[->, thin] (-5.4,0) -- (5.4,0) node[right] {$u$};
\draw[->, thin] (0,-4.8) -- (0,4.8) node[above] {$v$};

\filldraw[black] (-1,0) circle (1.5pt);
\filldraw[black] ( 1,0) circle (1.5pt);
\node[below left=1pt] at (-1,0) {$p_{-}$};
\node[below right=1pt] at (1,0) {$p_{+}$};

\draw[flowstable]
    ({117}:\R) .. controls (-1.15,2.85) and (-1,0.95) .. (-1,0);
\draw[flowstable]
    ({244}:\R) .. controls (-1.25,-2.95) and (-1.05,-1.30) .. (-1,0);
\draw[flowstable]
    ({6}:\R) .. controls (3.95,0.38) and (2.05,0.15) .. (1,0);
\draw[flowstable]
    ({-114}:\R) .. controls (-0.7,-3) and (-0.6,-1) .. (-0.5,-0.5) .. controls (-0.5,-0.15) and (0.2,-0.1) .. (1,0);

\draw[flowunstable]
    (-1,0) .. controls (-2.05,-0.15) and (-3.95,-0.38) .. ({-174}:\R);
\draw[flowunstable]
    (-1,0) .. controls (-0.2,0.1) and (0.5,0.15)..(0.5,0.5) .. controls (0.6,1) and (0.7,3) .. ({-294}:\R);
\draw[flowunstable]
    (1,0) .. controls (0.95,0.95) and (1.15,2.85) .. ({65}:\R);
\draw[flowunstable]
    (1,0) .. controls (1,-0.95) and (1.15,-2.85) .. ({297}:\R);
\end{tikzpicture}
}
\end{minipage}

\vspace{0.8em}

\begin{center}
\begin{tikzpicture}[>=Latex]
\draw[flowstable] (0,0) -- (1.6,0);
\node[right] at (1.8,0) {stable};

\draw[flowunstable] (4.2,0) -- (5.8,0);
\node[right] at (6.0,0) {unstable};
\draw[flowtraj] (8.9,0) -- (10.5,0);
\node[right] at (10.7,0) {connecting trajectory};
\end{tikzpicture}
\end{center}

\caption{The stable and unstable manifolds in Airy type: The shaded pink regions indicate the directions in which \(F_\theta\to +\infty\); accordingly, the stable manifolds (black) end in these regions, whereas the unstable manifolds (green) do not.}
\label{fig:airy-three-phases}
\end{figure}
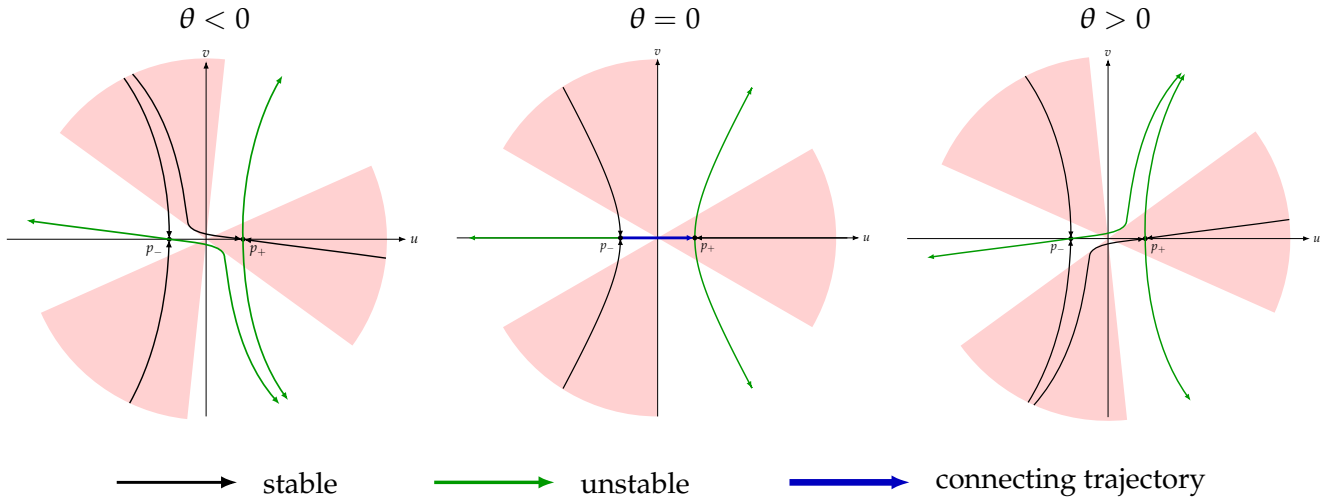

At the Stokes phase \(\theta_*=0\), one may read off the Picard-Lefschetz jump
\begin{equation*}\label{equ:Airyjump}
\bigl([\J_-^>] \ \ [\J_+^>]\bigr)
=
\bigl([\J_-^<] \ \ [\J_+^<]\bigr)
\begin{pmatrix}
1&1\\
0&1
\end{pmatrix}.
\end{equation*}
In particular, the number of connecting trajectories from $p_-$ to $p_+$ is $1$. The corresponding thimble jump is illustrated in  Figure~\ref{fig:airy-JK-two-panels}.

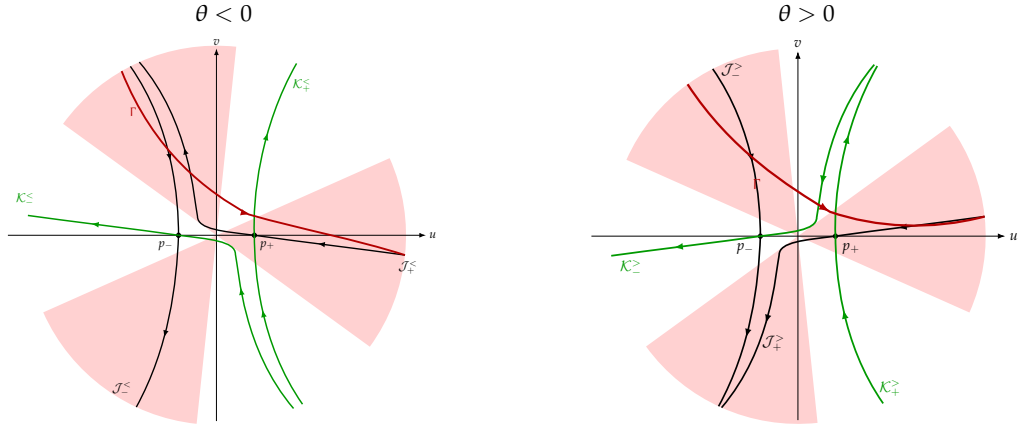
\begin{figure}[H]
\centering

\tikzset{
  Jflow/.style={
    black,
    line width=1.25pt,
    postaction={decorate},
    decoration={
      markings,
      mark=at position 0.58 with {\arrow{latex}}
    }
  },
  Kflow/.style={
    green!60!black,
    line width=1.35pt,
    postaction={decorate},
    decoration={
      markings,
      mark=at position 0.58 with {\arrow{latex}}
    }
  },
  GammaAi/.style={
    red!70!black,
    line width=1.7pt,
    postaction={decorate},
    decoration={
      markings,
      mark=at position 0.55 with {\arrow{latex}}
    }
  }
}

\colorlet{sectorred}{red!18}
\scalebox{0.8}{%
\begin{minipage}{\textwidth}
\begin{minipage}[t]{0.45\textwidth}
\centering
\[
\theta<0
\]
\resizebox{0.92\linewidth}{!}{
\begin{tikzpicture}[scale=1.2, >=Latex]

  \def\R{5.0}
  \def\th{-18}
\fill[sectorred] (0,0) --
    ({\th/3-30}:\R) arc ({\th/3-30}:{\th/3+30}:\R) -- cycle;
\fill[sectorred] (0,0) --
    ({\th/3+90}:\R) arc ({\th/3+90}:{\th/3+150}:\R) -- cycle;
\fill[sectorred] (0,0) --
    ({\th/3+210}:\R) arc ({\th/3+210}:{\th/3+270}:\R) -- cycle;

\draw[->, thin] (-5.5,0) -- (5.5,0) node[right] {$u$};
\draw[->, thin] (0,-4.9) -- (0,4.9) node[above] {$v$};

\filldraw[black] (-1,0) circle (1.6pt);
\filldraw[black] ( 1,0) circle (1.6pt);
\node[below left=1pt] at (-1,0) {$p_-$};
\node[below right=1pt] at (1,0) {$p_+$};

\draw[Jflow]
   ({117}:\R) .. controls (-1.15,2.85) and (-1,0.95) .. (-1,0);
\node[black] at (-2.5,-4.05) {$\J_-^<$};
\draw[Jflow]
   ({-6}:\R) .. controls (3.95,-0.38) and (2.05,-0.15) .. (1,0);
\node[black] at (5.05,-0.85) {$\J_+^<$};
\draw[Jflow]
  (1,0)   .. controls (0.2,0.1)  and (-0.5,0.15) .. (-0.5,0.5) .. controls (-0.6,1) and (-0.7,3) ..({114}:\R) ;
\draw[Jflow]
    (-1,0) .. controls (-1.05,-1.30) and (-1.25,-2.95) .. ({245}:\R);

\node[green!60!black] at (-5.05,0.95) {$\Kk_-^<$};
\node[green!60!black] at (2.25,3.95) {$\Kk_+^<$};
\draw[Kflow]
 ({294}:\R) .. controls (0.7,-3)  and (0.6,-1)..(0.5,-0.5) .. controls (0.5,-0.15) and  (-0.2,-0.1)..(-1,0);
\draw[Kflow]
   (-1,0) .. controls (-2.05,0.15) and (-3.95,0.38) .. ({174}:\R);
\draw[Kflow]
   ({297}:\R) .. controls (1.15,-2.85)  and (1,-0.95).. (1,0);
\draw[Kflow]
    (1,0) .. controls (0.95,0.95) and (1.15,2.85) .. ({65}:\R);
\draw[GammaAi]
    ({120}:\R) .. controls (-1.65,2.20) and (-0.25,1.15) ..
    (0.85,0.55) ..
    controls (2.25,0.15) and (3.65,-0.10) .. ({-6}:\R);
  \node[red!70!black] at (-2.2,3.25) {$\Gamma$};
\end{tikzpicture}
}
\end{minipage}
\hfill
\begin{minipage}[t]{0.45\textwidth}
\centering
\[
\theta>0
\]
\resizebox{0.92\linewidth}{!}{
\begin{tikzpicture}[scale=1.0, >=Latex]

  \def\R{5.0}
  \def\th{18}

\fill[sectorred] (0,0) --
    ({\th/3-30}:\R) arc ({\th/3-30}:{\th/3+30}:\R) -- cycle;
\fill[sectorred] (0,0) --
    ({\th/3+90}:\R) arc ({\th/3+90}:{\th/3+150}:\R) -- cycle;
\fill[sectorred] (0,0) --
    ({\th/3+210}:\R) arc ({\th/3+210}:{\th/3+270}:\R) -- cycle;

\draw[->, thin] (-5.5,0) -- (5.5,0) node[right] {$u$};
\draw[->, thin] (0,-4.9) -- (0,4.9) node[above] {$v$};

\filldraw[black] (-1,0) circle (1.6pt);
\filldraw[black] ( 1,0) circle (1.6pt);
\node[below left=1pt] at (-1,0) {$p_-$};
\node[below right=1pt] at (1,0) {$p_+$};
\node[black] at (-1.75,4.4) {$\J_-^>$};
 
\node[black] at (-0.65,-2.85) {$\J_+^>$};
\draw[Jflow]
   ({117}:\R) .. controls (-1.15,2.85) and (-1,0.95) .. (-1,0);
\draw[Jflow]
   ({6}:\R) .. controls (3.95,0.38) and (2.05,0.15) .. (1,0);
\draw[Jflow]
  (1,0)   .. controls (0.2,-0.1)  and (-0.5,-0.15) .. (-0.5,-0.5) .. controls (-0.6,-1) and (-0.7,-3) ..({-114}:\R) ;
\draw[Jflow]
    (-1,0) .. controls (-1.05,-1.30) and (-1.25,-2.95) .. ({245}:\R);
\draw[Kflow]
 ({-294}:\R) .. controls (0.7,3)  and (0.6,1)..(0.5,0.5) .. controls (0.5,0.15) and  (-0.2,0.1)..(-1,0);
\draw[Kflow]
   (-1,0) .. controls (-2.05,-0.15) and (-3.95,-0.38) .. ({-174}:\R);
\draw[Kflow]
   ({297}:\R) .. controls (1.15,-2.85)  and (1,-0.95).. (1,0);
\draw[Kflow]
    (1,0) .. controls (0.95,0.95) and (1.15,2.85) .. ({65}:\R);
\node[green!60!black] at (-4.45,-0.85) {$\Kk_-^>$};
\node[green!60!black] at (2.45,-4.05) {$\Kk_+^>$};
\draw[GammaAi]
    ({126}:\R) .. controls (-1.80,2.30) and (-0.30,1.35) ..
    (0.85,0.65) ..
    controls (2.10,0.22) and (3.55,0.15) .. ({6}:\R);
  \node[red!70!black] at (-1.10,1.40) {$\Gamma$};
\end{tikzpicture}
}
\end{minipage}
\end{minipage}
}
\caption{The local orientations are chosen so that the direction of \(\J_p^{>}\) (resp.\ \(\Kk_p^{>}\)) agrees with that of \(\J_p^{<}\) (resp.\ \(\Kk_p^{<}\)) near \(p\), and so that
\[
\langle [\J_p^\theta],[\Kk_q^\theta]\rangle=\delta_{pq}.
\]
This is compatible with the left-handed convention. The red curve \(\Gamma\) indicates a typical choice of contour for the Airy integral.
For every regular phase \(\theta\), \(\Gamma\) admits a Lefschetz-thimble decomposition
\[ [\Gamma] = \sum\limits_p\langle[\Gamma],[\Kk_p^\theta] \rangle [\J^\theta_p].\] }
\label{fig:airy-JK-two-panels}
\end{figure}

\smallskip
\subsection{Resurgence and the Stokes matrix}
We now run the resurgent analysis for the Airy type. For the two generic phases adjacent to \(\theta_*=0\), the four thimble integrals
\begin{equation*}
I_{\J_-^{<}}(\hbar),\quad I_{\J_+^{<}}(\hbar),\quad I_{\J_-^{>}}(\hbar),\quad I_{\J_+^{>}}(\hbar)
\end{equation*}
admit saddle-point expansions for \(\hbar\to0\) in the corresponding adjacent sectors. Since \(\J_+^{<}\) and \(\J_+^{>}\) have the same local germ at \(p_+=1\), and likewise \(\J_-^{<}\) and \(\J_-^{>}\) at \(p_-=-1\), there are only two formal objects:
\begin{equation*}
I_{\J_+^{{\gtrless}}}(\hbar)\sim\widetilde I_+(\hbar),\qquad
I_{\J_-^{{\gtrless}}}(\hbar)\sim \widetilde I_-(\hbar),
\end{equation*}
with
\begin{equation*}
\widetilde I_+(\hbar)=-\sqrt{\pi \hbar}\,e^{2/(3\hbar)}\,\widetilde\phi_+(\hbar),\qquad
\widetilde I_-(\hbar)=-i\sqrt{\pi \hbar}\,e^{-2/(3\hbar)}\,\widetilde\phi_-(\hbar),
\end{equation*}
where
\begin{equation*}
\widetilde\phi_+(\hbar)=\sum_{m\ge0}c_m\hbar^m,\qquad \widetilde\phi_-(\hbar)=\sum_{m\ge0}(-1)^m c_m\hbar^m,\qquad
c_m=\frac{(6m)!}{576^m(2m)!(3m)!}.
\end{equation*}
These formulas are obtained by the standard saddle-point substitution \(z=1-\sqrt{\hbar}\,t\) near \(p_+=1\) and \(z=-1-i\sqrt{\hbar}\,t\) near \(p_-=-1\), followed by expansion of the cubic correction. We now pass to the Borel plane. The hypergeometric Borel transform written below is the standard explicit form of the Airy saddle expansion; see, for example, the saddle-integral computations in \cite{CNP1993,DelabaerePham1999} and the review-style discussion in \cite{Dorigoni2014}. For the reduced formal series
\begin{equation*}\label{equ:ordinaryBorel}
\widetilde\phi(\hbar)=\sum_{m\ge0}a_m \hbar^m,
\end{equation*}
we use the ordinary Borel transform
\footnote{In this subsection we use the ordinary Borel transform of the
reduced series \(\widetilde\phi_\pm\), whereas in Section~\ref{sec:resurgence} the natural
geometric convention is the shifted Borel transform of
\(\sqrt{\hbar}\,\widetilde\phi_\pm\).  This does not change the pointed
alien relation for the full formal objects.  Indeed, multiplication by
\(\sqrt{\hbar}\) corresponds in the Borel plane to a universal fractional
convolution operator; more explicitly, if
\(\widehat\phi=\mathcal B_0\widetilde\phi\), then
\[
\mathcal B_{1/2}(\sqrt{\hbar}\,\widetilde\phi)
=
\partial_\xi\left(
\frac{1}{\Gamma(\frac12)}
\int_0^\xi
(\xi-\eta)^{-1/2}\widehat\phi(\eta)\,\dd\eta
\right).
\]
This universal factor has no singularity at any nonzero action difference
\(\omega\).  Hence the local variation at \(\omega\) passes through it; 
equivalently, alien derivatives are derivations for the product/convolution
algebra and
\[
\Delta_\omega^+(\sqrt{\hbar}\,\widetilde\phi)
=
\sqrt{\hbar}\,\Delta_\omega^+\widetilde\phi
\qquad (\omega\neq0).
\]
After restoring the outside Gaussian and exponential factors, the ordinary
Borel calculation for the reduced series therefore gives the same pointed
alien relation for the full objects \(\widetilde I_\pm\).}
\begin{equation*}
\widehat\phi(\xi):=\sum_{m\ge0}a_m\frac{\xi^m}{m!}.
\end{equation*}
Then
\begin{equation*}
\widehat\phi_+(\xi)
=
{}_2F_1\!\left(\frac16,\frac56;1;\frac{3\xi}{4}\right),
\qquad
\widehat\phi_-(\xi)
=
{}_2F_1\!\left(\frac16,\frac56;1;-\frac{3\xi}{4}\right).
\end{equation*}
Hence \(\widehat\phi_+\) has its first singularity at \(\xi=\frac43\), while \(\widehat\phi_-\) has its first singularity at \(\xi=-\frac43\). These numbers are exactly the differences of the Airy critical values: 
$$
S(p_-)-S(p_+)=\frac43\quad \text{and} \quad S(p_+)-S(p_-)=-\frac43.  
$$
The singularities of these hypergeometric germs are logarithmic. In particular, their local continuations have only integrable singularities in the sense of Section~\ref{sec:resurgence}. Thus 
\[-\sqrt{\pi \hbar} \widetilde{\phi}_+, \ -i\sqrt{\pi \hbar} \widetilde{\phi}_- \in \widetilde{\mathcal{R}}^{int}.\]
More precisely, there exist holomorphic germs \(H_\pm\) such that
\begin{equation*}
\widehat\phi_+(\xi)
=
H_+(\xi)
-\frac{i}{2\pi i}\log\!\left(\xi-\frac43\right)\widehat\phi_-\!\left(\xi-\frac43\right)
\qquad\text{near }\xi=\frac43,
\end{equation*}
and
\begin{equation*}
\widehat\phi_-(\xi)
=
H_-(\xi)
-\frac{i}{2\pi i}\log\!\left(\xi+\frac43\right)\widehat\phi_+\!\left(\xi+\frac43\right)
\qquad\text{near }\xi=-\frac43.
\end{equation*}

 Therefore the alien operators satisfy
\begin{equation*}
\Delta_{4/3}^{+}\widetilde\phi_+ = -\,i\,\widetilde\phi_-,
\qquad
\Delta_{-4/3}^{+}\widetilde\phi_- = -\,i\,\widetilde\phi_+.
\end{equation*}
We thus have
\begin{equation*}
\begin{aligned}
\dot\Delta_{4/3}^{+}\widetilde I_+
&=
e^{-4/(3\hbar)}\Delta_{4/3}^{+}
\Bigl(-\sqrt{\pi \hbar}\,e^{2/(3\hbar)}\,\widetilde\phi_+(\hbar)\Bigr)\\
&=
-e^{-4/(3\hbar)}
\sqrt{\pi \hbar}\,e^{2/(3\hbar)}\,\Delta_{4/3}^{+}\widetilde\phi_+(\hbar)
=
-\widetilde I_-(\hbar).
\\[0.1cm]
\dot\Delta_{-4/3}^{+}\widetilde I_-
&=
e^{4/(3\hbar)}\Delta_{-4/3}^{+}
\Bigl(-i\sqrt{\pi \hbar}\,e^{-2/(3\hbar)}\,\widetilde\phi_-(\hbar)\Bigr)\\
&=
-ie^{4/(3\hbar)}
\sqrt{\pi \hbar}\,e^{-2/(3\hbar)}\,\Delta_{-4/3}^{+}\widetilde\phi_-(\hbar)
=
\widetilde I_+(\hbar).
\end{aligned}
\end{equation*}

On the other hand, \(\widehat{\phi}_-\) (resp. \(\widehat{\phi}_+\)  ) has no singularity on the positive (resp. negative) real axis, so
\begin{equation*}
\Delta_{-4/3}^{+}\widetilde\phi_+=0,\qquad 
\Delta_{4/3}^{+}\widetilde\phi_-=0,
\end{equation*}
and therefore
\begin{equation*}
\dot\Delta_{-4/3}^{+}\widetilde I_+=0, \qquad
\dot\Delta_{4/3}^{+}\widetilde I_-=0.
\end{equation*}

 The Airy critical-value graded vector space is $ \C\,\widetilde I_- \oplus \C\,\widetilde I_+ $. With the ordering \[ \widetilde{\mathbf I}:=(\widetilde I_-,\widetilde I_+),\] we have
\[
\mathfrak S^{+}\widetilde{\mathbf I}
=
\widetilde{\mathbf I}
\begin{pmatrix}
1&-1\\
0&1
\end{pmatrix}.
\]
Therefore, by \eqref{Stokesauto} and Proposition~\ref{prop:thimble-lateral-sum},
\begin{equation*}
\mathbf I^{<}=\mathbf I^{>}
\begin{pmatrix}
1&-1\\
0&1
\end{pmatrix},
\qquad
\mathbf I^{>}=\mathbf I^{<}
\begin{pmatrix}
1&1\\
0&1
\end{pmatrix}.
\end{equation*}

\medskip 
\section{The Bessel model}
We now consider the Bessel-type exponential integral
\[I_{\Gamma}(\hbar)=\int_{\Gamma}e^{-S(z)/\hbar}\frac{\dd z}{z},\qquad S(z)=\frac{1}{2}\left(z-\frac{1}{z}\right),\quad z\in\C^*.\]
It is convenient to formulate this example on the logarithmic
cylinder.  Let
\[
w=\log z,
\qquad z=e^w,
\]
and view \(w\) as a coordinate on \(\C/(2\pi i\Z)\). 
Then the exponential integral becomes
\[
I_{\Gamma}(\hbar)
=
\int_{\Gamma}
e^{-\sinh w/\hbar}\,\dd w.
\]

\begin{rmk}
We call this the Bessel model because, if one keeps a parameter \(x\), then
\[
Y_\Gamma(x)
=
\int_\Gamma
\exp\left(\frac{x}{2}\left(z-\frac1z\right)\right)\frac{\dd z}{z}
\]
is the standard contour integral representation of a Bessel-type solution.
For the unit circle
\[\Gamma=\{|z|=1\}\subset\C^*\]
using the generating function
\[
\exp\left(\frac{x}{2}\left(z-\frac1z\right)\right)
=
\sum_{n\in\Z}J_n(x)z^n,
\]
one obtains
\[
Y_\Gamma(x)=2\pi i\,J_0(x).
\]
Hence \(Y_\Gamma(x)\) satisfies the order-zero Bessel equation
\[
x^2Y_\Gamma''(x)+xY_\Gamma'(x)+x^2Y_\Gamma(x)=0.
\]
In the present note we specialize to the large-parameter value
\(x=-1/\hbar\). 
\end{rmk}

The critical points 
$$
w_+=i\pi/2\quad  \text{and}\quad w_-=-i\pi/2
$$ 
are non-degenerate and 
the corresponding critical values are
\begin{equation*}
S(w_+)=\sinh\!\left(\frac{i\pi}{2}\right)=i,
\qquad
S(w_-)=\sinh\!\left(-\frac{i\pi}{2}\right)=-i.
\end{equation*}

Solving the equaiton
\[
\im(2ie^{-i\theta_*})=0
\]
gives the Stokes phases
\begin{equation*}
\theta_*=\frac{\pi}{2} \quad \pmod \pi.
\end{equation*}

\smallskip
\subsection{Lefschetz thimbles near the Stokes phase}

For fixed regular phase near the Stokes phase, 
\[\theta=\frac{\pi}{2}+\delta,\quad \delta\ne0 \,\text{ small enough},\]
we have the following descriptions. 
\begin{itemize}
    \item The stable and unstable manifolds through $w_+$ and $w_-$ lie in the level sets
\[
W^s(w_+)\cup W^u(w_+) \subset \{G_\theta=-\sin\delta\},
\qquad
W^s(w_-)\cup W^u(w_-) \subset \{G_\theta=+\sin\delta\}.
\]
\item The linearizations of the flow
\[
\dot{w}=-\overline{e^{-i\theta}\cosh w}
\]
at the critical points are given by
\[
\dot{\eta}=i\,e^{i\theta}\overline{\eta}
\qquad\text{at }~w_+,
\qquad
\dot{\eta}=-i\,e^{i\theta}\overline{\eta}
\qquad\text{at }~w_-.
\]
Hence locally the flow directions are as follows.
\[
\begin{aligned}
&\text{at }w_+:
&&\arg\eta=\frac{\delta}{2}\pmod{\pi}\ \text{stable},
&&\quad\arg\eta=\frac{\pi}{2}+\frac{\delta}{2}\pmod{\pi}\ \text{unstable},
\\
&\text{at }w_-:
&&\arg\eta=\frac{\pi}{2}+\frac{\delta}{2}\pmod{\pi}\ \text{stable},
&&\quad\arg\eta=\frac{\delta}{2}\pmod{\pi}\ \text{unstable}.
\end{aligned}
\]
Thus, the stable branches at \(w_+\) are locally almost horizontal, while those at \(w_-\) are locally almost vertical, with the orientations determined by the sign of $\delta$. 
\end{itemize}

Combined with the end asymptotics below, this determines the global thimble picture. In the real coordinates \(w=x+iy\), with \(-\pi\le y\le \pi\), the real part and imaginary part of the twisted function are given by
\begin{equation}\label{eq:bessel-G-positive-delta}
\begin{aligned}
F_\theta(x,y)
=
-\sin\delta\,\sinh x\cos y
+
\cos\delta\,\cosh x\sin y,\\ 
G_\theta(x,y)
=
-\sin\delta\,\cosh x\sin y
-
\cos\delta\,\sinh x\cos y.
\end{aligned}
\end{equation}
The flow equation reads as follows
\begin{equation*}\label{eq:bessel-flow-real-general}
\begin{aligned}
\dot{x}&=-\cos\theta\,\cosh x\,\cos y-\sin\theta\,\sinh x\,\sin y,\\
\dot{y}&=\phantom{-}\cos\theta\,\sinh x\,\sin y-\sin\theta\,\cosh x\,\cos y.
\end{aligned}
\end{equation*}
Expanding \eqref{eq:bessel-G-positive-delta} for \(x\to\pm\infty\), one obtains the following asymptotic positions.

\begin{lem}\label{lem:bessel-end-asymptotics}
The level set \(G_\theta=+\sin\delta\) contains two connected components associated to $w_{\pm}$, with the asymptotic expansions as follows:
\begin{itemize}
    \item at the component including $w_+$, that is $x=0$, $y=\pi/2$: 
    \begin{align*}
y=\frac{\pi}{2}-\delta-2\sin\delta\,e^x+O(e^{2x}),
&\qquad x\to -\infty,\\
y=\frac{\pi}{2}+\delta+2\sin\delta\,e^{-x}+O(e^{-2x}),
&\qquad x\to +\infty,
\end{align*}
   \item at the component including $w_-$, that is $x=0$, $y=-\pi/2$:
   \begin{align*}
y=-\frac{\pi}{2}-\delta+2\sin\delta\,e^x+O(e^{2x}),
&\qquad x\to -\infty,\\
y=-\frac{\pi}{2}+\delta-2\sin\delta\,e^{-x}+O(e^{-2x}),
&\qquad x\to +\infty.
\end{align*}
\end{itemize}
Similarly, for the level set \(G_\theta=-\sin\delta\), the asymptotic expansions are
\begin{itemize}
    \item at the component including $w_+$: 
\begin{align*}
y=\frac{\pi}{2}-\delta+2\sin\delta\,e^x+O(e^{2x}),
&\qquad x\to -\infty,\\
y=\frac{\pi}{2}+\delta-2\sin\delta\,e^{-x}+O(e^{-2x}),
&\qquad x\to +\infty,
\end{align*}
   \item at the component including $w_-$:
\begin{align*}
y=-\frac{\pi}{2}-\delta-2\sin\delta\,e^x+O(e^{2x}),
&\qquad x\to -\infty,\\
y=-\frac{\pi}{2}+\delta+2\sin\delta\,e^{-x}+O(e^{-2x}),
&\qquad x\to +\infty.
\end{align*}    
\end{itemize}
\end{lem}

\begin{proof}
This is obtained by inserting the ansatz
\[
y=\frac{\pi}{2}\mp\delta+\varepsilon(x)
\qquad\text{or}\qquad
y=-\frac{\pi}{2}\mp\delta+\varepsilon(x)
\]
into \eqref{eq:bessel-G-positive-delta} and solving to first order in \(\varepsilon\). The sign is then read off from the corresponding linearization.
\end{proof}

Because we are working on the cylinder \(\C/(2\pi i\Z)\), an invariant branch may leave the fundamental strip through \(y=\pi\) and re-enter through \(y=-\pi\).  This is exactly what happens for one of the stable branches of \(W^s(w_-)\). To see that one branch wraps, it is enough to locate where this level set meets the strip. Along the middle line \(y=0\), using \eqref{eq:bessel-G-positive-delta}, we get
\[
G_\theta(x,0)=-\cos\delta\,\sinh x.
\]
Hence
\[
G_\theta(x,0)=\sin\delta
\quad\Longleftrightarrow\quad
\sinh x=-\tan\delta.
\]
This gives the interior intersection
\[
(x,y)=(-a,0),\qquad \text{where }~a:=\operatorname{arsinh}(\tan\delta).
\]

On the other hand, along the top and bottom boundaries \(y=\pm\pi\), we have
\[
G_\theta(x,\pm\pi)=\cos\delta\,\sinh x.
\]
Thus
\[
G_\theta(x,\pm\pi)=\sin\delta
\quad\Longleftrightarrow\quad
\sinh x=\tan\delta,
\]
so the same level set meets the boundary at
\begin{equation*}\label{eq:bessel-boundary-intersection}
(x,y)=(a,\pi),
\qquad
(x,y)=(a,-\pi).
\end{equation*}

In summary, this component of \(G_\theta=\sin\delta\) has one branch crossing the strip directly through the interior point \((-a,0)\), and another branch meeting the strip only through the identified boundary points \((a,\pi)\) and \((a,-\pi)\).  Since the ambient space is the cylinder, the latter branch is a wrapped branch.  This proves that the right branch of \(W^s(w_-)\) wraps once around the cylinder.

The other wrapped/non-wrapped cases (see, Proposition~\ref{prop:bessel-endpoint-data}) are obtained in exactly the same way by checking the corresponding level set against the middle line and the boundary lines.

\begin{prop}[Endpoint data and wrapping information]
\label{prop:bessel-endpoint-data}
At \(\theta=\frac{\pi}{2}\pm\delta\) with \(0<\delta\ll 1\), the stable and unstable manifolds for \(w_\pm\) have the following behavior:
\[
\renewcommand{\arraystretch}{1.25}
\begin{array}{c|ccc|ccc}
& \multicolumn{3}{c|}{\theta=\frac{\pi}{2}-\delta}
& \multicolumn{3}{c}{\theta=\frac{\pi}{2}+\delta} \\
\text{manifold}
& x\to-\infty & x\to+\infty & \text{wrapped}
& x\to-\infty & x\to+\infty & \text{wrapped}
\\ \hline
W^s(w_+)
& y\to \frac{\pi}{2}+\delta
& y\to \frac{\pi}{2}-\delta
& \text{no}
& y\to \frac{\pi}{2}-\delta
& y\to \frac{\pi}{2}+\delta
& \text{no}
\\[1mm]
W^s(w_-)
& y\to \frac{\pi}{2}+\delta
& y\to \frac{\pi}{2}-\delta
& \text{yes}
& y\to \frac{\pi}{2}-\delta
& y\to \frac{\pi}{2}+\delta
& \text{yes}
\\[1mm]
W^u(w_-)
& y\to -\frac{\pi}{2}+\delta
& y\to -\frac{\pi}{2}-\delta
& \text{no}
& y\to -\frac{\pi}{2}-\delta
& y\to -\frac{\pi}{2}+\delta
& \text{no}
\\[1mm]
W^u(w_+)
& y\to -\frac{\pi}{2}+\delta
& y\to -\frac{\pi}{2}-\delta
& \text{yes}
& y\to -\frac{\pi}{2}-\delta
& y\to -\frac{\pi}{2}+\delta
& \text{yes}
\end{array}
\]
The two adjacent generic configurations are displayed in the left and right panels of Figure~\ref{fig:bessel-three-phases}.
\end{prop}

\smallskip
\subsection{Connecting trajectories and the Picard--Lefschetz jump formula}
\label{sec:BesselLef}

At \(\theta_*=\frac{\pi}{2}\), one has 
\[F_{\pi/2}(w_+)>F_{\pi/2}(w_-),\qquad G_{\pi/2}(w_+)=G_{\pi/2}(w_-)=0,\]
so there may exist connecting trajectories in the level set $G_{\pi/2}=0$ from \(w_+\) to \(w_-\). Now the flow equation \eqref{eq:bessel-flow-real-general} becomes
\begin{equation}\label{eq:bessel-flow-pi-over-2}
\dot x=-\sinh x\,\sin y,
\qquad
\dot y=-\cosh x\,\cos y.
\end{equation}
Moreover, 
\[
G_{\pi/2}(x,y)=-\sinh x\,\cos y=0~~\Longrightarrow~~ \cos y=0~\text{ or }~ \sinh x=0.
\]
If \(\cos y=0\), then \(y\) is constant by \eqref{eq:bessel-flow-pi-over-2}, so such a trajectory never connects the two critical points.  Since $\sinh x=0$ implies $x=0$, \eqref{eq:bessel-flow-pi-over-2} reduces to
\[
\dot y=-\cos y.
\]
There are exactly two solutions connecting \(w_+\) to \(w_-\):
\begin{equation*}\label{eq:bessel-traj-right}
\begin{cases}
    y(s)=2\arctan\,\bigl(e^{-(s-s_0)}\bigr)-\frac{\pi}{2}  \quad  &\text{as } \, y(s)\in\left(-\frac{\pi}{2},\frac{\pi}{2}\right) 
    \\[0.1cm]
    y(s)=2\arctan\,\bigl(e^{\,s-s_0}\bigr)+\frac{\pi}{2}
    \quad  &\text{as } \, y(s)\in\left(\frac{\pi}{2},\frac{3\pi}{2}\right) 
\end{cases}.
\end{equation*}

The Stokes configuration with the two blue connecting trajectories is displayed in the middle panel of Figure~\ref{fig:bessel-three-phases}.

\begin{figure}[H]
\centering

\tikzset{
  stab/.style={
    black,
    line width=1.2pt,
    postaction={decorate},
    decoration={
      markings,
      mark=at position 0.78 with {\arrow{latex}}
    }
  },
  unstab/.style={
    green!60!black,
    line width=1.3pt,
    postaction={decorate},
    decoration={
      markings,
      mark=at position 0.22 with {\arrow{latex}}
    }
  },
  broken/.style={
    blue!75!black,
    line width=2.5pt,
    postaction={decorate},
    decoration={
      markings,
      mark=at position 0.45 with {\arrow{latex}}
    }
  }
}
\colorlet{sectorred}{red!18}
\begin{minipage}[t]{0.47\textwidth}
\centering
\[
\theta=\frac{\pi}{2}-\delta
\]
\resizebox{\linewidth}{!}{
\begin{tikzpicture}[scale=1.02, >=Latex]

\def\xL{-5.5}
\def\xR{5.5}
\def\d{0.38}
\def\a{0.50}
\def\b{0.50}
\def\pih{1.5708}
\def\ymin{-3.1416}
\def\ymax{3.1416}

\draw[thin] (\xL,\ymin) rectangle (\xR,\ymax);
\fill[sectorred] (\xL,\d) rectangle (\xL+0.95,\ymax);
\fill[sectorred] (\xL,\ymin) rectangle (\xL+0.95,\ymin+\d);
\fill[sectorred] (\xR-0.95,-\d) rectangle (\xR,\ymax-\d);

\draw[densely dashed, thin] (\xL,\pih) -- (\xR,\pih);
\draw[densely dashed, thin] (\xL,-\pih) -- (\xR,-\pih);
\draw[densely dashed, red!65!black] (\xL,\pih+\d) -- (\xL+1.25,\pih+\d);
\draw[densely dashed, red!65!black] (\xR-1.25,\pih-\d) -- (\xR,\pih-\d);

\draw[<->, red!70!black] (\xL+1.15,\pih) -- (\xL+1.15,\pih+\d);
\node[left, red!70!black] at (\xL+1.10,\pih+0.5*\d) {\(\delta\)};
\draw[<->, red!70!black] (\xL+1.15,\ymin) -- (\xL+1.15,\ymin+\d);
\node[left, red!70!black] at (\xL+1.10,\ymin+0.5*\d) {\(\delta\)};
\draw[<->, red!70!black] (\xR-1.15,\pih-\d) -- (\xR-1.15,\pih);
\node[right, red!70!black] at (\xR-1.10,\pih-0.5*\d) {\(\delta\)};

\filldraw[black] (0,\pih) circle (1.6pt);
\filldraw[black] (0,-\pih) circle (1.6pt);
\node[above right] at (0.14,\pih) {\(w_+\)};
\node[above right] at (0.14,-\pih) {\(w_-\)};

\draw[stab]
    (\xL,\pih+\d-0.22)
    .. controls (-4.35,\pih+\d-0.16) and (-1.65,\pih+0.05) ..
    (0,\pih);
\draw[stab]
    (\xR,\pih-\d+0.22)
    .. controls (4.35,\pih-\d+0.16) and (1.65,\pih-0.05) ..
    (0,\pih);
\draw[stab]
    (\xL,\pih+\d+0.22)
    .. controls (-4.35,\pih+\d+0.22) and (-0.5,\pih+\d+0.22) ..
    (-\b,\ymax-0.03);
\draw[stab]
    (-\b,\ymin+0.03)
    .. controls (-\b+0.2,\ymin+0.2) and (0,-\pih-0.8) ..
    (0,-\pih);
\draw[stab]
    (\xR,\pih-\d-0.22)
    .. controls (1.35,\pih-\d-0.12) and (0.6,\pih-0.3) .. (0.5,\pih-0.5) ..controls (0.3,\pih-0.4) and (0,-1) ..
    (0,-\pih);

\draw[unstab]
    (\b,-\ymax+0.03)
    .. controls (0.5,-\pih-\d-0.22) and (4.35,-\pih-\d-0.22) ..
    (-\xL,-\pih-\d-0.22) ;
\draw[unstab]
    (0,+\pih)
    .. controls (0,\pih+0.8) and (\b-0.2,-\ymin+0.2)  ..
    (\b,-\ymin-0.03);
\draw[unstab]
    (0,\pih)
    .. controls (0,1)  and (-0.3,-\pih+0.4) .. (-0.5,-\pih+0.5) ..controls (-0.6,-\pih+0.3) and (-1.35,-\pih+\d+0.12) ..
    (-\xR,-\pih+\d+0.22); 
\draw[unstab]
    (0,-\pih)
    .. controls (-1.45,-1.48) and (-4.35,-1.24) ..
    (\xL,-\pih+\d-0.12);
\draw[unstab]
    (0,-\pih)
    .. controls (1.45,-1.58) and (4.35,-1.88) ..
    (\xR,-\pih-\d+0.12);
\draw[gray!70,->,thin] (-\b,\ymax-0.18) -- (-\b,\ymax-0.02);
\draw[gray!70,->,thin] (-\b,\ymin+0.18) -- (-\b,\ymin+0.02);
\draw[gray!70,->,thin] (\a,\ymax-0.18) -- (\a,\ymax-0.02);
\draw[gray!70,->,thin] (\a,\ymin+0.18) -- (\a,\ymin+0.02);
\end{tikzpicture}
}
\end{minipage}
\hfill
\begin{minipage}[t]{0.47\textwidth}
\centering
\[
\theta=\frac{\pi}{2}+\delta
\]
\resizebox{\linewidth}{!}{
\begin{tikzpicture}[scale=1.02, >=Latex]

  \def\xL{-5.5}
  \def\xR{5.5}
  \def\d{0.38}
  \def\a{0.5}
  \def\b{0.5}
  \def\pih{1.5708}
  \def\ymin{-3.1416};
  \def\ymax{3.1416};

\draw[thin] (\xL,\ymin) rectangle (\xR,\ymax);
\fill[sectorred] (\xL,-\d) rectangle (\xL+0.95,\ymax-\d);
\fill[sectorred] (\xR-0.95,\d) rectangle (\xR,\ymax);
\fill[sectorred] (\xR-0.95,\ymin) rectangle (\xR,\ymin+\d);

\draw[densely dashed, thin] (\xL,\pih) -- (\xR,\pih);
\draw[densely dashed, thin] (\xL,-\pih) -- (\xR,-\pih);
\draw[densely dashed, red!65!black] (\xL,\pih-\d) -- (\xL+1.25,\pih-\d);
\draw[densely dashed, red!65!black] (\xR-1.25,\pih+\d) -- (\xR,\pih+\d);
\draw[<->, red!70!black] (\xL+1.15,\pih-\d) -- (\xL+1.15,\pih);
\node[left, red!70!black] at (\xL+1.10,\pih-0.5*\d) {\(\delta\)};
\draw[<->, red!70!black] (\xR-1.15,\pih) -- (\xR-1.15,\pih+\d);
\node[right, red!70!black] at (\xR-1.10,\pih+0.5*\d) {\(\delta\)};
\draw[<->, red!70!black] (\xR-1.15,\ymin) -- (\xR-1.15,\ymin+\d);
\node[right, red!70!black] at (\xR-1.10,\ymin+0.5*\d) {\(\delta\)};

\filldraw[black] (0,\pih) circle (1.6pt);
\filldraw[black] (0,-\pih) circle (1.6pt);
\node[above right] at (0.14,\pih) {\(w_+\)};
\node[above right] at (0.14,-\pih) {\(w_-\)};

\draw[stab]
    (\xL,\pih-\d+0.22)
    .. controls (-4.35,\pih-\d+0.18) and (-1.65,\pih-0.05) ..
    (0,\pih);
\draw[stab]
    (\xR,\pih+\d-0.22)
    .. controls (4.35,\pih+\d-0.18) and (1.65,\pih+0.05) ..
    (0,\pih);
\draw[stab]
    (-\xL,\pih+\d+0.22)
    .. controls (4.35,\pih+\d+0.22) and (0.5,\pih+\d+0.22) ..
    (\b,\ymax-0.03);
\draw[stab]
    (\b,\ymin+0.03)
    .. controls (\b-0.2,\ymin-0.2)  and (0,-\pih-0.8) ..
    (0,-\pih);
\draw[stab]
    (-\xR,\pih-\d-0.22)
    .. controls (-1.35,\pih-\d-0.12) and (-0.6,\pih-0.3) .. (-0.5,\pih-0.5) ..controls (-0.3,\pih-0.4) and (0,-1) ..
    (0,-\pih);

\draw[unstab]
    (-\b,-\ymax+0.03)
    .. controls (-0.5,-\pih-\d-0.22) and (-4.35,-\pih-\d-0.22) ..
    (\xL,-\pih-\d-0.22) ;
\draw[unstab]
    (0,+\pih)
    .. controls (0,\pih+0.8) and (-\b+0.2,-\ymin-0.2)  ..
    (-\b,-\ymin-0.03);
\draw[unstab]
    (0,\pih)
    .. controls (0,1)  and (0.3,-\pih+0.4) .. (0.5,-\pih+0.5) ..controls (0.6,-\pih+0.3) and (1.35,-\pih+\d+0.12) ..
    (\xR,-\pih+\d+0.22); 
\draw[unstab]
    (0,-\pih)
    .. controls (-1.45,-1.58) and (-4.35,-1.88) ..
    (\xL,-\pih-\d+0.12);
\draw[unstab]
    (0,-\pih)
    .. controls (1.45,-1.48) and (4.35,-1.24) ..
    (\xR,-\pih+\d-0.12);

\draw[gray!70,->,thin] (\a,\ymax-0.18) -- (\a,\ymax-0.02);
\draw[gray!70,->,thin] (\a,\ymin+0.18) -- (\a,\ymin+0.02);
\draw[gray!70,->,thin] (-\b,\ymax-0.18) -- (-\b,\ymax-0.02);
\draw[gray!70,->,thin] (-\b,\ymin+0.18) -- (-\b,\ymin+0.02);
\end{tikzpicture}
}
\end{minipage}

\vspace{1.0em}

\begin{minipage}[t]{0.47\textwidth}
\centering
\[
\theta_*=\frac{\pi}{2}
\]
\resizebox{\linewidth}{!}{
\begin{tikzpicture}[scale=1.02, >=Latex]

  \def\xL{-5.5}
  \def\xR{5.5}
  \def\pih{1.5708}
  \def\ymin{-3.1416}
  \def\ymax{3.1416}

\draw[thin] (\xL,\ymin) rectangle (\xR,\ymax);
\fill[sectorred] (\xL,0) rectangle (\xL+0.95,\ymax);
\fill[sectorred] (\xR-0.95,0) rectangle (\xR,\ymax);

\draw[densely dashed, thin] (\xL,\pih) -- (\xR,\pih);
\draw[densely dashed, thin] (\xL,-\pih) -- (\xR,-\pih);
\draw[densely dashed, thin] (\xL,0) -- (\xR,0);

\filldraw[black] (0,\pih) circle (1.6pt);
\filldraw[black] (0,-\pih) circle (1.6pt);
\node[above right] at (0.14,\pih) {\(w_+\)};
\node[above right] at (0.14,-\pih) {\(w_-\)};

\draw[stab] (\xL,\pih) -- (0,\pih);
\draw[stab] (\xR,\pih) -- (0,\pih);
\draw[unstab] (0,-\pih) -- (\xL,-\pih);
\draw[unstab] (0,-\pih) -- (\xR,-\pih);

\draw[broken] (0,\pih) -- (0,-\pih);
\draw[broken] (0,\pih) -- (0,\ymax);
\draw[broken] (0,\ymin) -- (0,-\pih);
\end{tikzpicture}
}
\end{minipage}

\begin{minipage}[c]{0.14\textwidth}
\centering
\resizebox{0.95\linewidth}{!}{
\begin{tikzpicture}[>=Latex]
\tikzset{
  flowstable/.style={
    black,
    line width=1.1pt,
    postaction={decorate},
    decoration={
      markings,
      mark=at position 0.55 with {\arrow{latex}}
    }
  },
  flowunstable/.style={
    green!60!black,
    line width=1.3pt,
    postaction={decorate},
    decoration={
      markings,
      mark=at position 0.55 with {\arrow{latex}}
    }
  },
  flowtraj/.style={
    blue!75!black,
    line width=2.5pt,
    postaction={decorate},
    decoration={
      markings,
      mark=at position 0.55 with {\arrow{latex}}
    }
  }
}
  \draw[flowstable] (0,0) -- (1.15,0);
  \node[right,scale=0.78] at (1.28,0) {stable};

  \draw[flowunstable] (0,-0.55) -- (1.15,-0.55);
  \node[right,scale=0.78] at (1.28,-0.55) {unstable};

  \draw[flowtraj] (0,-1.10) -- (1.15,-1.10);
  \node[right,scale=0.78] at (1.28,-1.10) {connecting trajectory};
\end{tikzpicture}
}
\end{minipage}

\caption{
Flow pictures for the Bessel action on the \(w\)-cylinder (here \(0<\delta\ll 1\)).
The pink shaded regions are the stable/decay sectors at the ends.
The black curves are stable manifolds, the green curves are unstable manifolds, and the thick blue curves in the middle panel are the two connecting trajectories at \(\theta_*=\frac{\pi}{2}\).
}
\label{fig:bessel-three-phases}
\end{figure}
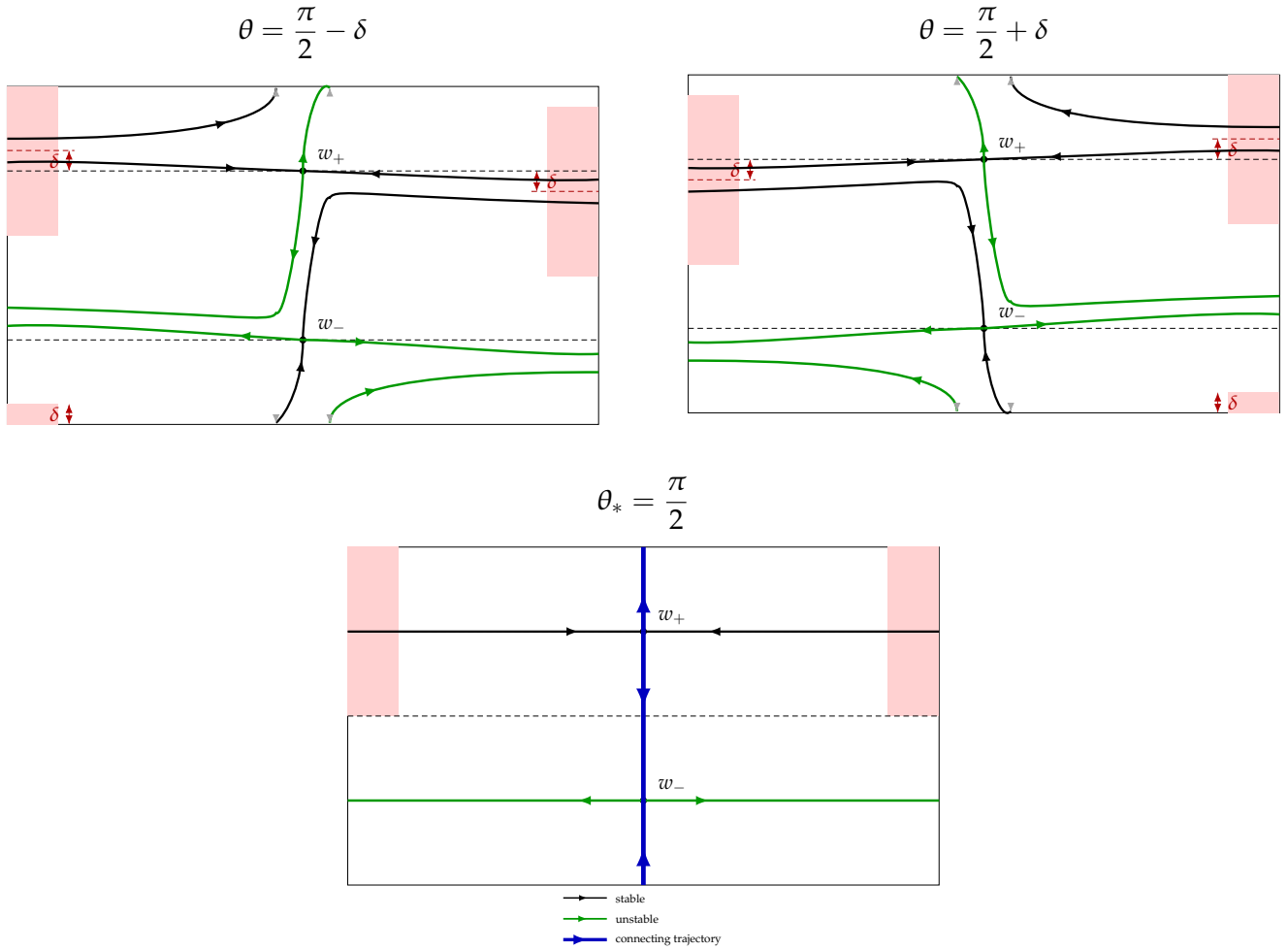

\begin{rmk}
At the opposite special phase \(\theta_*=-\frac{\pi}{2}\), the same argument shows that there are exactly two connecting trajectories in the reverse direction, namely from \(w_-\) to \(w_+\).
\end{rmk}

The local orientations of the thimbles are chosen so that the direction of \(\J_p^{>}\) (resp.\ \(\Kk_p^{>}\)) agrees with that of \(\J_p^{<}\) (resp.\ \(\Kk_p^{<}\)) near \(p=w_{\pm}\), and so that
$
\langle [\J_p^\theta],[\Kk_q^\theta]\rangle=\delta_{pq}
$ with $p,q=w_{\pm}$. With these orientation conventions, the thimble and dual-thimble pictures are shown in Figure~\ref{fig:bessel-JK-two-panels}.

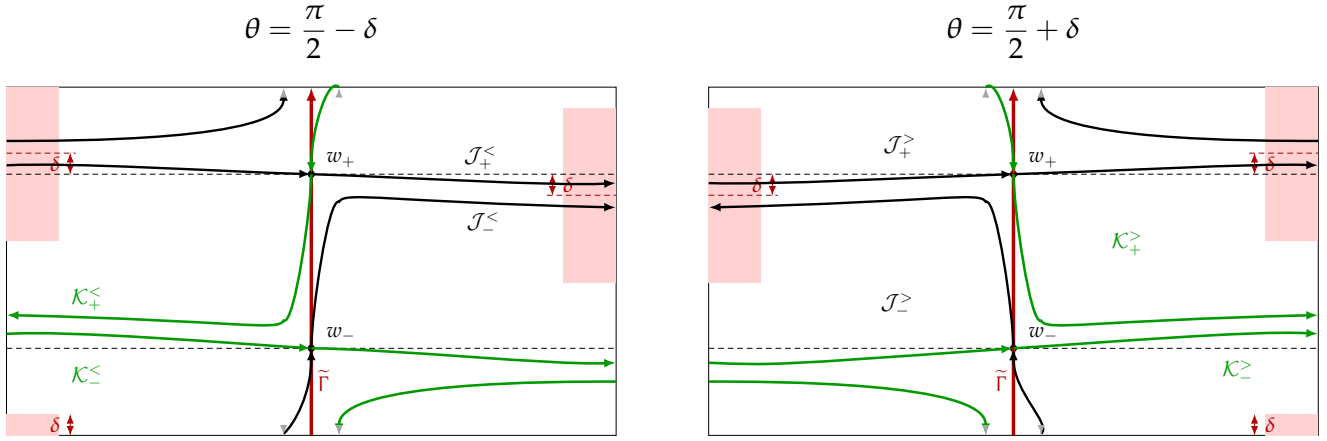
\begin{figure}[H]
\centering
\colorlet{GammaRed}{red!70!black}
\tikzset{
  GammaBessel/.style={
    GammaRed,
    line width=1.8pt,
    -{Latex[length=2.5mm]}
  }
}

\tikzset{
  Jflow/.style={
    black,
    line width=1.25pt,
    -{Latex[length=2.2mm]}
  },
  Kflow/.style={
    green!60!black,
    line width=1.35pt,
       -{Latex[length=2.2mm]}
  }
}
\colorlet{sectorred}{red!18}
\begin{minipage}[t]{0.47\textwidth}
\centering
\[
\theta=\frac{\pi}{2}-\delta
\]
\resizebox{\linewidth}{!}{
\begin{tikzpicture}[scale=1.02, >=Latex]

  \def\xL{-5.5}
  \def\xR{5.5}
  \def\d{0.38}
  \def\a{0.5}
  \def\b{0.5}
  \def\pih{1.5708}
  \def\ymin{-3.1416}
  \def\ymax{3.1416}

\draw[thin] (\xL,\ymin) rectangle (\xR,\ymax);

\fill[sectorred] (\xL,\d) rectangle (\xL+0.95,\ymax);
\fill[sectorred] (\xL,\ymin) rectangle (\xL+0.95,\ymin+\d);
\fill[sectorred] (\xR-0.95,-\d) rectangle (\xR,\ymax-\d);

\draw[densely dashed, thin] (\xL,\pih) -- (\xR,\pih);
\draw[densely dashed, thin] (\xL,-\pih) -- (\xR,-\pih);
\draw[densely dashed, red!65!black] (\xL,\pih+\d) -- (\xL+1.25,\pih+\d);
\draw[densely dashed, red!65!black] (\xR-1.25,\pih-\d) -- (\xR,\pih-\d);

\draw[<->, red!70!black] (\xL+1.15,\pih) -- (\xL+1.15,\pih+\d);
\node[left, red!70!black] at (\xL+1.10,\pih+0.5*\d) {\(\delta\)};
\draw[<->, red!70!black] (\xL+1.15,\ymin) -- (\xL+1.15,\ymin+\d);
\node[left, red!70!black] at (\xL+1.10,\ymin+0.5*\d) {\(\delta\)};
\draw[<->, red!70!black] (\xR-1.15,\pih-\d) -- (\xR-1.15,\pih);
\node[right, red!70!black] at (\xR-1.10,\pih-0.5*\d) {\(\delta\)};

\draw[GammaBessel] (0,\ymin) -- (0,\ymax);
\node[GammaRed] at (0.22,-2.15) {\(\widetilde\Gamma\)};
\filldraw[black] (0,\pih) circle (1.6pt);
\filldraw[black] (0,-\pih) circle (1.6pt);
\node[above right] at (0.14,\pih) {\(w_+\)};
\node[above right] at (0.14,-\pih) {\(w_-\)};

\draw[Jflow]
    (\xL,\pih+\d+0.22)
    .. controls (-4.35,\pih+\d+0.22) and (-0.5,\pih+\d+0.22) ..
    (-\b,\ymax-0.03);
\draw[Jflow]
    (-\b,\ymin+0.03)
    .. controls (-\b+0.2,\ymin+0.2) and (0,-\pih-0.8) ..
    (0,-\pih);
\draw[Jflow]
    (0,-\pih)
    .. controls (0,-1) and (0.3,\pih-0.4)  .. (0.5,\pih-0.5) ..controls (0.6,\pih-0.3) and (1.35,\pih-\d-0.12) ..
    (\xR,\pih-\d-0.22);
\draw[Jflow]
    (\xL,\pih+\d-0.22)
    .. controls (-4.35,\pih+\d-0.16) and (-1.65,\pih+0.05) ..
    (0,\pih);
\draw[Jflow]
    (0,\pih)
    .. controls (1.65,\pih-0.05) and (4.35,\pih-\d+0.16) ..
    (\xR,\pih-\d+0.22);

\draw[Kflow]
    (\xL,-\pih+\d-0.12)
    .. controls (-4.35,-1.24) and (-1.45,-1.48) ..
    (0,-\pih);
\draw[Kflow]
    (0,-\pih)
    .. controls (1.45,-1.58) and (4.35,-1.88) ..
    (\xR,-\pih-\d+0.12);
\draw[Kflow]
  (-\xL,-\pih-\d-0.22) 
    .. controls (4.35,-\pih-\d-0.22) and (0.5,-\pih-\d-0.22)  ..
   (\b,-\ymax)  ;
\draw[Kflow]
    (\b,\ymax)
    .. controls (\b-0.2,-\ymin+0.2) and (0,\pih+0.8)  ..
    (0,+\pih);
\draw[Kflow]
    (0,\pih)
    .. controls (0,1)  and (-0.3,-\pih+0.4) .. (-0.5,-\pih+0.5) ..controls (-0.6,-\pih+0.3) and (-1.35,-\pih+\d+0.12) ..
    (-\xR,-\pih+\d+0.22); 
  
\draw[gray!70,->,thin] (-\b,\ymax-0.18) -- (-\b,\ymax-0.02);
\draw[gray!70,->,thin] (-\b,\ymin+0.18) -- (-\b,\ymin+0.02);
\draw[gray!70,->,thin] (\a,\ymax-0.18) -- (\a,\ymax-0.02);
\draw[gray!70,->,thin] (\a,\ymin+0.18) -- (\a,\ymin+0.02);

\node[black] at (3.05,\pih+0.30) {\(\J_+^<\)};
\node[black] at (3.10,0.68) {\(\J_-^<\)};
\node[green!60!black] at (-4.05,-\pih-0.52) {\(\Kk_-^<\)};
\node[green!60!black] at (-4.05,-0.68) {\(\Kk_+^<\)};
\end{tikzpicture}
}
\end{minipage}
\hfill
\begin{minipage}[t]{0.47\textwidth}
\centering
\[
\theta=\frac{\pi}{2}+\delta
\]
\resizebox{\linewidth}{!}{
\begin{tikzpicture}[scale=1.02, >=Latex]

  \def\xL{-5.5}
  \def\xR{5.5}
  \def\d{0.38}
  \def\a{0.5}
  \def\b{0.5}
  \def\pih{1.5708}
  \def\ymin{-3.1416}
  \def\ymax{3.1416}

\draw[thin] (\xL,\ymin) rectangle (\xR,\ymax);

\fill[sectorred] (\xL,-\d) rectangle (\xL+0.95,\ymax-\d);
\fill[sectorred] (\xR-0.95,\d) rectangle (\xR,\ymax);
\fill[sectorred] (\xR-0.95,\ymin) rectangle (\xR,\ymin+\d);

\draw[densely dashed, thin] (\xL,\pih) -- (\xR,\pih);
\draw[densely dashed, thin] (\xL,-\pih) -- (\xR,-\pih);
\draw[densely dashed, red!65!black] (\xL,\pih-\d) -- (\xL+1.25,\pih-\d);
\draw[densely dashed, red!65!black] (\xR-1.25,\pih+\d) -- (\xR,\pih+\d);

\draw[<->, red!70!black] (\xL+1.15,\pih-\d) -- (\xL+1.15,\pih);
\node[left, red!70!black] at (\xL+1.10,\pih-0.5*\d) {\(\delta\)};
\draw[<->, red!70!black] (\xR-1.15,\pih) -- (\xR-1.15,\pih+\d);
\node[right, red!70!black] at (\xR-1.10,\pih+0.5*\d) {\(\delta\)};
\draw[<->, red!70!black] (\xR-1.15,\ymin) -- (\xR-1.15,\ymin+\d);
\node[right, red!70!black] at (\xR-1.10,\ymin+0.5*\d) {\(\delta\)};

\filldraw[black] (0,\pih) circle (1.6pt);
\filldraw[black] (0,-\pih) circle (1.6pt);
\node[above right] at (0.14,\pih) {\(w_+\)};
\node[above right] at (0.14,-\pih) {\(w_-\)};
\draw[GammaBessel] (0,\ymin) -- (0,\ymax);
\node[GammaRed] at (-0.22,-2.15) {\(\widetilde\Gamma\)};

\draw[Jflow]
    (\xL,\pih-\d+0.22)
    .. controls (-4.35,\pih-\d+0.18) and (-1.65,\pih-0.05) ..
    (0,\pih);
\draw[Jflow]
    (0,\pih)
    .. controls (1.65,\pih+0.05) and (4.35,\pih+\d-0.18) ..
    (\xR,\pih+\d-0.22);
\draw[Jflow]
    (-\xL,\pih+\d+0.22)
    .. controls (4.35,\pih+\d+0.22) and (0.5,\pih+\d+0.22) ..
    (\b,\ymax-0.03);
\draw[Jflow]
    (\b,\ymin+0.03)
    .. controls (\b+0.2,\ymin+0.2) and (0,-\pih-0.8) ..
    (0,-\pih);
\draw[Jflow]
    (0,-\pih)
    .. controls (0,-1) and (-0.3,\pih-0.4)  .. (-0.5,\pih-0.5) ..controls (-0.6,\pih-0.3) and (-1.35,\pih-\d-0.12) ..
    (-\xR,\pih-\d-0.22);

\draw[Kflow]
    (0,-\pih)
    .. controls (1.45,-1.48) and   (4.35,-1.24)..
   (\xR,-\pih+\d-0.12) ;
\draw[Kflow]
    (\xL,-\pih-\d+0.12)
    .. controls (-4.35,-1.88) and (-4.35,-1.88) ..
    (0,-\pih);
\draw[Kflow]
  (\xL,-\pih-\d-0.22) 
    .. controls (-4.35,-\pih-\d-0.22) and (-0.5,-\pih-\d-0.22)  ..
   (-\b,-\ymax)  ;
\draw[Kflow]
    (-\b,\ymax)
    .. controls (-\b+0.2,-\ymin+0.2) and (0,\pih+0.8)  ..
    (0,+\pih);
\draw[Kflow]
    (0,\pih)
    .. controls (0,1)  and (0.3,-\pih+0.4) .. (0.5,-\pih+0.5) ..controls (0.6,-\pih+0.3) and (1.35,-\pih+\d+0.12) ..
    (\xR,-\pih+\d+0.22); 

\draw[gray!70,->,thin] (\a,\ymax-0.18) -- (\a,\ymax-0.02);
\draw[gray!70,->,thin] (\a,\ymin+0.18) -- (\a,\ymin+0.02);
\draw[gray!70,->,thin] (-\b,\ymax-0.18) -- (-\b,\ymax-0.02);
\draw[gray!70,->,thin] (-\b,\ymin+0.18) -- (-\b,\ymin+0.02);

\node[black] at (-2.05,\pih+0.50) {\(\J_+^>\)};
\node[black] at (-2.10,-0.78) {\(\J_-^>\)};
\node[green!60!black] at (4.05,-\pih-0.42) {\(\Kk_-^>\)};
\node[green!60!black] at (2.05,0.32) {\(\Kk_+^>\)};
\end{tikzpicture}
}
\end{minipage}

\caption{
The red curve \(\widetilde\Gamma\) indicates a typical choice of contour for the Bessel integral (here \(0<\delta\ll 1\)).
}
\label{fig:bessel-JK-two-panels}
\end{figure}

The two connecting trajectories can be visualized in the cylindrical
schematic picture in Figure~\ref{fig:bessel-cylinder-stokes-schematic}.

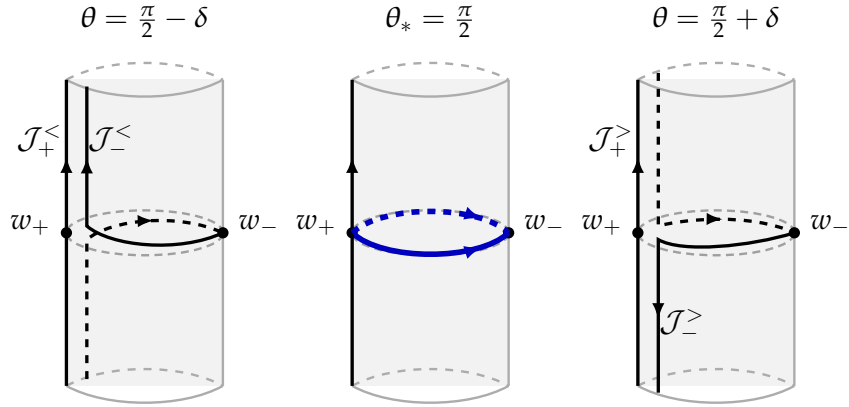
\begin{figure}[H]
\centering

\begin{tikzpicture}[>=Latex, scale=0.9]

\tikzset{
  cyl/.style={
    gray!65,
    line width=0.9pt
  },
  cylfill/.style={
    fill=gray!10,
    draw=none
  },
  waist/.style={
    gray!75,
    densely dashed,
    line width=0.9pt
  },
  crit/.style={
    circle,
    fill=black,
    inner sep=1.5pt
  },
  Jcurve/.style={
    black,
    line width=1.25pt,
    postaction={decorate},
    decoration={
      markings,
      mark=at position 0.75 with {\arrow{Latex[length=2.2mm]}}
    }
  },
  JcurveBack/.style={
    black,
    line width=1.25pt,
    dashed,
    postaction={decorate},
    decoration={
      markings,
      mark=at position 0.75 with {\arrow{Latex[length=2.2mm]}}
    }
  },
  traj/.style={
    blue!75!black,
    line width=2.1pt,
    postaction={decorate},
    decoration={
      markings,
      mark=at position 0.82 with {\arrow{Latex[length=2.4mm]}}
    }
  },
  trajback/.style={
    blue!75!black,
    line width=2.1pt,
    dashed,
    postaction={decorate},
    decoration={
      markings,
      mark=at position 0.82 with {\arrow{Latex[length=2.4mm]}}
    }
  }
}
\def\rx{1.15}
\def\ry{0.33}
\def\H{2.25}

\newcommand{\drawCylinderBase}{
\fill[cylfill] (-\rx,-\H) rectangle (\rx,\H);

\draw[cyl,dashed] (-\rx,\H)
    .. controls (-0.55,\H+\ry) and (0.55,\H+\ry) ..
    (\rx,\H);
\draw[cyl,dashed] (-\rx,-\H)
    .. controls (-0.55,-\H+\ry) and (0.55,-\H+\ry) ..
    (\rx,-\H);

\draw[cyl] (-\rx,-\H) -- (-\rx,\H);
\draw[cyl] (\rx,-\H) -- (\rx,\H);

\draw[cyl] (-\rx,\H)
    .. controls (-0.55,\H-\ry) and (0.55,\H-\ry) ..
    (\rx,\H);
\draw[cyl] (-\rx,-\H)
    .. controls (-0.55,-\H-\ry) and (0.55,-\H-\ry) ..
    (\rx,-\H);
\draw[waist] (0,0) ellipse[x radius=\rx, y radius=\ry];

\node[crit] (wp) at (-\rx,0) {};
\node[crit] (wm) at (\rx,0) {};
\node[left] at (-\rx-0.08,0.12) {$w_+$};
\node[right] at (\rx+0.08,0.12) {$w_-$};
}

\begin{scope}[shift={(-4.2,0)}]
  \node at (0,\H+0.85) {$\theta=\frac\pi2-\delta$};
  \drawCylinderBase

\draw[Jcurve]
(\rx,0) .. controls  (0.5,-0.3) and  (-0.5,-0.2)..
    (-0.85,0.1)
    --
    (-0.85,\H-0.1);
\draw[JcurveBack]
    (-0.85,-\H+0.1)
    --
    (-0.85,-0.1)
    .. controls (-0.5,0.2) and (0.5,0.3) ..
    (\rx,0);
\draw[Jcurve]
    (-\rx,-\H)--(-\rx,\H);
  \node[black] at (-1.55,1.35) {$\J_+^{<}$};
  \node[black] at (-0.5,1.35) {$\J_-^{<}$};
\end{scope}

\begin{scope}[shift={(0,0)}]
  \node at (0,\H+0.85) {$\theta_*=\frac{\pi}{2}$};
  \drawCylinderBase

\draw[traj]
    (-\rx,0)
    .. controls (-0.70,-0.42) and (0.70,-0.42) ..
    (\rx,0);
\draw[Jcurve]
    (-\rx,-\H)--(-\rx,\H);
\draw[trajback]
    (-\rx,0)
    .. controls (-0.70,0.42) and (0.70,0.42) ..
    (\rx,0);
\end{scope}

\begin{scope}[shift={(4.2,0)}]
  \node at (0,\H+0.85) {$\theta=\frac\pi2+\delta$};
  \drawCylinderBase
\draw[Jcurve]
(\rx,0) .. controls  (0.5,-0.2) and  (-0.5,-0.3)..
    (-0.85,-0.1)
    --
    (-0.85,-\H-0.1);
\draw[JcurveBack]
    (-0.85,\H+0.1)
    --
    (-0.85,0.1)
    .. controls (-0.5,0.2) and (0.5,0.3) ..
    (\rx,0);
\draw[Jcurve]
    (-\rx,-\H)--(-\rx,\H);
  \node[black] at (-1.55,1.35) {$\J_+^{>}$};
\node[black] at (-0.5,-1.35) {$\J_-^{>}$};
\end{scope}

\end{tikzpicture}

\caption{Schematic cylinder pictures for the Bessel model (here \(0<\delta\ll 1\)).}
\label{fig:bessel-cylinder-stokes-schematic}
\end{figure}

One can read off the jump matrix
\begin{equation}\label{equ:Besseljump}
\bigl([\J_+^>] \ \ [\J_-^>]\bigr)
=
\bigl([\J_+^<] \ \ [\J_-^<]\bigr)
\begin{pmatrix}
1&-2\\
0&1
\end{pmatrix}.
\end{equation}

\smallskip
\subsection{Resurgence and the Stokes matrix}

We now run the resurgent analysis for the Bessel type, in parallel with the Airy case.
Since the Stokes phase is \(\theta_*=\frac{\pi}{2}\), it is convenient to rotate the action by this phase and set
\[
\mathcal A(w):=e^{-i\pi/2}S(w)=-\,i\,\sinh w.
\]
Then
\[
\mathcal A(w_+)=1,\qquad \mathcal A(w_-)=-1,
\qquad
\mathcal A(w_+)-\mathcal A(w_-)=2.
\]

For the two generic phases adjacent to \(\theta_*=\frac{\pi}{2}\), the four thimble integrals
\[
I_{\J_+^{<}}(\hbar),\quad I_{\J_-^{<}}(\hbar),\quad I_{\J_+^{>}}(\hbar),\quad I_{\J_-^{>}}(\hbar)
\]
admit saddle-point expansions for \(\hbar\to0\) in the corresponding adjacent sectors. Since
\(\J_+^{<}\) and \(\J_+^{>}\) have the same local germ at \(w_+\), and likewise
\(\J_-^{<}\) and \(\J_-^{>}\) at \(w_-\), there are only two formal objects:
\[
I_{\J_+^{\gtrless}}(\hbar)\sim \widetilde I_+(\hbar),
\qquad
I_{\J_-^{\gtrless}}(\hbar)\sim \widetilde I_-(\hbar).
\]

With the local orientations fixed as in Figure~\ref{fig:bessel-JK-two-panels}, one obtains
\[
\widetilde I_+(\hbar)=\sqrt{2\pi\hbar}\,e^{-1/\hbar}\,\widetilde\phi_+(\hbar),
\qquad
\widetilde I_-(\hbar)=i\sqrt{2\pi\hbar}\,e^{1/\hbar}\,\widetilde\phi_-(\hbar),
\]
where
\[
\widetilde\phi_+(\hbar)=\sum_{m\ge0}(-1)^m a_m\hbar^m,
\qquad
\widetilde\phi_-(\hbar)=\sum_{m\ge0}a_m\hbar^m,
\]
and
\[
a_m=\frac{((2m-1)!!)^2}{m!\,8^m}
=\frac{\left(\frac12\right)_m^2}{2^m m!}.
\]
These formulas are obtained by the standard saddle-point substitutions
\[
w=w_++\sqrt{\hbar}\,t
\qquad\text{and}\qquad
w=w_-+i\sqrt{\hbar}\,t,
\]
followed by expansion of \(\cosh(\sqrt{\hbar}\,t)\) and \(\cos(\sqrt{\hbar}\,t)\). By using the ordinary Borel transform similar to that for equation \eqref{equ:ordinaryBorel}, we have
\[
\widehat\phi_+(\xi)
=
{}_2F_1\!\left(\frac12,\frac12;1;-\frac{\xi}{2}\right),
\qquad
\widehat\phi_-(\xi)
=
{}_2F_1\!\left(\frac12,\frac12;1;\frac{\xi}{2}\right).
\]
Hence \(\widehat\phi_-\) has its first singularity at \(\xi=2\), while \(\widehat\phi_+\) has its first singularity at \(\xi=-2\). The singularities of these hypergeometric germs are logarithmic. Thus
\[
\sqrt{2\pi\hbar}\,\widetilde\phi_+,\quad
i\sqrt{2\pi\hbar}\,\widetilde\phi_-
\in \widetilde{\mathcal R}^{\mathrm{int}}.
\]
More precisely, there exists a holomorphic germ \(H_-(\xi)\) near \(\xi=2\) such that
\[
\widehat\phi_-(\xi)
=
H_-(\xi)
-\frac{2i}{2\pi i}\log\!\left(\xi-2\right)\widehat\phi_+(\xi-2).
\]
Therefore, with the convention of Section~\ref{sec:general} for the alien operator,
\[
\Delta_{2}^{+}\widetilde\phi_-=-\,2i\,\widetilde\phi_+,
\qquad
\Delta_{2}^{+}\widetilde\phi_+=0.
\]

Similarly, \(\widehat\phi_+\) has its first singularity at \(\xi=-2\). More precisely, there exists a holomorphic germ \(H_+(\xi)\) near \(\xi=-2\) such that
\[
\widehat\phi_+(\xi)
=
H_+(\xi)
-\frac{2i}{2\pi i}\log\!\left(\xi+2\right)\widehat\phi_-(\xi+2).
\]
Therefore
\[
\Delta_{-2}^{+}\widetilde\phi_+ = -\,2i\,\widetilde\phi_-,
\qquad
\Delta_{-2}^{+}\widetilde\phi_- = 0.
\]

We thus have
\[
\begin{split}
    \dot\Delta_{2}^{+}\widetilde I_-
&=
e^{-2/\hbar}\Delta_{2}^{+}
\Bigl(i\sqrt{2\pi\hbar}\,e^{1/\hbar}\widetilde\phi_-(\hbar)\Bigr)
=
2\,\widetilde I_+(\hbar),
\qquad
\dot\Delta_{2}^{+}\widetilde I_+=0.
\\
\dot\Delta_{-2}^{+}\widetilde I_+
&=
e^{2/\hbar}\Delta_{-2}^{+}
\Bigl(\sqrt{2\pi\hbar}\,e^{-1/\hbar}\widetilde\phi_+(\hbar)\Bigr)
=
-2\,\widetilde I_-(\hbar), \qquad
\dot\Delta_{-2}^{+}\widetilde I_-=0.
\end{split}
\]
The Bessel critical-value graded vector space, for the rotated action
\(\mathcal A=-iS\), is $
\mathcal M_{\mathrm{Bes}}
=
\C\,\widetilde I_+\oplus \C\,\widetilde I_-$.
Hence, if we write
\[
\widetilde{\mathbf I}:=(\widetilde I_+,\widetilde I_-),
\]
then the Stokes automorphism along the positive Borel ray is
\[
\mathfrak S^{+}\widetilde{\mathbf I}
=
\widetilde{\mathbf I}
\begin{pmatrix}
1&2\\
0&1
\end{pmatrix}.
\]

By \eqref{Stokesauto} and Proposition~\ref{prop:thimble-lateral-sum}, this gives
\[
\mathbf I^{<}=\mathbf I^{>}
\begin{pmatrix}
1&2\\
0&1
\end{pmatrix},
\qquad
\mathbf I^{>}=\mathbf I^{<}
\begin{pmatrix}
1&-2\\
0&1
\end{pmatrix},
\]
where
\[
\mathbf I^{<}:=\bigl(I_{\J_+^{<}},I_{\J_-^{<}}\bigr),
\qquad
\mathbf I^{>}:=\bigl(I_{\J_+^{>}},I_{\J_-^{>}}\bigr).
\]
Since the integral is linear in the cycle, the same coefficients govern the wall-crossing of the thimbles:
\[
(\J_+^<\ \ \J_-^<)
=
(\J_+^>\ \ \J_-^>)
\begin{pmatrix}
1&2\\
0&1
\end{pmatrix}, \qquad
(\J_+^>\ \ \J_-^>)
=
(\J_+^<\ \ \J_-^<)
\begin{pmatrix}
1&-2\\
0&1
\end{pmatrix}.
\]
This agrees with \eqref{equ:Besseljump}. In particular, the resurgent Stokes coefficient is \(2\), exactly matching the two connecting trajectories found in Section \ref{sec:BesselLef}.

\medskip
\section{The Gamma model}\label{sec:gamma-model}

We now turn to a third example, whose main new feature is that the holomorphic action has infinitely many critical points. Following \cite[Appendix~C]{HarlowMaltzWitten2011}, we consider
\[
S(\phi):=e^\phi-\phi,
\qquad \phi\in\C,
\]
and the associated exponential integral
\[
I_{\mathcal C}(\hbar):=\int_{\mathcal C} e^{-S(\phi)/\hbar}\,\dd\phi.
\]
Here \(\hbar\in\C^\times\) is small and \(\mathcal C\) is an oriented contour in the \(\phi\)-plane.

\begin{rmk}\label{rmk:gamma-model-contour}
For the distinguished contour \(\mathcal C=\R\) and $\hbar>0$, one has
\begin{equation*}\label{eq:gammaisgamma}
I_{\R}(\hbar)
=
\int_{-\infty}^{\infty} e^{-(e^\phi-\phi)/\hbar}\,\dd\phi
=
\hbar^{1/\hbar}\int_{0}^{\infty} t^{\frac{1}{\hbar}-1}e^{-t}\,\dd t
=
\hbar^{1/\hbar}\Gamma(1/\hbar),
\end{equation*}
where we use the change of variables \(t=e^\phi/\hbar\). Thus, up to the elementary prefactor \(\hbar^{1/\hbar}\), the present model is exactly the Gamma function.
\end{rmk}

The critical points of  $S(\phi)$ are
\[\{p_n:=2\pi i n\mid n\in\Z\}, \]
which are all non-degenerate. The corresponding critical values are
\[
S_n:=S(p_n)=1-2\pi i n,
\]
so the critical values form an infinite arithmetic progression in the imaginary direction.

For any two distinct critical points \(p_n\) and \(p_m\), one has
\[
S_n-S_m=-2\pi i(n-m).
\]
It follows that the Stokes phases are
\[
\theta_*=\pm\frac{\pi}{2}.
\]
The case \(\theta_*=-\frac{\pi}{2}\) is obtained from the case
\(\theta_*=\frac{\pi}{2}\) by reversing the relevant ordering of the critical values, so below we analyze \(\theta_*=\frac{\pi}{2}\) in detail.

\smallskip
\subsection{Lefschetz thimbles near the Stokes phase}
In real coordinates $\phi=x+iy$, $S(\phi)$ reads
\[
S(\phi)=e^x\cos y-x+i(e^x\sin y-y),
\]
hence
\begin{equation*}\label{eq:Gamma-FG}
\begin{aligned}
F_\theta(x,y)&=e^x\cos(y-\theta)-x\cos\theta-y\sin\theta,\\
G_\theta(x,y)&=e^x\sin(y-\theta)+x\sin\theta-y\cos\theta.
\end{aligned}    
\end{equation*}
With respect to the standard hermitian metric, the negative gradient flow is
\begin{equation*}\label{eq:gamma-flow-real-general}
\begin{cases}
\dot x=-e^x\cos(y-\theta)+\cos\theta,
\\[0.1cm]
\dot y=\,e^x\sin(y-\theta)+\sin\theta.
\end{cases}
\end{equation*}
Now we study the Lefschetz thimble and its dual thimble attached to the critical point \(p_n=2\pi i n\). 
\begin{itemize}
\item The stable and unstable manifolds through \(p_n\) are the connected components of the level set
\begin{equation*}\label{eq:gamma-critical-level-generic}
G_\theta(x,y)=e^x\sin(y-\theta)+x\sin\theta-y\cos\theta=-\sin\theta-2\pi n\cos\theta
\end{equation*}
selected by the local stable and unstable directions. 
\item Linearizing the flow equation at $p_n$, we obtain
\[
\dot\eta=-e^{i\theta}\overline{\eta}
\qquad
(\phi\sim p_n+\eta),
\]
the local stable and unstable directions are
\begin{equation*}\label{eq:gamma-local-directions}
\text{stable:} \quad \arg\eta=\frac{\theta}{2}\pmod{\pi},
\qquad 
\text{unstable:} \quad  \arg\eta=\frac{\theta}{2}+\frac{\pi}{2}\pmod{\pi}.
\end{equation*}
\end{itemize}

\begin{rmk}
This local picture of $p_m$ is the same as the one of $p_n$, up to vertical translation by \(2\pi (n-m)\). This is due to the translation invairance of the flow equation under $y\to y+2\pi$.    
\end{rmk}

It remains to identify the ends of these branches. We have the following result.

\begin{prop}
\label{prop:gamma-generic-endpoints-W}
Let
\[
\theta_-:=\frac{\pi}{2}-\delta,
\qquad
\theta_+:=\frac{\pi}{2}+\delta,
\qquad
0<\delta\ll1.
\]
For \(p_n=2\pi i n\), the stable and unstable manifolds have the following ends.
\begin{enumerate}
\item For \(\theta=\theta_-\), the asymptotics are given by
\[
W^s_{\theta_-}(p_n):
\quad
\begin{cases}
x\to-\infty,\quad y=x\tan\theta_-+\tan\theta_-+2\pi n+o(1),\\[0.1cm]
x\to+\infty,\quad y=\theta_-+2\pi n+O(xe^{-x}),
\end{cases}
\]
and
\[
W^u_{\theta_-}(p_n):
\quad
\begin{cases}
x\to+\infty,\quad y=\theta_- -\pi+2\pi n+O(xe^{-x}),\\[0.1cm]
x\to+\infty,\quad y=\theta_- +\pi+2\pi n+O(xe^{-x}).
\end{cases}
\]
\item For \(\theta=\theta_+\), the asymptotics are given by
\[
W^s_{\theta_+}(p_n):
\quad
\begin{cases}
x\to+\infty,\quad y=\theta_+ +2\pi n+O(xe^{-x}),\\[0.1cm]
x\to+\infty,\quad y=\theta_+ +2\pi(n-1)+O(xe^{-x}),
\end{cases}
\]
and
\[
W^u_{\theta_+}(p_n):
\quad
\begin{cases}
x\to-\infty,\quad y=x\tan\theta_+ +\tan\theta_+ +2\pi n+o(1),\\[0.1cm]
x\to+\infty,\quad y=\theta_+-\pi+2\pi n+O(xe^{-x}).
\end{cases}
\]
\end{enumerate}
\end{prop}

\begin{proof} 
\noindent\textbf{The case $x\rightarrow+\infty$}:

Fix \(n\in\mathbb Z\).  Along the stable and unstable manifolds of
\(p_n=2\pi i n\), the function \(G_\theta\) is constant and equal to
\(G_\theta(p_n)\).  Hence the relevant level equation is
\begin{equation}\label{eq:gamma-level-right-proof}
H_{\theta,n}(x,y):=
G_\theta(x,y)-G_\theta(p_n)
=
e^x\sin(y-\theta)+(x+1)\sin\theta-(y-2\pi n)\cos\theta
=0 .
\end{equation}

By its initial direction, the zero levels set enters the region
\[
  R_{n,\theta}
  :=
  \{\,x>0,\ 2\pi n<y<2\pi n+\theta\,\}.
\]
\begin{itemize}
    \item On the lower horizontal boundary \(y=2\pi n\), for our $\theta$, we have
\[
  H_{\theta,n}(x,2\pi n)
  =
  \sin\theta\,(1+x-e^x)<0,\quad x>0.
\]
    \item On the upper horizontal boundary \(y=2\pi n+\theta\), we have
\[
  H_{\theta,n}(x,2\pi n+\theta)
  =
  (x+1)\sin\theta-\theta\cos\theta .
\]
If \(\theta=\theta_+\), then \(\cos\theta<0\), so the right-hand side is
strictly positive. If \(\theta=\theta_-\), then \(\cos\theta>0\), but
\[
  \sin\theta-\theta\cos\theta>0,\quad 0<\theta<\pi/2.
\] 
Hence
\[
  H_{\theta,n}(x,2\pi n+\theta)>0,
  \qquad x\geq 0 .
\]
\item On the  vertical side \(x=0\), it is easy to check 
\[
  H_{\theta,n}(0,y)>0,
  \qquad 2\pi n<y<2\pi n+\theta.
\]
\end{itemize}
By the above analysis, the right-up branch
cannot leave \(R_{n,\theta}\), otherwise it would have to cross to one of the above boundary pieces. In particular \(y=O(1)\)
along that right end. Hence
\eqref{eq:gamma-level-right-proof} gives
\[
e^x\sin(y-\theta)
=
-(x+1)\sin\theta+(y-2\pi n)\cos\theta
=
O(x).
\]
Thus
\[
\sin(y-\theta)=O(xe^{-x}).
\]
Consequently, on any such right end, \(y-\theta\) must approach a zero of
\(\sin\), and hence there exists \(k\in\mathbb Z\) such that
\[
y=\theta+k\pi+\varepsilon(x),
\qquad
\varepsilon(x)\to 0
\qquad
(x\to+\infty).
\]

Substituting this form into \eqref{eq:gamma-level-right-proof}, we obtain
\[
e^x\sin(k\pi+\varepsilon)
+(x+1)\sin\theta
-(\theta+k\pi+\varepsilon-2\pi n)\cos\theta
=0.
\]
Since
\[
\sin(k\pi+\varepsilon)=(-1)^k\sin\varepsilon
=
(-1)^k\left(\varepsilon+O(\varepsilon^3)\right),
\]
the equation becomes
\[
\bigl((-1)^k e^x-\cos\theta\bigr)\varepsilon
+
(x+1)\sin\theta
-
(\theta+k\pi-2\pi n)\cos\theta
+
O(e^x\varepsilon^3)
=0.
\]
The preliminary estimate \(\sin(y-\theta)=O(xe^{-x})\) gives
\[
\varepsilon(x)=O(xe^{-x}),
\]
and hence
\[
e^x\varepsilon^3=O(x^3e^{-2x}).
\]

Solving the previous equation for \(\varepsilon\), and using
\[
\frac{1}{(-1)^k e^x-\cos\theta}
=
(-1)^k e^{-x}+O(e^{-2x}),
\]
we obtain
\begin{equation*}\label{eq:gamma-right-epsilon-proof}
\varepsilon(x)
=
(-1)^k e^{-x}
\left(
-(x+1)\sin\theta
+
(\theta+k\pi-2\pi n)\cos\theta
\right)
+
O(x^2e^{-2x}).
\end{equation*}
Therefore every right end of \(H_{\theta,n}=0\) in the \(k\)-th horizontal
strip has the asymptotic expansion
\begin{equation}\label{eq:gamma-right-end-proof}
y
=
\theta+k\pi
+
(-1)^k e^{-x}
\left(
-(x+1)\sin\theta
+
(\theta+k\pi-2\pi n)\cos\theta
\right)
+
O(x^2e^{-2x}).
\end{equation}
In particular,
\[
y=\theta+k\pi+O(xe^{-x}),
\qquad x\to+\infty.
\]

It remains to identify whether such a right end is stable or unstable.  Recall that we have
\[
F_\theta(x,y)
=
e^x\cos(y-\theta)-x\cos\theta-y\sin\theta.
\]
Since along \eqref{eq:gamma-right-end-proof}, \(y-\theta=k\pi+\varepsilon(x)\) and \(\varepsilon(x)=O(xe^{-x})\),
\[
\cos(y-\theta)
=
\cos(k\pi+\varepsilon)
=
(-1)^k+O(\varepsilon^2).
\]
Therefore
\[
F_\theta(x,y)
=
(-1)^k e^x+O(x).
\]
Hence, if \(k\) is even, then
\[
F_\theta(x,y)\to+\infty,
\]
so this is a stable, rapid-decay end.  If \(k\) is odd, then
\[
F_\theta(x,y)\to-\infty,
\]
so this is an unstable end.

Thus the stable right ends are precisely the even horizontal asymptotic lines
\[
y=\theta+2\pi \ell+O(xe^{-x}),
\qquad \ell\in\mathbb Z,
\]
whereas the unstable right ends are precisely the odd horizontal asymptotic
lines
\[
y=\theta+(2\ell+1)\pi+O(xe^{-x}),
\qquad \ell\in\mathbb Z.
\]
Specializing these alternatives to the two adjacent strips containing the
branches issuing from \(p_n\) gives exactly the \(x\to+\infty\) endpoint
statements in this proposition.

\medskip
\noindent\textbf{The case $x\rightarrow-\infty$:}

Fix the critical point \(p_n=2\pi i n\).  As before, the stable and unstable
manifolds through \(p_n\) are contained in the level set \eqref{eq:gamma-level-right-proof}.

We now consider the left end, namely an end on which \(x\to-\infty\).  Since
\(\theta=\theta_\pm=\frac{\pi}{2}\pm\delta\) with \(0<\delta\ll1\), we have
\(\cos\theta\neq0\).  Solving \eqref{eq:gamma-level-right-proof} for \(y\), we get
\begin{equation*}\label{eq:gamma-left-solved-proof}
y
=
x\tan\theta
+
\tan\theta
+
2\pi n
+
\frac{e^x}{\cos\theta}\sin(y-\theta).
\end{equation*}
The last term is \(O(e^x)\), because \(\sin(y-\theta)\) is bounded and
\(\theta\) is fixed away from \(\frac{\pi}{2}\) by the small nonzero parameter
\(\delta\).  Therefore
\begin{equation}\label{eq:gamma-left-asymptotic-proof}
y
=
x\tan\theta
+
\tan\theta
+
2\pi n
+
O(e^x),
\qquad x\to-\infty.
\end{equation}
In particular,
\[
y
=
x\tan\theta
+
\tan\theta
+
2\pi n
+
o(1),
\qquad x\to-\infty.
\]
This proves the claimed left-end asymptotic line.

For completeness, let us also record the corresponding asymptotic behavior of
\(F_\theta\), since this determines whether the left end is stable or unstable.
Along the left end, \eqref{eq:gamma-left-asymptotic-proof} gives
\[
y
=
x\tan\theta
+
\tan\theta
+
2\pi n
+
O(e^x).
\]
Substituting this into
\[
F_\theta(x,y)
=
e^x\cos(y-\theta)
-
x\cos\theta
-
y\sin\theta,
\]
we obtain
\[
\begin{aligned}
F_\theta(x,y)
&=
e^x\cos(y-\theta)
-
x\cos\theta
-
\left(
x\tan\theta+\tan\theta+2\pi n+O(e^x)
\right)\sin\theta  \\
&=
-
x\left(
\cos\theta+\frac{\sin^2\theta}{\cos\theta}
\right)
+O(1)=
-\frac{x}{\cos\theta}+O(1).
\end{aligned}
\]
Hence, as \(x\to-\infty\),
\[
F_\theta(x,y)\to+\infty
\quad\text{if}\quad
\cos\theta>0,
\]
whereas
\[
F_\theta(x,y)\to-\infty
\quad\text{if}\quad
\cos\theta<0.
\]
Since \(\cos\theta_->0\) and \(\cos\theta_+<0\), the left end is a stable end
for \(\theta=\theta_-\), and an unstable end for \(\theta=\theta_+\).

Therefore the unique left asymptotic line of the level set through \(p_n\) is
\[
x\to-\infty,
\qquad
y
=
x\tan\theta
+
\tan\theta
+
2\pi n
+
o(1),
\]
and its stable/unstable type is determined by the sign of \(\cos\theta\), as
claimed in this proposition.
\end{proof}

The two adjacent generic Gamma configurations are shown in the first and third panels of Figure~\ref{fig:gamma-thimbles-three-phases}.

\smallskip
\subsection{Connecting trajectories and the Picard--Lefschetz jump formula}

At the Stokes phase \(\theta_*=\frac{\pi}{2}\), we prove the following.

\begin{prop}[Connecting trajectories at the Stokes phase]\label{prop:gamma-neighboring-trajectories}
At the Stokes phase \(\theta_*=\frac{\pi}{2}\), the direct connecting trajectories are exactly the neighboring ones
\[
p_n\longrightarrow p_{n+1},
\qquad n\in\Z.
\]
Moreover, for each \(n\in\Z\), there is exactly one such trajectory, up to translation in the flow parameter \(s\).
\end{prop}

\begin{proof}
At \(\theta_*=\frac{\pi}{2}\), the flow equation is
\begin{equation}\label{eq:gamma-flow-pi2}
    \dot x=-e^x\sin y,
\qquad
\dot y=1-e^x\cos y.
\end{equation}

Now the level set condition
\[
G_{\pi/2}(x,y)=x-e^x\cos y=G_{\pi/2}(p_n)=-1\quad \forall~n\in\Z
\]
means
\begin{equation}\label{eq:gamma-critical-level-proof}
x+1=e^x\cos y.
\end{equation}

We first analyze the strip
\[
0\le y\le 2\pi.
\]
We claim that the compact component of \eqref{eq:gamma-critical-level-proof} joining
\[
p_0=(0,0)
\qquad\text{and}\qquad
p_1=(0,2\pi)
\]
is given by\footnote{Here \(W_0\) denotes the principal branch of the Lambert \(W\)-function, namely the branch obtained from the unique holomorphic inverse of \(u\mapsto ue^u\) near \(u=0\), normalized by \(W_0(0)=0\). Thus \(W_0(z)e^{W_0(z)}=z\). On the real interval \((-e^{-1},\infty)\), \(W_0\) is real-valued and extends continuously to \(-e^{-1}\) with \(W_0(-e^{-1})=-1\).}
\begin{equation}\label{eq:gamma-chi-branch}
x=\chi(y):=-1-W_0\bigl(-e^{-1}\cos y\bigr),
\qquad 0\le y\le 2\pi.
\end{equation}
One can compute
\[
\chi(0)=\chi(2\pi)=0, \qquad \text{and} \qquad 
\chi(y)<0 \ \ \text{as} \ \ 0<y<2\pi.
\]
Thus \eqref{eq:gamma-chi-branch} is a compact arc in the left half-plane joining \(p_0\) to \(p_1\). It is straightforward to show that \(\phi(s)=\chi(y(s))+iy(s)\) is a solution of the flow equation \eqref{eq:gamma-flow-pi2}. Thus the compact component \eqref{eq:gamma-chi-branch} gives one direct trajectory
\[
p_0\longrightarrow p_1.
\]

The flow equation \eqref{eq:gamma-flow-pi2} is invariant under vertical translations
\[
\phi\longmapsto \phi+2\pi i n,
\qquad n\in\Z,
\]
because \(e^{\phi+2\pi i n}=e^\phi\). Hence translating the above solution gives, for every \(n\in\Z\), a direct trajectory
\[
\phi_n(s):=2\pi i n+\phi(s)
\]
from $p_n$ to $p_{n+1}$. Thus every neighboring pair \(p_n,p_{n+1}\) is connected by exactly one direct trajectory, up to translation in the flow parameter \(s\).
It remains to exclude direct trajectories between non-neighboring critical points.  By the local normal form discussed in Section~\ref{sec:general}, the level set \(G_{\pi/2}=-1\) has only two unstable half-branches and two stable half-branches at each critical point \(p_n\).  The unstable half-branch leaving \(p_n\) toward the left is precisely the compact trajectory constructed above, namely the one ending at \(p_{n+1}\); the other unstable half-branch is noncompact and escapes to \(x\to+\infty\).  Hence there is no remaining local branch along which a direct trajectory from \(p_n\) to a non-neighboring critical point could start.
\end{proof}
The Gamma configurations of the connecting trajectories are shown in the second panels of Figure~\ref{fig:gamma-thimbles-three-phases}.

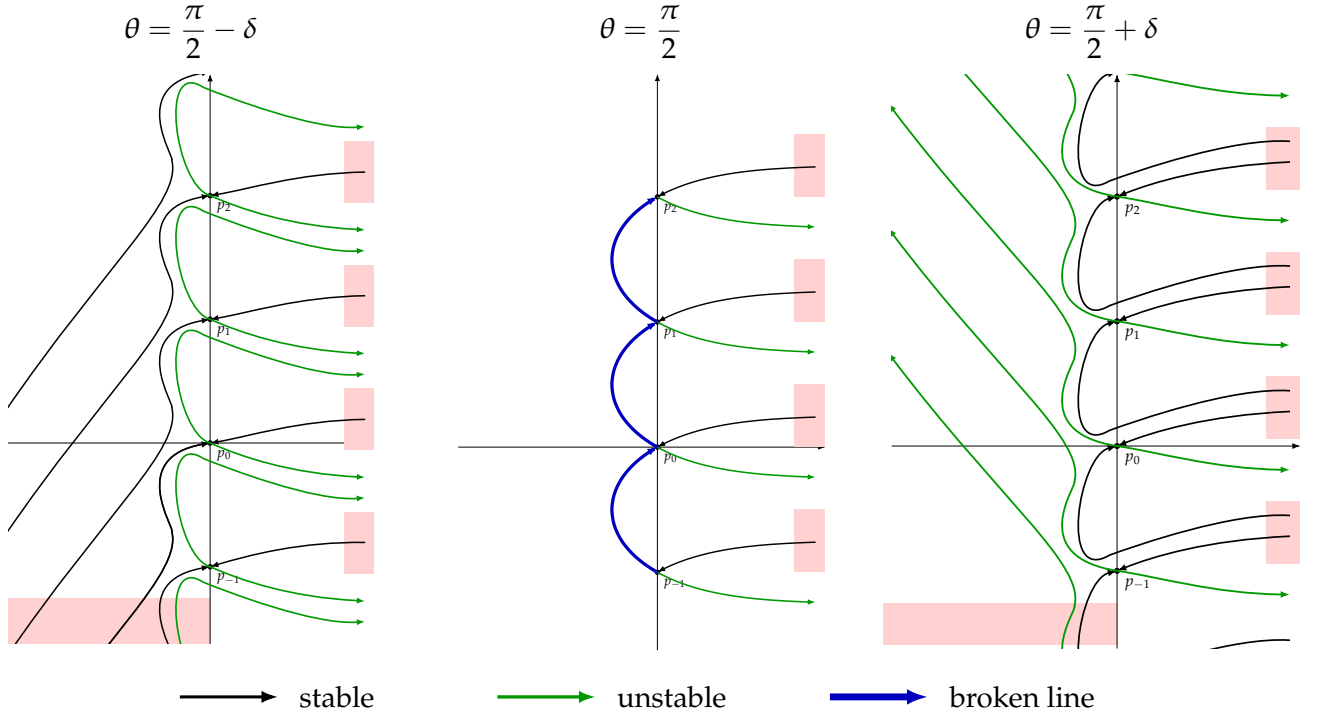
\begin{figure}[H]
\centering
\tikzset{
  stab/.style={
    black,
    line width=1.15pt,
       -{Latex[length=2.2mm]}
},
  unstab/.style={
    green!60!black,
    line width=1.25pt,
    -{Latex[length=2.2mm]}},
  broken/.style={
    blue!75!black,
    line width=2.6pt,
    -{Latex[length=2.8mm]}
}
}

\colorlet{sectorred}{red!18}
\begin{minipage}[t]{0.32\textwidth}
\centering
\[
\theta=\frac{\pi}{2}-\delta
\]
\resizebox{0.88\linewidth}{!}{
\begin{tikzpicture}[x=0.92cm,y=0.52cm,>=Latex]

  \def\xL{-5.8}
  \def\xR{4.7}
  \def\T{6.283}
  \def\th{1.20}
  \def\hp{1.571}

\clip (\xL,-10.2) rectangle (\xR,18.7);
\fill[sectorred] (-8.6,-10.9)rectangle (0,-7.9);
\draw[->, thin] (\xL,0) -- (\xR,0) node[right] {$x$};
\draw[->, thin] (0,-10.2) -- (0,18.7) node[above] {$y$};

\foreach \k in {-1,0,1,2} {
    \pgfmathsetmacro{\yc}{\th+\k*\T}
    \fill[sectorred] (\xR-0.85,\yc-\hp) rectangle (\xR,\yc+\hp);
  }

\filldraw[black] (0,-6.283) circle (1.7pt);
\filldraw[black] (0,0) circle (1.7pt);
\filldraw[black] (0,6.283) circle (1.7pt);
\filldraw[black] (0,12.566) circle (1.7pt);

\node[below right=1pt] at (0,-6.283) {$p_{-1}$};
\node[below right=1pt] at (0,0) {$p_0$};
\node[below right=1pt] at (0,6.283) {$p_1$};
\node[below right=1pt] at (0,12.566) {$p_2$};

\draw[stab]
    (4.45,-5.05) .. controls (2.6,-5.00) and (1.1,-5.70) .. (0,-6.283);
\draw[stab]
    (4.45,1.20) .. controls (2.5,1.15) and (1.1,0.32) .. (0,0);
\draw[stab]
    (-5.9,-11-6.283) .. controls (-3,-3.5-6.283) and (-0.8,-0.3-6.283) .. (-1.1,1.8-6.283).. controls (-1,1.9-6.283) and (-2.4,5.55-6.283) .. (0,6.283-6.283);
\draw[stab]
    (4.45,7.48) .. controls (2.5,7.42) and (1.1,6.60) .. (0,6.283);
\draw[stab]
    (4.45,13.76) .. controls (2.5,13.70) and (1.1,12.88) .. (0,12.566);
\draw[stab]
    (-5.9,-11) .. controls (-3,-3.5) and (-0.8,-0.3) .. (-1.1,1.8).. controls (-1,1.9) and (-2.4,5.55) .. (0,6.283);
\draw[stab]
    (-5.9,-11+6.283) .. controls (-3,-3.5+6.283) and (-0.8,-0.3+6.283) .. (-1.1,1.8+6.283).. controls (-1,1.9+6.283) and (-2.4,5.55+6.283) .. (0,6.283+6.283);
\draw[stab]
    (-5.9,-11-6.283) .. controls (-3,-3.5-6.283) and (-0.8,-0.3-6.283) .. (-1.1,1.8-6.283).. controls (-1,1.9-6.283) and (-2.4,5.55-6.283) .. (0,6.283-6.283);
\draw[stab]
    (-5.9,-11-6.283-6.283) .. controls (-3,-3.5-6.283-6.283) and (-0.8,-0.3-6.283-6.283) .. (-1.1,1.8-6.283-6.283).. controls (-1,1.9-6.283-6.283) and (-2.4,5.55-6.283-6.283) .. (0,6.283-6.283-6.283);
\draw[stab]
    (-5.9,-11+6.283+6.283) .. controls (-3,-3.5+6.283+6.283) and (-0.8,-0.3+6.283+6.283) .. (-1.1,1.8+6.283+6.283).. controls (-1,1.9+6.283+6.283) and (-2.4,5.55+6.283+6.283) .. (0,6.283+6.283+6.283);

\draw[unstab]
    (0,0) .. controls (-1,0.5) and (-1.5,7) .. (-0.2,5.5) ..  controls (0.5,5) and (3,3.3) .. (4.45,3.5);
\draw[unstab]
    (0,0) .. controls (0.9,-0.70) and (2.3,-1.55) .. (4.45,-1.74);
\draw[unstab]
    (0,-12.566) .. controls (-1,0.5-12.566) and (-1.5,7-12.566) .. (-0.2,5.5-12.566) ..  controls (0.5,5-12.566) and (3,3.3-12.566) .. (4.45,3.5-12.566);
\draw[unstab]
    (0,0-12.566) .. controls (0.9,-0.70-12.566) and (2.3,-1.55-12.566) .. (4.45,-1.74-12.566);

\draw[unstab]
  (0,6.283) .. controls (-1,6.783) and (-1.5,13.283) .. (-0.2,11.783) .. controls (0.5,11.283) and (3,9.583) .. (4.45,9.783);
\draw[unstab]
  (0,6.283) .. controls (0.9,5.583) and (2.3,4.733) .. (4.45,4.543);
\draw[unstab]
  (0,-6.283) .. controls (-1,-5.783) and (-1.5,0.717) .. (-0.2,-0.783) .. controls (0.5,-1.283) and (3,-2.983) .. (4.45,-2.783);
\draw[unstab]
  (0,-6.283) .. controls (0.9,-6.983) and (2.3,-7.833) .. (4.45,-8.023);
\draw[unstab]
  (0,12.566) .. controls (-1,13.066) and (-1.5,19.566) .. (-0.2,18.066) .. controls (0.5,17.566) and (3,15.866) .. (4.45,16.066);
\draw[unstab]
  (0,12.566) .. controls (0.9,11.866) and (2.3,11.016) .. (4.45,10.826);

\end{tikzpicture}
}
\end{minipage}
\hfill
\begin{minipage}[t]{0.32\textwidth}
\centering
\[
\theta=\frac{\pi}{2}
\]
\resizebox{0.88\linewidth}{!}{
\begin{tikzpicture}[x=1.0cm,y=0.56cm,>=Latex]

  \def\xL{-5.6}
  \def\xR{4.7}
  \def\T{6.283}
  \def\th{1.571}
  \def\hp{1.571}

  \clip (\xL,-10.2) rectangle (\xR,18.7);

\draw[->, thin] (\xL,0) -- (\xR,0) node[right] {$x$};
\draw[->, thin] (0,-10.2) -- (0,18.7) node[above] {$y$};

\foreach \k in {-1,0,1,2} {
    \pgfmathsetmacro{\yc}{\th+\k*\T}
    \fill[sectorred] (\xR-0.85,\yc-\hp) rectangle (\xR,\yc+\hp);
  }

\filldraw[black] (0,-6.283) circle (1.7pt);
\filldraw[black] (0,0.000) circle (1.7pt);
\filldraw[black] (0,6.283) circle (1.7pt);
\filldraw[black] (0,12.566) circle (1.7pt);

  \node[below right=1pt] at (0,-6.283) {$p_{-1}$};
  \node[below right=1pt] at (0,0.000) {$p_0$};
  \node[below right=1pt] at (0,6.283) {$p_1$};
  \node[below right=1pt] at (0,12.566) {$p_2$};

\draw[stab] plot[smooth] coordinates {
    (4.450,-4.776) (4.000,-4.804) (3.200,-4.884) (2.500,-5.004)
    (1.900,-5.161) (1.400,-5.346) (1.000,-5.539) (0.700,-5.717)
    (0.450,-5.892) (0.250,-6.052) (0.100,-6.186) (0.000,-6.283)
  };
\draw[stab] plot[smooth] coordinates {
    (4.450,1.507) (4.000,1.479) (3.200,1.399) (2.500,1.279)
    (1.900,1.122) (1.400,0.937) (1.000,0.744) (0.700,0.566)
    (0.450,0.391) (0.250,0.231) (0.100,0.097) (0.000,0.000)
  };
\draw[stab] plot[smooth] coordinates {
    (4.450,7.790) (4.000,7.762) (3.200,7.682) (2.500,7.563)
    (1.900,7.405) (1.400,7.221) (1.000,7.027) (0.700,6.849)
    (0.450,6.674) (0.250,6.514) (0.100,6.380) (0.000,6.283)
  };
\draw[stab] plot[smooth] coordinates {
    (4.450,14.073) (4.000,14.045) (3.200,13.965) (2.500,13.846)
    (1.900,13.689) (1.400,13.504) (1.000,13.310) (0.700,13.132)
    (0.450,12.957) (0.250,12.797) (0.100,12.663) (0.000,12.566)
  };

\draw[unstab] plot[smooth] coordinates {
    (0.000,-6.283) (0.100,-6.380) (0.250,-6.514) (0.450,-6.674)
    (0.700,-6.849) (1.000,-7.027) (1.400,-7.221) (1.900,-7.405)
    (2.500,-7.563) (3.200,-7.682) (4.000,-7.762) (4.450,-7.790)
  };
\draw[unstab] plot[smooth] coordinates {
    (0.000,0.000) (0.100,-0.097) (0.250,-0.231) (0.450,-0.391)
    (0.700,-0.566) (1.000,-0.744) (1.400,-0.937) (1.900,-1.122)
    (2.500,-1.279) (3.200,-1.399) (4.000,-1.479) (4.450,-1.507)
  };
\draw[unstab] plot[smooth] coordinates {
    (0.000,6.283) (0.100,6.186) (0.250,6.052) (0.450,5.892)
    (0.700,5.717) (1.000,5.539) (1.400,5.346) (1.900,5.161)
    (2.500,5.004) (3.200,4.884) (4.000,4.804) (4.450,4.776)
  };
\draw[unstab] plot[smooth] coordinates {
    (0.000,12.566) (0.100,12.470) (0.250,12.336) (0.450,12.175)
    (0.700,12.001) (1.000,11.822) (1.400,11.629) (1.900,11.444)
    (2.500,11.287) (3.200,11.168) (4.000,11.087) (4.450,11.059)
  };

\draw[broken] plot[smooth] coordinates {
    (-0.000,-6.283) (-0.143,-6.133) (-0.313,-5.933) (-0.499,-5.683)
    (-0.687,-5.383) (-0.844,-5.083) (-0.954,-4.833) (-1.000,-4.713)
    (-1.077,-4.483) (-1.181,-4.083) (-1.246,-3.683) (-1.276,-3.283)
    (-1.278,-3.142) (-1.276,-2.983) (-1.244,-2.583) (-1.177,-2.183)
    (-1.072,-1.783) (-0.947,-1.433) (-0.787,-1.083) (-0.587,-0.733)
    (-0.378,-0.433) (-0.173,-0.183) (-0.000,0.000)
  };
\draw[broken] plot[smooth] coordinates {
    (-0.000,0.000) (-0.143,0.150) (-0.313,0.350) (-0.499,0.600)
    (-0.687,0.900) (-0.844,1.200) (-0.954,1.450) (-1.000,1.570)
    (-1.077,1.800) (-1.181,2.200) (-1.246,2.600) (-1.276,3.000)
    (-1.278,3.142) (-1.276,3.300) (-1.244,3.700) (-1.177,4.100)
    (-1.072,4.500) (-0.947,4.850) (-0.787,5.200) (-0.587,5.550)
    (-0.378,5.850) (-0.173,6.100) (-0.000,6.283)
  };
\draw[broken] plot[smooth] coordinates {
    (-0.000,6.283) (-0.143,6.433) (-0.313,6.633) (-0.499,6.883)
    (-0.687,7.183) (-0.844,7.483) (-0.954,7.733) (-1.000,7.853)
    (-1.077,8.083) (-1.181,8.483) (-1.246,8.883) (-1.276,9.283)
    (-1.278,9.425) (-1.276,9.583) (-1.244,9.983) (-1.177,10.383)
    (-1.072,10.783) (-0.947,11.133) (-0.787,11.483) (-0.587,11.833)
    (-0.378,12.133) (-0.173,12.383) (-0.000,12.566)
  };

\end{tikzpicture}
}
\end{minipage}
\hfill
\begin{minipage}[t]{0.32\textwidth}
\centering
\[
\theta=\frac{\pi}{2}+\delta
\]
\resizebox{\linewidth}{!}{
\begin{tikzpicture}[x=0.92cm,y=0.47cm,>=Latex]

\fill[sectorred] (-6,-10)rectangle (0,-7.9);

  \def\xL{-5.8}
  \def\xR{4.7}
  \def\T{6.283}
  \def\th{1.921}
  \def\hp{1.571}

  \clip (\xL,-10.2) rectangle (\xR,18.7);

\draw[->, thin] (\xL,0) -- (\xR,0) node[right] {$x$};
  \draw[->, thin] (0,-10.2) -- (0,18.7) node[above] {$y$};
\foreach \k in {-1,0,1,2} {
    \pgfmathsetmacro{\yc}{\th+\k*\T}
    \fill[sectorred] (\xR-0.85,\yc-\hp) rectangle (\xR,\yc+\hp);
  }

\filldraw[black] (0,-6.283) circle (1.7pt);
\filldraw[black] (0, 0) circle (1.7pt);
\filldraw[black] (0, 6.283) circle (1.7pt);
\filldraw[black] (0,12.566) circle (1.7pt);

\node[below right=1pt] at (0,-6.283) {$p_{-1}$};
\node[below right=1pt] at (0,0) {$p_0$};
\node[below right=1pt] at (0,6.283) {$p_1$};
\node[below right=1pt] at (0,12.566) {$p_2$};

\draw[stab]
  (4.45,-3.5+18.848) .. controls (3,-3.3+18.848) and (0.5,-5+18.848) .. (-0.2,-5.5+18.848) .. controls (-1.5,-7+18.848) and (-1,-0.5+18.848) .. (0,0+18.848);
\draw[stab]
  (4.45,1.74+18.848) .. controls (2.3,1.55+18.848) and (0.9,0.70+18.848) .. (0,0+18.848);
\draw[stab]
  (4.45,-3.5) .. controls (3,-3.3) and (0.5,-5) .. (-0.2,-5.5) .. controls (-1.5,-7) and (-1,-0.5) .. (0,0);
\draw[stab]
  (4.45,1.74) .. controls (2.3,1.55) and (0.9,0.70) .. (0,0);
\draw[stab]
  (4.45,-22.349) .. controls (3,-22.149) and (0.5,-23.849) .. (-0.2,-24.349) .. controls (-1.5,-25.849) and (-1,-19.349) .. (0,-18.849);
\draw[stab]
  (4.45,-17.109) .. controls (2.3,-17.299) and (0.9,-18.149) .. (0,-18.849);
\draw[stab]
  (4.45,-16.066) .. controls (3,-15.866) and (0.5,-17.566) .. (-0.2,-18.066) .. controls (-1.5,-19.566) and (-1,-13.066) .. (0,-12.566);
\draw[stab]
  (4.45,-10.826) .. controls (2.3,-11.016) and (0.9,-11.866) .. (0,-12.566);
\draw[stab]
  (4.45,9.066) .. controls (3,9.266) and (0.5,7.566) .. (-0.2,7.066) .. controls (-1.5,5.566) and (-1,12.066) .. (0,12.566);
\draw[stab]
  (4.45,14.306) .. controls (2.3,14.116) and (0.9,13.266) .. (0,12.566);
\draw[stab]
  (4.45,-9.783) .. controls (3,-9.583) and (0.5,-11.283) .. (-0.2,-11.783) .. controls (-1.5,-13.283) and (-1,-6.783) .. (0,-6.283);
\draw[stab]
  (4.45,-4.543) .. controls (2.3,-4.733) and (0.9,-5.583) .. (0,-6.283);
\draw[stab]
  (4.45,2.783) .. controls (3,2.983) and (0.5,1.283) .. (-0.2,0.783) .. controls (-1.5,-0.717) and (-1,5.783) .. (0,6.283);
\draw[stab]
  (4.45,8.023) .. controls (2.3,7.833) and (0.9,6.983) .. (0,6.283);

\draw[unstab]
  (0,0) .. controls (1.1,-0.32) and (2.5,-1.15) .. (4.45,-1.20);
\draw[unstab]
  (0,0) .. controls (-2.4,0.733) and (-1,4.383) .. (-1.1,4.483)
  .. controls (-0.8,6.583) and (-3,9.783) .. (-5.9,17.283);
\draw[unstab]
  (0,6.283) .. controls (1.1,5.963) and (2.5,5.133) .. (4.45,5.083);
\draw[unstab]
  (0,6.283) .. controls (-2.4,7.016) and (-1,10.666) .. (-1.1,10.766)
  .. controls (-0.8,12.866) and (-3,16.066) .. (-5.9,23.566);
\draw[unstab]
  (0,12.566) .. controls (1.1,12.246) and (2.5,11.416) .. (4.45,11.366);
\draw[unstab]
  (0,12.566) .. controls (-2.4,13.299) and (-1,16.949) .. (-1.1,17.049)
  .. controls (-0.8,19.149) and (-3,22.349) .. (-5.9,29.849);
\draw[unstab]
  (0,18.849) .. controls (1.1,18.529) and (2.5,17.699) .. (4.45,17.649);
\draw[unstab]
  (0,18.849) .. controls (-2.4,19.582) and (-1,23.232) .. (-1.1,23.332)
  .. controls (-0.8,25.432) and (-3,28.632) .. (-5.9,36.132);
\draw[unstab]
  (0,-6.283) .. controls (1.1,-6.603) and (2.5,-7.433) .. (4.45,-7.483);
\draw[unstab]
  (0,-6.283) .. controls (-2.4,-5.550) and (-1,-1.900) .. (-1.1,-1.800)
  .. controls (-0.8,0.300) and (-3,3.500) .. (-5.9,11.000);
\draw[unstab]
  (0,-12.566) .. controls (1.1,-12.886) and (2.5,-13.716) .. (4.45,-13.766);
\draw[unstab]
  (0,-12.566) .. controls (-2.4,-11.833) and (-1,-8.183) .. (-1.1,-8.083)
  .. controls (-0.8,-5.983) and (-3,-2.783) .. (-5.9,4.717);
\end{tikzpicture}
}
\end{minipage}

\vspace{0.6em}

\begin{center}
\begin{tikzpicture}[>=Latex]
\draw[stab] (0,0) -- (1.3,0);
\node[right] at (1.45,0) {stable};
\draw[unstab] (4.2,0) -- (5.5,0);
\node[right] at (5.65,0) {unstable};
\draw[broken] (8.6,0) -- (9.9,0);
\node[right] at (10.05,0) {broken line};
\end{tikzpicture}
\end{center}

\caption{
Stable/unstable manifolds and connecting trajectories for the Gamma model
\(
S(\phi)=e^\phi-\phi
\)
at
\(
\theta=\frac{\pi}{2}-\delta,\ \frac{\pi}{2},\ \frac{\pi}{2}+\delta
\).
The pink regions indicate rapid-decay directions for admissible contours. }
\label{fig:gamma-thimbles-three-phases}
\end{figure}

Let
\[
\J_n^{<},\quad \J_n^{>},
\quad
\Kk_n^{<},\quad \Kk_n^{>},
\quad n\in\Z,
\]
denote the Lefschetz thimbles and dual thimbles attached to
\[
p_n=2\pi i n
\]
for the two one-sided phases
\[
\theta=\frac{\pi}{2}-\delta
\quad\text{and}\quad
\theta=\frac{\pi}{2}+\delta,
\qquad 0<\delta\ll1.
\]

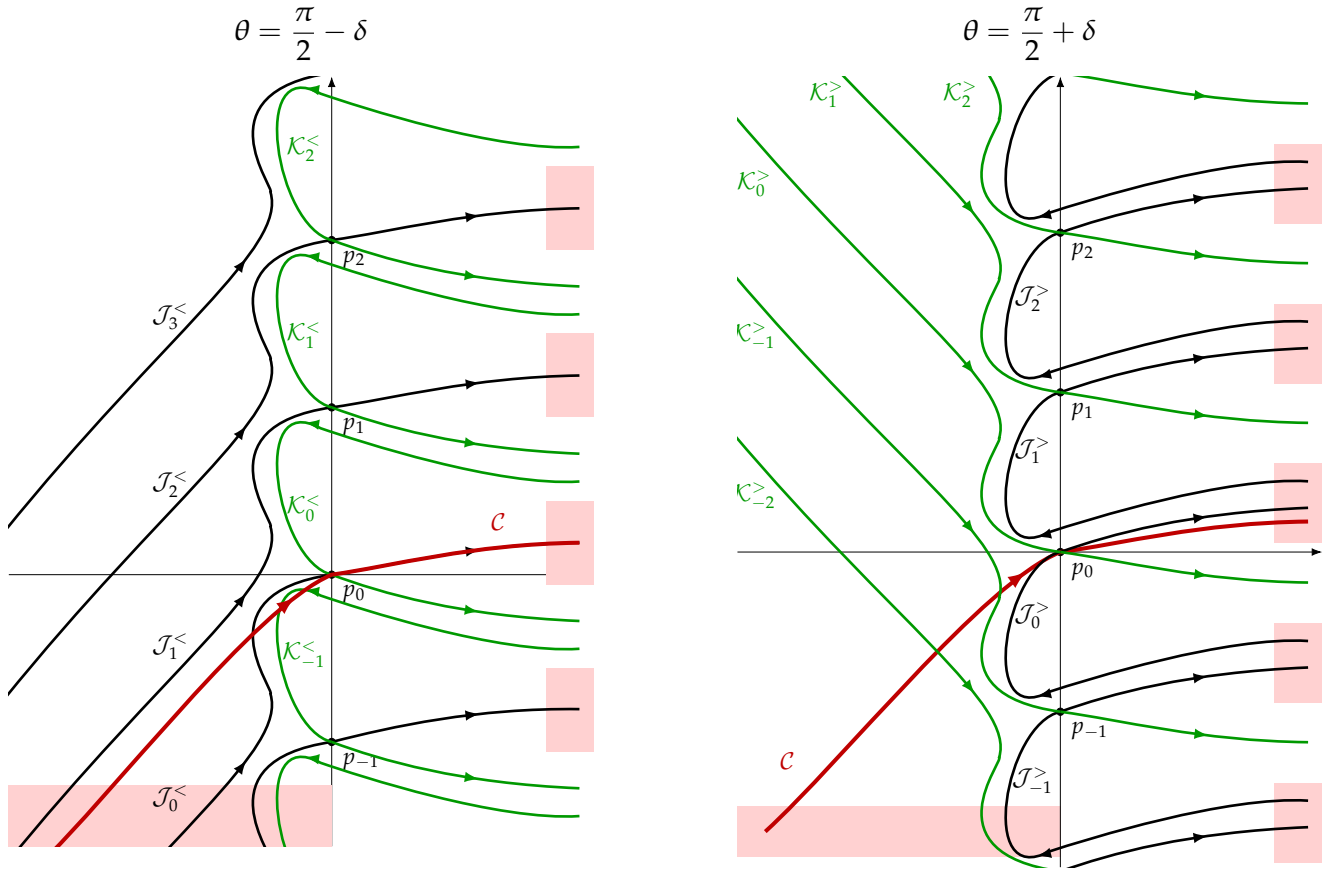
\begin{figure}[!htbp]
\centering
\tikzset{
  thimble/.style={
    black,
    line width=1.25pt,
    postaction={decorate},
    decoration={
      markings,
      mark=at position 0.60 with {\arrow{latex}}
    }
  },
  dual/.style={
    green!60!black,
    line width=1.25pt,
    postaction={decorate},
    decoration={
      markings,
      mark=at position 0.60 with {\arrow{latex}}
    }
  },
  gammaC/.style={
    red!75!black,
    line width=1.9pt,
    postaction={decorate},
    decoration={
      markings,
      mark=at position 0.56 with {\arrow{latex}}
    }
  }
}

\colorlet{sectorred}{red!18}
\begin{minipage}[t]{0.45\textwidth}
\centering
\[
\theta=\frac{\pi}{2}-\delta
\]
\resizebox{\linewidth}{!}{
\begin{tikzpicture}[x=0.92cm,y=0.44cm,>=Latex]

  \def\xL{-5.8}
  \def\xR{4.7}
  \def\T{6.283}
  \def\th{1.20}
  \def\hp{1.571}

  \clip (\xL,-10.2) rectangle (\xR,18.7);

\draw[->, thin] (\xL,0) -- (\xR,0) node[right] {$x$};
\draw[->, thin] (0,-10.2) -- (0,18.7) node[above] {$y$};

\foreach \k in {-1,0,1,2} {
    \pgfmathsetmacro{\yc}{\th+\k*\T}
    \fill[sectorred] (\xR-0.85,\yc-\hp) rectangle (\xR,\yc+\hp);
  }
\fill[sectorred] (-8.6,-10.9)rectangle (0,-7.9);
\filldraw[black] (0,-6.283) circle (1.7pt);
\filldraw[black] (0,0) circle (1.7pt);
\filldraw[black] (0,6.283) circle (1.7pt);
\filldraw[black] (0,12.566) circle (1.7pt);

\node[below right=1pt] at (0,-6.283) {$p_{-1}$};
\node[below right=1pt] at (0,0) {$p_0$};
\node[below right=1pt] at (0,6.283) {$p_1$};
\node[below right=1pt] at (0,12.566) {$p_2$};
\node[red!75!black] at (3,2) {$\mathcal{C}$};
\node[] at (-2.9,-2.7) {$\J_1^<$};
\node[] at (-2.9,3.4) {$\J_2^<$};
\node[] at (-2.9,9.7) {$\J_3^<$};
\node[] at (-2.9,-8.4) {$\J_{0}^<$};
\node[green!60!black] at (-0.5,2.5) {$\Kk_{0}^<$};
\node[green!60!black] at (-0.5,9) {$\Kk_{1}^<$};
\node[green!60!black] at (-0.5,16) {$\Kk_{2}^<$};
\node[green!60!black] at (-0.5,-3) {$\Kk_{-1}^<$};

\draw[thimble]
  (0,-6.283) .. controls (1.1,-5.70) and (2.6,-5.00) .. (4.45,-5.05);
\draw[thimble]
  (0,0) .. controls (1.1,0.32) and (2.5,1.15) .. (4.45,1.20);
\draw[thimble]
  (0,6.283) .. controls (1.1,6.60) and (2.5,7.42) .. (4.45,7.48);
\draw[thimble]
  (0,12.566) .. controls (1.1,12.88) and (2.5,13.70) .. (4.45,13.76);
\draw[thimble]
    (-5.9,-11) .. controls (-3,-3.5) and (-0.8,-0.3) .. (-1.1,1.8).. controls (-1,1.9) and (-2.4,5.55) .. (0,6.283);
\draw[thimble]
    (-5.9,-11+6.283) .. controls (-3,-3.5+6.283) and (-0.8,-0.3+6.283) .. (-1.1,1.8+6.283).. controls (-1,1.9+6.283) and (-2.4,5.55+6.283) .. (0,6.283+6.283);
\draw[thimble]
    (-5.9,-11-6.283) .. controls (-3,-3.5-6.283) and (-0.8,-0.3-6.283) .. (-1.1,1.8-6.283).. controls (-1,1.9-6.283) and (-2.4,5.55-6.283) .. (0,6.283-6.283);
\draw[thimble]
    (-5.9,-11-6.283-6.283) .. controls (-3,-3.5-6.283-6.283) and (-0.8,-0.3-6.283-6.283) .. (-1.1,1.8-6.283-6.283).. controls (-1,1.9-6.283-6.283) and (-2.4,5.55-6.283-6.283) .. (0,6.283-6.283-6.283);
    \draw[thimble]
    (-5.9,-11+6.283+6.283) .. controls (-3,-3.5+6.283+6.283) and (-0.8,-0.3+6.283+6.283) .. (-1.1,1.8+6.283+6.283).. controls (-1,1.9+6.283+6.283) and (-2.4,5.55+6.283+6.283) .. (0,6.283+6.283+6.283);

\draw[dual]
    (0,0) .. controls (0.9,-0.70) and (2.3,-1.55) .. (4.45,-1.74);
\draw[dual]
    (0,0-12.566) .. controls (0.9,-0.70-12.566) and (2.3,-1.55-12.566) .. (4.45,-1.74-12.566);
\draw[dual]
  (0,6.283) .. controls (0.9,5.583) and (2.3,4.733) .. (4.45,4.543);
\draw[dual]
  (0,-6.283) .. controls (0.9,-6.983) and (2.3,-7.833) .. (4.45,-8.023);
\draw[dual]
  (4.45,3.5) .. controls (3,3.3) and (0.5,5) .. (-0.2,5.5)
  .. controls (-1.5,7) and (-1,0.5) .. (0,0);
\draw[dual]
  (4.45,-9.066) .. controls (3,-9.266) and (0.5,-7.566) .. (-0.2,-7.066)
  .. controls (-1.5,-5.566) and (-1,-12.066) .. (0,-12.566);
\draw[dual]
  (4.45,9.783) .. controls (3,9.583) and (0.5,11.283) .. (-0.2,11.783)
  .. controls (-1.5,13.283) and (-1,6.783) .. (0,6.283);
\draw[dual]
  (4.45,-2.783) .. controls (3,-2.983) and (0.5,-1.283) .. (-0.2,-0.783)
  .. controls (-1.5,0.717) and (-1,-5.783) .. (0,-6.283);
\draw[dual]
  (4.45,16.066) .. controls (3,15.866) and (0.5,17.566) .. (-0.2,18.066)
  .. controls (-1.5,19.566) and (-1,13.066) .. (0,12.566);
\draw[dual]
  (0,12.566) .. controls (0.9,11.866) and (2.3,11.016) .. (4.45,10.826);
 
 \draw[gammaC]
    (-5.3,-11.0) .. controls (-3.9,-8.2) and (-1.4,-1.25) .. (0,0)
                 .. controls (1.1,0.32) and (2.5,1.15) .. (4.45,1.20);
\end{tikzpicture}
}
\end{minipage}
\hfill
\begin{minipage}[t]{0.45\textwidth}
\centering
\[
\theta=\frac{\pi}{2}+\delta
\]
\resizebox{\linewidth}{!}{
\begin{tikzpicture}[x=0.92cm,y=0.42cm,>=Latex]

  \def\xL{-5.8}
  \def\xR{4.7}
  \def\T{6.283}
  \def\th{1.921}
  \def\hp{1.571}

\clip (\xL,-12.4) rectangle (\xR,18.7);
\fill[sectorred] (-6,-12)rectangle (0,-10);
\draw[->, thin] (\xL,0) -- (\xR,0) node[right] {$x$};
\draw[->, thin] (0,-12.4) -- (0,18.7) node[above] {$y$};

\foreach \k in {-2,-1,0,1,2} {
    \pgfmathsetmacro{\yc}{\th+\k*\T}
    \fill[sectorred] (\xR-0.85,\yc-\hp) rectangle (\xR,\yc+\hp);
  }

\filldraw[black] (0,-6.283) circle (1.7pt);
\filldraw[black] (0,0) circle (1.7pt);
\filldraw[black] (0,6.283) circle (1.7pt);
\filldraw[black] (0,12.566) circle (1.7pt);

\node[below right=1pt] at (0,-6.283) {$p_{-1}$};
\node[below right=1pt] at (0,0) {$p_0$};
\node[below right=1pt] at (0,6.283) {$p_1$};
\node[below right=1pt] at (0,12.566) {$p_2$};

\draw[gammaC]
    (-5.3,-11.0) .. controls (-3.9,-8.2) and (-1.4,-1.25) .. (0,0)
                 .. controls (1.1,0.32) and (2.5,1.15) .. (4.45,1.20);
\draw[thimble]
  (4.45,-3.5+18.848) .. controls (3,-3.3+18.848) and (0.5,-5+18.848) .. (-0.2,-5.5+18.848) .. controls (-1.5,-7+18.848) and (-1,-0.5+18.848) .. (0,0+18.848);
\draw[thimble]
  (4.45,-3.5) .. controls (3,-3.3) and (0.5,-5) .. (-0.2,-5.5) .. controls (-1.5,-7) and (-1,-0.5) .. (0,0);
\draw[thimble]
  (4.45,-22.349) .. controls (3,-22.149) and (0.5,-23.849) .. (-0.2,-24.349) .. controls (-1.5,-25.849) and (-1,-19.349) .. (0,-18.849);
\draw[thimble]
  (4.45,-16.066) .. controls (3,-15.866) and (0.5,-17.566) .. (-0.2,-18.066) .. controls (-1.5,-19.566) and (-1,-13.066) .. (0,-12.566);
\draw[thimble]
  (4.45,9.066) .. controls (3,9.266) and (0.5,7.566) .. (-0.2,7.066) .. controls (-1.5,5.566) and (-1,12.066) .. (0,12.566);
\draw[thimble]
  (4.45,-9.783) .. controls (3,-9.583) and (0.5,-11.283) .. (-0.2,-11.783) .. controls (-1.5,-13.283) and (-1,-6.783) .. (0,-6.283);
\draw[thimble]
  (4.45,2.783) .. controls (3,2.983) and (0.5,1.283) .. (-0.2,0.783) .. controls (-1.5,-0.717) and (-1,5.783) .. (0,6.283);
\draw[thimble]
  (0,18.848) .. controls (0.9,19.548) and (2.3,20.398) .. (4.45,20.588);
\draw[thimble]
  (0,0) .. controls (0.9,0.70) and (2.3,1.55) .. (4.45,1.74);
\draw[thimble]
  (0,-18.849) .. controls (0.9,-18.149) and (2.3,-17.299) .. (4.45,-17.109);
\draw[thimble]
  (0,-12.566) .. controls (0.9,-11.866) and (2.3,-11.016) .. (4.45,-10.826);
\draw[thimble]
  (0,12.566) .. controls (0.9,13.266) and (2.3,14.116) .. (4.45,14.306);
\draw[thimble]
  (0,-6.283) .. controls (0.9,-5.583) and (2.3,-4.733) .. (4.45,-4.543);
\draw[thimble]
  (0,6.283) .. controls (0.9,6.983) and (2.3,7.833) .. (4.45,8.023);

\draw[dual]
  (0,0) .. controls (1.1,-0.32) and (2.5,-1.15) .. (4.45,-1.20);
\draw[dual]
  (0,6.283) .. controls (1.1,5.963) and (2.5,5.133) .. (4.45,5.083);
\draw[dual]
  (0,12.566) .. controls (1.1,12.246) and (2.5,11.416) .. (4.45,11.366);
\draw[dual]
  (0,18.849) .. controls (1.1,18.529) and (2.5,17.699) .. (4.45,17.649);
\draw[dual]
  (0,-6.283) .. controls (1.1,-6.603) and (2.5,-7.433) .. (4.45,-7.483);
\draw[dual]
  (0,-12.566) .. controls (1.1,-12.886) and (2.5,-13.716) .. (4.45,-13.766);
\draw[dual]
  (-5.9,17.283) .. controls (-3,9.783) and (-0.8,6.583) .. (-1.1,4.483)
  .. controls (-1,4.383) and (-2.4,0.733) .. (0,0);
\draw[dual]
  (-5.9,23.566) .. controls (-3,16.066) and (-0.8,12.866) .. (-1.1,10.766)
  .. controls (-1,10.666) and (-2.4,7.016) .. (0,6.283);
\draw[dual]
  (-5.9,29.849) .. controls (-3,22.349) and (-0.8,19.149) .. (-1.1,17.049)
  .. controls (-1,16.949) and (-2.4,13.299) .. (0,12.566);
\draw[dual]
  (-5.9,36.132) .. controls (-3,28.632) and (-0.8,25.432) .. (-1.1,23.332)
  .. controls (-1,23.232) and (-2.4,19.582) .. (0,18.849);
\draw[dual]
  (-5.9,11.000) .. controls (-3,3.500) and (-0.8,0.300) .. (-1.1,-1.800)
  .. controls (-1,-1.900) and (-2.4,-5.550) .. (0,-6.283);
\draw[dual]
  (-5.9,4.717) .. controls (-3,-2.783) and (-0.8,-5.983) .. (-1.1,-8.083)
  .. controls (-1,-8.183) and (-2.4,-11.833) .. (0,-12.566);

\node[red!75!black] at (-4.9,-8.2) {$\mathcal{C}$};
\node[] at (-0.5,10) {$\J_{2}^>$};
\node[] at (-0.5,4) {$\J_{1}^>$};
\node[] at (-0.5,-2.5) {$\J_{0}^>$};
\node[] at (-0.5,-9) {$\J_{-1}^>$};

\node[green!60!black] at (-1.8,18) {$\Kk_{2}^>$};
\node[green!60!black] at (-4.2,18) {$\Kk_{1}^>$};
\node[green!60!black] at (-5.5,14.5) {$\Kk_{0}^>$};
\node[green!60!black] at (-5.5,8.5) {$\Kk_{-1}^>$};
\node[green!60!black] at (-5.5,2.2) {$\Kk_{-2}^>$};
\end{tikzpicture}
}
\end{minipage}

\caption{
Thimble pictures for the Gamma model at
\(
\theta=\frac{\pi}{2}\mp\delta
\).
Black curves are thimbles \(\J_n^\theta\), green curves are dual thimbles \(\Kk_n^\theta\), and red is the contour \(\mathcal{C}\) defining the Gamma function. For \(\theta=\frac{\pi}{2}-\delta\), one has \([\mathcal{C}]=[\J_0^<]\). For \(\theta=\frac{\pi}{2}+\delta\), \([\mathcal{C}]\) becomes \(\sum_{n\le0}[\J_n^>]\).
}
\label{fig:gamma-thimbles-with-Gamma}
\end{figure}

See Figure \ref{fig:gamma-thimbles-with-Gamma}.  The figure gives the corresponding change of thimble basis:
\begin{equation}\label{eq:gamma-infinite-thimble-relation}
[\J_n^{<}]
=
\sum_{m\le n}[\J_m^{>}],
\qquad n\in\Z.
\end{equation}
The inverse relation is the finite-difference formula
\begin{equation}\label{eq:gamma-finite-difference-relation}
[\J_n^{>}]
=
[\J_n^{<}]-[\J_{n-1}^{<}].
\end{equation}
Thus the same wall-crossing can be expressed either as the infinite triangular expansion
\eqref{eq:gamma-infinite-thimble-relation} or as the nearest-neighbor relation
\eqref{eq:gamma-finite-difference-relation}.

Following the general rule explained in Section~\ref{sec:general}, direct connecting
trajectories should be read from the elementary, or nearest-neighbor, part
of the Stokes matrix.  In the Gamma model, the finite-difference relation
\eqref{eq:gamma-finite-difference-relation} isolates precisely this
nearest-neighbor part: with the present ordering, it gives one elementary
entry coupling \(\J_{n-1}\) to \(\J_n\).  Hence the corresponding direct
Picard--Lefschetz trajectory is
\[
p_{n-1}\longrightarrow p_n,
\]
with signed coefficient \(1\).  This agrees with
Proposition~\ref{prop:gamma-neighboring-trajectories}.

The full triangular relation \eqref{eq:gamma-infinite-thimble-relation}
contains more information, but its non-nearest-neighbor entries should be
interpreted as broken-line contributions.  For example, on a finite block
\((\J_m,\J_{m+1},\ldots,\J_n)\), the full Stokes matrix has the form
\[
\begin{pmatrix}
1 & 1 & 1 & \cdots & 1\\
0 & 1 & 1 & \cdots & 1\\
0 & 0 & 1 & \cdots & 1\\
\vdots & \vdots & \vdots & \ddots & \vdots\\
0 & 0 & 0 & \cdots & 1
\end{pmatrix}.
\]
The nearest-neighbor entries of this matrix are the elementary coefficients, and hence
encode the direct trajectories
\[
p_r\longrightarrow p_{r+1}.
\]
The entries farther above the diagonal are obtained by multiplying these
elementary nearest-neighbor matrices.  Thus the entry relating \(\J_m\) to
\(\J_n\), with \(m<n\), records the broken chain
\[
p_m\longrightarrow p_{m+1}\longrightarrow\cdots\longrightarrow p_n,
\]
not a new direct trajectory from \(p_m\) to \(p_n\).  This is the
broken-line interpretation used in \cite[Appendix~C]{HarlowMaltzWitten2011}.

\smallskip
\subsection{Resurgence and the Stokes matrix}

We now compute the resurgent data for the Gamma model.  Since the Stokes phase under consideration is $\frac{\pi}{2}$, we rotate the action by this phase and set
\[
\mathcal A(\phi):=e^{-i\pi/2}S(\phi)=-iS(\phi).
\]
We study the integrals
\[
I_n^{<}(\hbar):=\int_{\J_n^{<}} e^{-\mathcal A(\phi)/\hbar}\,\dd\phi,
\qquad
I_n^{>}(\hbar):=\int_{\J_n^{>}} e^{-\mathcal A(\phi)/\hbar}\,\dd\phi,
\]
where now \(\hbar>0\) is real and small.  The critical points are still
\[
p_n=2\pi i n,
\qquad n\in\Z,
\]
and their rotated critical values are
\[
A_n:=\mathcal A(p_n)=-iS(p_n)=-i-2\pi n.
\]

Now we show that, as \(\hbar\to0^+\),
\begin{equation*}\label{eq:gamma-formal-saddle-expansion}
I_n^{\lessgtr}(\hbar)
\sim
\widetilde I_n(\hbar)
:=
e^{-A_n/\hbar}\,
e^{i\pi/4}\sqrt{2\pi\hbar}\,
\widetilde\phi_\Gamma(\hbar),
\qquad
\widetilde\phi_\Gamma(\hbar)
=
\widetilde\phi_{\mathrm{St}}(i\hbar).
\end{equation*}
Here
\begin{equation*}\label{eq:gamma-stirling-series-closed}
\widetilde\phi_{\mathrm{St}}(\hbar)
=
\exp\left(
\sum_{k\ge1}
\frac{B_{2k}}{2k(2k-1)}\hbar^{2k-1}
\right).
\end{equation*}

Indeed, for \(n=0\), the reduced asymptotic series is fixed by the classical Gamma-function asymptotics. The classical Stirling expansion, written directly in the small parameter \(\hbar\), is
\[
\Gamma(1/\hbar)
\sim
\sqrt{2\pi}\,
\hbar^{\frac12-\frac1\hbar}
e^{-1/\hbar}
\widetilde\phi_{\mathrm{St}}(\hbar).
\]
Therefore by \eqref{eq:gammaisgamma},
\[
I_{\R}(\hbar)
\sim
e^{-1/\hbar}\sqrt{2\pi\hbar}\,
\widetilde\phi_{\mathrm{St}}(\hbar).
\]
We have
\begin{equation*}
I_0^{\lessgtr}(\hbar)
\sim
\widetilde I_0(\hbar)
=
e^{-1/i\hbar}\sqrt{2\pi i\hbar}\,
\widetilde\phi_{\mathrm{St}}(i\hbar)
=
e^{-A_0/\hbar}\,
e^{i\pi/4}\sqrt{2\pi\hbar}\,
\widetilde\phi_\Gamma(\hbar).
\end{equation*}
For general \(n\), the saddle-point method and the vertical translation symmetry show that the series (except the exponential part) does not change from one saddle to another.  Indeed, the translation \(\phi\mapsto \phi+2\pi i n\) sends the saddle \(p_0\) to \(p_n\), and changes \(\mathcal A\) only by the constant \(-2\pi n\).  Hence it only changes the exponential factor \(e^{-A_n/\hbar}\), while the Gaussian phase and the reduced series remain the same:
\[
I_n^{\lessgtr}(\hbar)
\sim
\widetilde I_n(\hbar)
=
e^{-A_n/\hbar}\,
e^{i\pi/4}\sqrt{2\pi\hbar}\,
\widetilde\phi_\Gamma(\hbar).
\]
By the analysis in 
\cite[Theorem~1 and the paragraph between (3.7) and (3.8)]{Sauzin2021Gamma}, we have
\[
e^{i\pi/4}\sqrt{2\pi\hbar}\,
\widetilde\phi_\Gamma(\hbar)
\in
\widetilde{\mathcal R}^{\mathrm{int}}.
\]

Accordingly, the Gamma formal objects form a completed \(2\pi\mathbb Z\)-graded vector space
\[
\mathcal M_\Gamma
:=
\widehat{\bigoplus}_{n\in\mathbb Z}\C\,\widetilde I_n,
\qquad
\deg \widetilde I_n=A_n
\quad
(\text{equivalently, up to a common shift, } \deg \widetilde I_n=-2\pi n).
\]
Then we use the known alien calculus for the Stirling series in the book \cite{MitschiSauzin2016} (Formula (6.99)):
\[
\Delta^+_{2\pi i m}\widetilde\phi_{\mathrm{St}}(\hbar)
=
\begin{cases}
\widetilde\phi_{\mathrm{St}}(\hbar), & m=-1,1,2,3,\ldots,\\
0, & \text{otherwise}.
\end{cases}
\]
Since
\[
\widetilde\phi_\Gamma(\hbar)=\widetilde\phi_{\mathrm{St}}(i\hbar),
\]
the Borel singularities are rotated from \(2\pi i m\) to \(2\pi m\). Hence
\begin{equation}\label{eq:gamma-alien-reduced}
\Delta^+_{2\pi m}\widetilde\phi_\Gamma(\hbar)
=
\begin{cases}
\widetilde\phi_\Gamma(\hbar), & m=-1,1,2,3,\ldots,\\
0, & \text{otherwise}.
\end{cases}
\end{equation}

Now recall that
\[
\widetilde I_n(\hbar)
=
e^{-A_n/\hbar}
e^{i\pi/4}\sqrt{2\pi\hbar}\,
\widetilde\phi_\Gamma(\hbar),
\qquad
A_n=-i-2\pi n.
\]
Using \(A_{n-m}=A_n+2\pi m\), \eqref{eq:gamma-alien-reduced} gives
\begin{equation}\label{eq:gamma-dot-alien-full}
\dot\Delta^+_{2\pi m}\widetilde I_n(\hbar)
=
e^{-2\pi m/\hbar}\Delta^+_{2\pi m}\widetilde I_n(\hbar)
=
\begin{cases}
\widetilde I_{n-m}(\hbar), & m=-1,1,2,3,\ldots,\\
0, & \text{otherwise}.
\end{cases}
\end{equation}

The pointed alien operator \(\dot\Delta_{2\pi k}^+\) is
homogeneous of degree \(2\pi k\). For the wall-crossing at \(\theta=\frac{\pi}{2}\), the relevant Stokes ray in the rotated Borel plane is the positive real ray.  Therefore only the singularities \(2\pi m\) with \(m\ge1\) contribute. We obtain
\begin{equation*}\label{eq:gamma-stokes-auto-positive}
\mathfrak S^+\widetilde I_n
=
\widetilde I_n+\sum_{m\ge1}\widetilde I_{n-m}
=
\sum_{\ell\le n}\widetilde I_\ell.
\end{equation*}
Equivalently, if \(T\widetilde I_n=\widetilde I_{n-1}\), then
\[
\mathfrak S^+=1+T+T^2+\cdots .
\]
\begin{rmk}
Taking the logarithm of the Stokes automorphism gives
\[
\log\mathfrak S^+
=
-\log(1-T)
=
\sum_{k\ge1}\frac{T^k}{k}.\]
By \eqref{eq:dotDelta-log-Stokes},
$\log\mathfrak S^+=\sum_{k\ge1}\dot\Delta_{2\pi k}$, we have $\dot\Delta_{2\pi k}=\frac1kT^k$, and
\[
\dot\Delta_{2\pi k}\widetilde I_n
=
\frac1k\widetilde I_{n-k},
\qquad
\Delta_{2\pi k}\widetilde I_n
=
\frac1k\widetilde I_n.
\]
This is the computation of alien derivation $\Delta_{2\pi k}$ in Gamma model.
\end{rmk}

The inverse Stokes automorphism is
\begin{equation}\label{eq:gamma-stokes-inverse}
\mathfrak S^-=(\mathfrak S^+)^{-1}=1-T.
\end{equation}
Equivalently,
\begin{equation*}\label{eq:gamma-stokes-minus-formal}
\mathfrak S^-\widetilde I_n
=
\widetilde I_n-\widetilde I_{n-1}.
\end{equation*}
Thus the two Stokes automorphisms act on the formal objects by
\[
\mathfrak S^+\widetilde I_n
=
\sum_{m\le n}\widetilde I_m,
\qquad
\mathfrak S^-\widetilde I_n
=
\widetilde I_n-\widetilde I_{n-1}.
\]

By the general relation between lateral Borel sums and Stokes automorphisms,
\[
\mathcal S^{<}=\mathcal S^{>}\circ\mathfrak S^+,
\qquad
\mathcal S^{>}=\mathcal S^{<}\circ\mathfrak S^-,
\]
we obtain
\begin{equation*}\label{eq:gamma-integral-stokes-plus}
I_n^{<}
=
\sum_{m\le n}I_m^{>},
\end{equation*}
and conversely
\begin{equation*}\label{eq:gamma-integral-stokes-minus}
I_n^{>}
=
I_n^{<}-I_{n-1}^{<}.
\end{equation*}

Equivalently, if the formal series are ordered by decreasing index,
\[
\cdots,\widetilde I_2,\widetilde I_1,\widetilde I_0,\widetilde I_{-1},\cdots,
\]
then \(\mathfrak S^+\) is represented by the upper triangular matrix
\[
R_+
=
\begin{pmatrix}
\ddots &        &        &        &        \\
       & 1      & 1      & 1      & \cdots \\
       & 0      & 1      & 1      & \cdots \\
       & 0      & 0      & 1      & \cdots \\
       &        &        &        & \ddots
\end{pmatrix},
\]
while \(\mathfrak S^-\) is represented by the inverse upper bidiagonal matrix
\[
R_-
=
\begin{pmatrix}
\ddots &        &        &        &        \\
       & 1      & -1     & 0      & \cdots \\
       & 0      & 1      & -1     & \cdots \\
       & 0      & 0      & 1      & \cdots \\
       &        &        &        & \ddots
\end{pmatrix}.
\]
By Proposition~\ref{prop:thimble-lateral-sum}, the matrix relations obtained
from the resurgent Stokes automorphisms translate directly into the
wall-crossing relations for the Lefschetz thimbles.  Thus
\begin{equation*}\label{eq:gamma-thimble-stokes-plus}
[\J_n^{<}]
=
\sum_{m\le n}[\J_m^{>}],
\qquad \text{and} \qquad 
[\J_n^{>}]
=
[\J_n^{<}]-[\J_{n-1}^{<}].
\end{equation*}
These are precisely the two Picard--Lefschetz wall-crossing formulas
obtained above in \eqref{eq:gamma-infinite-thimble-relation} and
\eqref{eq:gamma-finite-difference-relation}.  Hence the resurgent Stokes
automorphisms recover the same jump formula.

\smallskip
\subsection{Hopf-algebraic interpretation}
\label{subsec:gamma-Hopf-interpretation}

We finally interpret the Gamma computation in the Hopf-algebraic language of
Subsection~\ref{subsec:alien-Hopf-plus} and explain Table \ref{table:gamma-hopf-dictionary}.  In the normalization used above, the
relevant Stokes ray is
\[
d=\mathbb R_{>0},
\]
and the corresponding semigroup is
\[
\Lambda_d=2\pi\mathbb N^\ast.
\]

By \eqref{eq:gamma-dot-alien-full}, the pointed alien operators act on the
formal objects in Gamma model by
\begin{equation}\label{eq:gamma-dotDelta-plus-shift-Hopf}
\dot\Delta_{2\pi k}^{+}\widetilde I_n
=
\widetilde I_{n-k},
\qquad k\ge1.
\end{equation}
Therefore the product in the Hopf algebra, which is operator composition,
becomes the composition of these shifts.  For \(a,b\in\mathbb N^\ast\),
\[
\dot\Delta_{2\pi a}^{+}\dot\Delta_{2\pi b}^{+}\widetilde I_n
=
\dot\Delta_{2\pi a}^{+}\widetilde I_{n-b}
=
\widetilde I_{n-a-b}
=
\dot\Delta_{2\pi(a+b)}^{+}\widetilde I_n.
\]
In particular,
\[
\dot\Delta_{4\pi}^{+}\widetilde I_n
=
(\dot\Delta_{2\pi}^{+})^2\widetilde I_n
=
\widetilde I_{n-2}.
\]
Thus the Stokes-matrix entry which shifts the sector by two steps is the
product of two nearest-neighbor shifts. On the Picard--Lefschetz
side, this is interpreted as the broken chain
\[
p_{n-2}\longrightarrow p_{n-1}\longrightarrow p_n,
\]
rather than as a new primitive trajectory from \(p_{n-2}\) to \(p_n\).

The coproduct records a different operation.  For example,
\[
\Delta_{\mathcal H}(\dot\Delta_{4\pi}^{+})
=
\dot\Delta_{4\pi}^{+}\otimes1
+
1\otimes\dot\Delta_{4\pi}^{+}
+
\dot\Delta_{2\pi}^{+}\otimes\dot\Delta_{2\pi}^{+}.
\]
This should be interpreted through product thimbles.  If two Gamma thimble
integrals are multiplied, the corresponding cycle is the product thimble
\(\J_a\times \J_b\).  The total action difference \(4\pi\) may then be split
between the two factors as
\[
(4\pi,0),\qquad (0,4\pi),\qquad (2\pi,2\pi).
\]
The three terms correspond respectively to a jump in the first factor, a
jump in the second factor, and simultaneous nearest-neighbor jumps in both
factors.  Thus the coproduct encodes the product rule for Stokes
automorphisms on product thimbles, rather than another broken chain in the
same one-dimensional Gamma system.

Finally, the antipode gives the inverse Stokes transformation already computed
in \eqref{eq:gamma-stokes-inverse}:
\[
\mathsf S_{\mathcal H}(\mathfrak S^+)
=
(\mathfrak S^+)^{-1}
=
\mathfrak S^-.
\]

\end{document}